\documentclass{article}
\usepackage[utf8]{inputenc}
\usepackage[margin=1.2in,footskip=0.25in]{geometry}
\usepackage{amsmath, amsthm, amscd, amsfonts, amssymb, graphicx, color}
\usepackage[bookmarksnumbered, colorlinks, plainpages]{hyperref}
\usepackage[table]{xcolor}
\usepackage{tikz}
\usepackage{float}
\usepackage{verbatim}
\usepackage{overpic}
\usepackage{subcaption}
\usepackage{mathtools,amsmath}
\usepackage[title]{appendix}
\usepackage{commath}
\usepackage{graphicx}
\usepackage{enumerate}
\usepackage{amssymb}
\usepackage{amsthm}
\usepackage{dsfont}
\usepackage{subcaption}
\usepackage{mdframed}
\usepackage{dirtytalk}
\usepackage{algorithm,algpseudocode}
\usepackage{footnotehyper}
\makesavenoteenv{algorithm}
\usepackage[normalem]{ulem}
\usepackage{tabularx} 
\usepackage[most]{tcolorbox}
\usepackage{authblk}
\usepackage{makecell}

\usepackage[backend=biber,style=numeric,sorting=none,maxbibnames=3,maxcitenames=2]{biblatex}
\usepackage{tikz}
\usetikzlibrary{backgrounds}

\tcbset{
  mygraybox/.style={
    colback=gray!10,       
    colframe=gray!60,      
    arc=3mm,               
    boxrule=0.8pt,         
    left=6pt, right=6pt,   
    top=6pt, bottom=6pt,   
    width=\linewidth,      
    sharp corners=uphill,  
  }
}
\usepackage[nameinlink,capitalize,noabbrev]{cleveref}
\hypersetup{
  pdfpagemode=UseNone,
  colorlinks=true,
  citecolor=blue,
  linkcolor=blue,
  urlcolor=blue
}
\usepackage{thmtools}
\usepackage{thm-restate}

\declaretheorem[name=Theorem]{thm}
\declaretheorem[name=Lemma,sibling=thm]{lemma}

\DeclarePairedDelimiter\ket{\lvert}{\rangle}

\DeclarePairedDelimiter\ceil{\lceil}{\rceil}

\usepackage{multirow}

\usepackage{siunitx}

\crefname{section}{Section}{Sections}
\crefname{subsection}{Subsection}{subsections}
\crefname{theorem}{Theorem}{Theorems}
\crefname{corollary}{Corollary}{Corollaries}
\crefname{lemma}{Lemma}{Lemmas}
\crefname{appendix}{Appendix}{Appendices}
\crefname{definition}{Definition}{Definitions}
\crefname{equation}{Eq.}{Eqs.}
\crefname{algorithm}{Algorithm}{Algorithms}
\crefname{table}{Table}{Tables}

\title{Compiling the 2D Fermi-Hubbard ground-state energy estimation algorithm for active volume quantum architectures}

\author{Harriet Apel\thanks{Corresponding author: hapel@psiquantum.com; these authors contributed equally.}}
\author{Athena Caesura}
\author{Carys Harvey}
\author{Sam Heavey}
\author{Angus Kan}
\author{Jessica Lemieux}
\author{Ryan Levy}
\author{Sam Pallister}
\author{Joseph Peetz}
\author{William Pol}
\author{Sukin Sim\thanks{Corresponding author: ssim@psiquantum.com; these authors contributed equally.}}
\author{William A. Simon}
\author{Mark Steudtner}
\author{Gideon Uchehara}
\affil{PsiQuantum, 700 Hansen Way, Palo Alto, CA 94304, USA}
\date{\today}

\begin{document}

\maketitle

\begin{abstract}
As quantum computing enters the early fault-tolerant era, circuit compilation choices will increasingly depend on details of the underlying architecture rather than solely optimizing for generic proxies such as non-Clifford count.
We present an active-volume-aware compilation of the ground-state energy estimation algorithm for the two-dimensional Fermi--Hubbard model using quantum phase estimation and Trotterized time evolution.
The proposed compilation reduces the active volume across $L\times L$ square lattices with $L=4$ to $20$, achieving up to a $3.9\times$ reduction over prior work optimized for non-Clifford cost.
As a by-product of these compilation improvements, the resulting circuits also achieve state-of-the-art Toffoli counts, with a $\sim 2\times$ reduction for the $L=20$ case. 
Lastly, the active volume architecture and recent execution scheduling advances provide a means of translating these reduction trends into runtime.
This demonstrates the increasing importance of architecture-aware compilation for practical early fault-tolerant quantum computing.
\end{abstract}

\setcounter{tocdepth}{2} 
\tableofcontents

\section{Introduction}

Quantum computing has the potential to transform applications ranging from simulation of many-body systems and chemistry to cryptography and optimization.
While many of these applications ultimately require large scale fault-tolerant quantum computers, recent advances in quantum error correction and quantum processors are starting to bring the first generation of fault-tolerant quantum computers within reach~\cite{psi-arch-paper,google-below-threshold, paetznick2024demonstrationlogicalqubitsrepeated, Bulvstein2024}.
Together with continued algorithmic progress, these developments motivate the study of \emph{early generation fault-tolerant quantum computers}: devices in which both logical memory and computational time are severely constrained, making every logical resource a scarce commodity.
This regime lies between noisy intermediate-scale quantum (NISQ) devices and large scale fault-tolerance quantum computers.
Developing practical applications for first fault-tolerant quantum computers (FTQCs) requires reducing the quantum resource requirements of algorithms by balancing asymptotic scaling, constant factors, and architecture-aware design.

Simulation of fermionic lattice models is one of the flagship applications of quantum computing.
It has motivated experimental demonstrations on current quantum hardware~\cite{alam2025fermionicdynamicstrappedionquantum,alam2025programmabledigitalquantumsimulation,arute2020observationseparateddynamicscharge,Stanisic_2022,Chowdhury_2026, hartnett2026fastaccuratehighresolutionsimulation}, alongside extensive algorithmic development for fault-tolerant quantum computers~\cite{campbell2022early,kan2024resource-optimized,babbush2018encoding,kivlichan2020improved}.
Lattice Hamiltonians are particularly well-suited to the early fault-tolerant regime as despite being specified by relatively few local interactions, they capture rich many-body physics which can be systematically extended towards realistic models of materials and rapidly become intractable for classical simulation. 
Among these, the two-dimensional Fermi-Hubbard model~\cite{Hubbard1963} has emerged as a canonical benchmark for quantum resource estimation, routinely used to evaluate the capabilities of proposed fault-tolerant quantum architectures on a practically relevant application~\cite{webster2026pinnaclearchitecturereducingcost, ismail2026fast, chung2026partiallyfaulttolerantquantumcomputation, bourdoncle2026, khan2026architectingearlyfaulttolerant, hartnett2026fastaccuratehighresolutionsimulation,PRXQuantum.3.010329}.

Fault-tolerant algorithm development has traditionally focused on minimizing the number of non-Clifford operations, using $T$ gates as a proxy for execution cost.
However, as early fault-tolerant architectures evolve, this proxy is becoming increasingly disconnected from the true cost of computation,
motivating architecture-aware resource metrics~\cite{litinski2022activevolume, litinski2025blocklet, FLASQ, tqec, bourdoncle2026}.
One such metric is \emph{active volume}, modeling the execution cost on fault-tolerant architectures that support limited non-local logical connectivity, such as photonic platforms and other architectures with flexible long-range interactions~\cite{litinski2022activevolume}.
Active volume compilation has already yielded substantial reductions in the resource requirements of cryptographic and quantum chemistry applications~\cite{litinski2022activevolume,litinski2023compute256bitellipticcurve,ECC_activevolume,caesura2025faster}. This illustrates architecture-aware compilation translating the increased flexibility of non-local connections into practical reductions in application cost.
However, no active-volume-aware compilation of an FTQC lattice Hamiltonian algorithm has yet been reported.
In this work, we address this gap by compiling the ground-state energy estimation algorithm for the two-dimensional Fermi-Hubbard model with active volume as the primary optimization objective.

We apply these ideas to the problem of estimating the ground-state energy of the two-dimensional Fermi--Hubbard model using quantum phase estimation with Trotterized time evolution focusing on the $20\times20$ lattice size.
We first demonstrate a $12.1\times$ reduction in cost enabled by the active volume architecture compared with conventional circuit volume estimates, highlighting the impact of architectural design on practical resource requirements.
The principal contribution of this work is then to demonstrate a further $3.9\times$ reduction through architecture-aware compilation.
While the compilation is optimized for active volume, the substantial contribution of Toffoli gates to this metric means that the resulting circuits also establish the lowest reported Toffoli counts for this problem.
Finally, we combine these optimized circuits with the recently proposed active volume scheduler \cite{heavey2026scheduler} to demonstrate the reduction in runtime.
For the $20\times20$ Fermi-Hubbard ground-state energy algorithm, we show how increasing the workspace qubit count can significantly reduce the runtime up until the reaction limit.
Together, these results highlight the importance of co-design across the quantum computing stack, showing how algorithms, architectures, compilation, and execution strategies combine to make lattice Hamiltonian simulation increasingly practical on early FTQCs. 

The remainder of this paper is organized as follows.
\cref{sect:background} reviews the Fermi--Hubbard ground-state energy estimation algorithm and introduces the active volume architecture and resource metric.
We then present our quantum resource estimates in \cref{sect:results}, emphasizing the impact of the compilation techniques before describing them in detail. \Cref{sect:compilation} subsequently presents the compilation methods underlying these improvements.
Finally, \cref{sect:conclusions} discusses future directions and concludes.
We also highlight \cref{appen:av-subroutine} which includes a worked example of how to derive the active volume of a quantum subroutine, we hope serves as a practical tutorial for future active volume resource estimation.

\section{Background}\label{sect:background}

\subsection{Quantum simulation of the Fermi-Hubbard model}

\paragraph{The Fermi-Hubbard model and applications} The single band Fermi-Hubbard model on a 2D lattice of size $L_x \times L_y$ is given by 
\begin{equation}
    H_\text{FH} = -t\sum_{\langle i,j\rangle \sigma} (c^\dagger_{i\sigma} c_{j\sigma} + c^\dagger_{j\sigma}c_{i\sigma}) + u \sum_i (n_{i\uparrow}-\mathds{1}/2)(n_{i\downarrow}-\mathds{1}/2)
    \label{eq:fhm}
\end{equation} 
where $c_{i\sigma}$ ($c^\dagger_{i\sigma}$) denotes a fermionic annihilation (creation) operator of spin $\sigma$ at site $i$, $n_{i \sigma}$ is the fermionic number operator, $u$ quantifies the interaction strength, $t$ denotes the hopping amplitude and $\langle i,j\rangle$ denotes nearest neighbors. 
For our simulation we take $t=1$, $u=8$. We denote the first term in \cref{eq:fhm} as the kinetic or hopping term, and the second term as the interaction term. 
In \cref{eq:fhm} the bare interaction term $u\sum_i n_{i\uparrow}n_{i\downarrow}$ is shifted by $-\frac{u}{2}\sum_i(n_{i\uparrow}+n_{i\downarrow}) +\frac{u}{4} \mathds{1}$, which commutes with the Hamiltonian and shifts the energy spectrum by a known constant $(-uN/4+u/4)$ where $N$ is the total number of electrons.
This shift, studied in \cite{campbell2022early}, facilitates a cheaper quantum implementation of the Hamiltonian.
We assume periodic boundary conditions (PBC) in both directions.
The model's ground state remains difficult to 
numerically simulate with classical methods 
due to the competition between a stripe ordered ground state and superconducting correlations \cite{Qin2020,sorella2023systematically, roth2025superconductivity}, with a lattice size of $20\times 20$ proposed as an important scale to resolve finite-size effects \cite{Agrawal2024}.
This model can serve as a stepping stone towards more generalized lattice models with richer physics such as 
the cuprate, pnictide, or extended models 
\cite{kan2024resource-optimized, baysmidt2025faulttolerant}.

\paragraph{Extensive error} The ground-state energy of the Fermi–Hubbard model is extensive in the number of lattice sites, whereas the energy per site is finite in the thermodynamic limit. This motivates choosing an extensive error budget for the total energy of $0.0051tL_xL_y$ so that the error in the energy density remains fixed with system size~\cite{kan2024resource-optimized,campbell2022early,kivlichan2020improved}.
This error is consistent with the error observed in classical simulations for an exactly solvable system size at what is considered the toughest regime ($u/t=8$ and electron filling per site $n=0.875$) \cite{leblanc2015solutions}. 

\paragraph{High-level algorithm} 
We employ quantum phase estimation (QPE) to estimate the ground-state energy through the eigenphase of the time evolution operator generated by the lattice Hamiltonian.
To achieve state-of-the-art scaling with respect to accuracy, the operator should be constructed by linear combination of unitaries (LCU) techniques \cite{lcu2012}, such as qubitization \cite{Low2019hamiltonian, OptimalQSPHS}.
However, the extensive error budget makes the required accuracy progressively looser with increasing lattice size.
For the larger lattices considered in this work, beginning beyond the exactly tractable regime (e.g. square lattice with $L\gtrsim 6$), the asymptotic advantage in simulation scaling is therefore not the dominant consideration.
Therefore, consistent with prior work, we adopt second-order Suzuki--Trotter for the time evolution, as previous analyses have found it to be the most competitive approach in the regime considered here~\cite{campbell2022early,kan2024resource-optimized}.
Although the gate cost per Trotter step scales extensively as $O(L^2)$, the extensive error budget requires fewer Trotter steps with lattice size, almost exactly compensating for this scaling as noted in~\cite{campbell2022early}.
As a result, the overall gate cost remains approximately constant across the lattice sizes considered here.
This contrasts with LCU-based approaches, where the effective one-norm scaling results in an increasing cost with system size -- see \cite[Fig. 2]{kan2024resource-optimized} for an example of this in Toffoli resource counts.

\paragraph{Plaquette Suzuki-Trotter scheme}
We approximate the time-evolution operator $U(\tau) = \mathrm{e}^{\mathrm{i}H\tau}$ using a second-order Suzuki–Trotter decomposition~\cite{suzuki1990fractal, suzuki1991general}:
\begin{align}\label{eqn:trotter-intro}
    U(\tau) \approx \left(\tilde{U}_{\tau/r}\right)^r := \left(\prod_{k=1}^P \mathrm{e}^{\mathrm{i} H_k \tau / 2r}\prod_{k=P}^1 \mathrm{e}^{\mathrm{i} H_k \tau / 2r}\right)^r
\end{align}
where $r$ denotes the number of Trotter steps and the generating Hamiltonian is decomposed into a sum of terms $H = \sum_{k=1}^P H_k$. 

Reference \cite{campbell2022early} identifies a splitting of the Hamiltonian that is particularly efficient from both an implementation standpoint and a commutation bound: the onsite interaction term $H_I$ and two disjoint portions of the hopping term $H_h = H_h^p + H_h^g$; see \cref{fig:fh_plaq} for an illustration.
\begin{equation}
    H_\textup{FH} = H_I + H_h^p + H_h^g
\end{equation}
where
    \begin{align}
    H_I &= u\sum_{i} (n_{i\uparrow}-1/2)(n_{i\downarrow}-1/2) \label{eqn:int}\\
    H_h^p & =-t \sum_{\langle i j\rangle\in \text{pink}, \sigma }( c_{i\sigma}^\dagger c_{j\sigma} + c_{j\sigma}^\dagger c_{i\sigma}) \label{eqn:pink}\\
    H_h^g & = -t \sum_{\langle i j\rangle \in \text{gold}, \sigma}( c_{i\sigma}^\dagger c_{j\sigma} + c_{j\sigma}^\dagger c_{i\sigma}).\label{eqn:gold}
\end{align}

\begin{figure}[h!]
    \centering
    \includegraphics[width=0.35\linewidth]{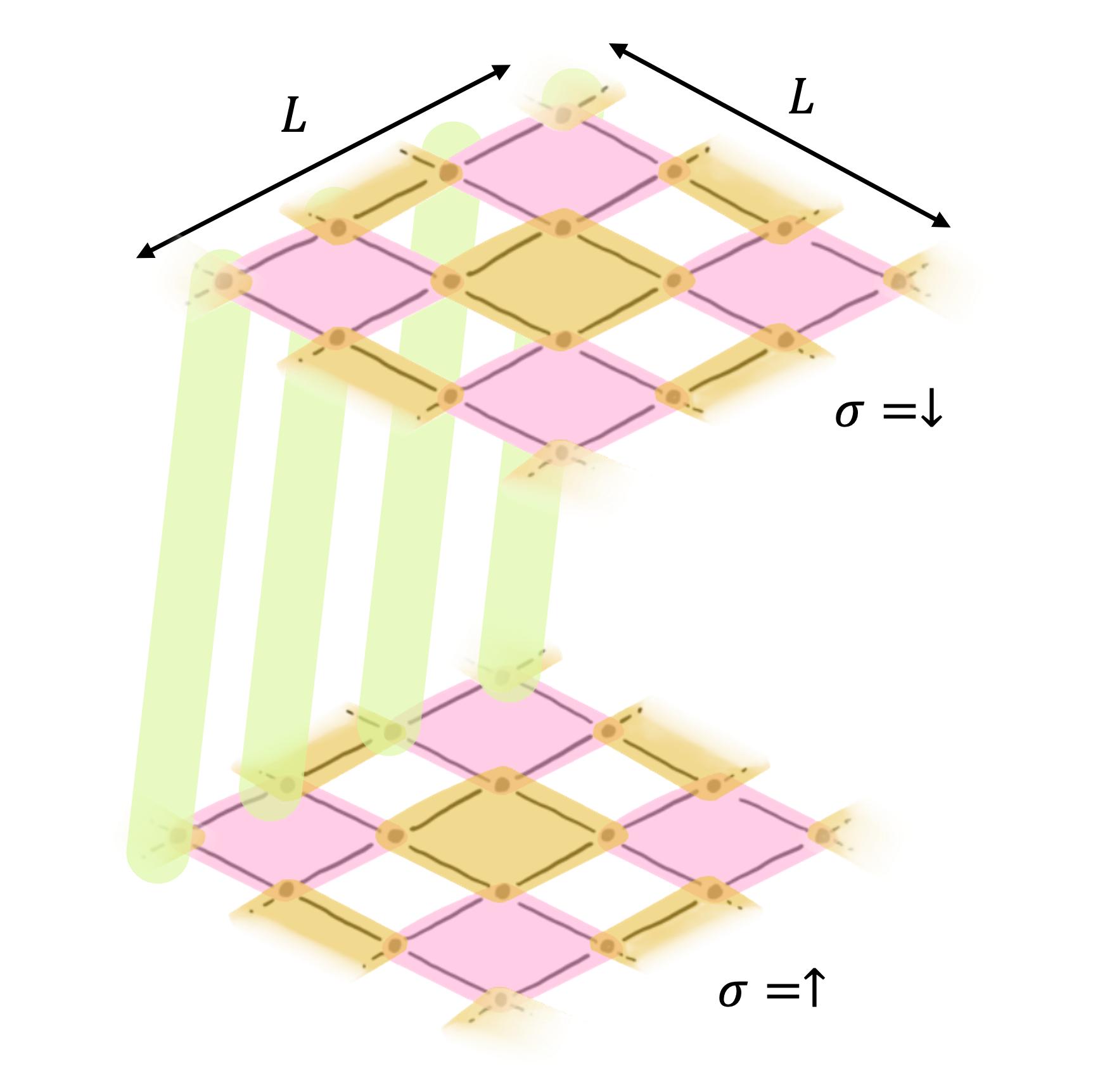}
    \caption{ The Fermi-Hubbard lattice Hamiltonian. The model Hamiltonian considered is given in \cref{eq:fhm} where $i$ indexes lattice sites on a $L \times L$ square lattice with two spin sectors $\sigma\in\{\uparrow, \downarrow \}$. $\langle i, j\rangle$ denotes nearest neighbors on the lattice drawn as edges. For the product formula the Hamiltonian is split into three components. The interaction term \cref{eqn:int} is shown in green and is between fermions on the same site across the two spin sectors. The kinetic term is split into two disjoint sets of plaquettes shown in pink \cref{eqn:pink} and gold \cref{eqn:gold} where the shaded plaquette indicates a grouping of four edges. 
}
    \label{fig:fh_plaq}
\end{figure}

\subsection{Compiling to an active volume architecture}\label{sect:avprimer}

Determining the cost of a quantum computation ultimately requires specifying a target architecture.
However, as leading platforms continue to evolve, it is important to develop compilation strategies guided by the relevant resource metrics at our disposal.
A useful resource metric proxy should reflect the dominant costs across leading architectures while remaining tractable to evaluate for large-scale algorithms.
This introduces an inherent trade-off between architectural specificity and accessibility in resource estimation.

Most early work in fault-tolerant quantum computing assumes surface-code-based encodings \cite{Litinski_2019, Fowler2018LowOQ, KITAEV20032, Bravyi:1998sy, PhysRevA.86.032324}. 
Within this framework, non-Clifford $T$ gates require the distillation of magic states which dominate the cost of many algorithms. 
To simplify resource estimation, prior literature has often focused on metrics such as $T$-count and logical qubit count, typically neglecting the cost of Clifford operations. 
One volumetric cost based on this is the \emph{circuit volume} $=$ $T$ count $\times$ logical memory qubits.
We note that circuit volume is a representative metric for Pauli-based computation \cite{Litinski_2019}. For architectures that improve upon the Pauli-based computation framework, e.g. parallelizing Cliffords, non-Cliffords, and/or magic state factories, circuit volume becomes a less representative cost proxy.

An alternative metric that includes both Clifford and non-Clifford cost is the computational or spacetime volume, which can be decomposed into two components:
\begin{enumerate}
    \item[] \emph{Active volume} — the computational volume associated with implementing logical operations.
    \item[] \emph{Idle volume} — the computational volume incurred when qubits remain idle while other operations are performed.
\end{enumerate}

In this work, we will use the active volume of the circuit as the resource metric assuming an active volume architecture and make compilation choices based on this.
In some (but not all) cases we will see that this aligns with other proxies such as 
$T$ count but motivates a more Clifford-aware perspective.
In the next sections we will motivate why this is an appropriate metric for certain early generation fault tolerant architectures and recap how it is calculated. 

\paragraph{Active volume architectures} 
A common quantum architecture consists of static physical qubits arranged in a 2D array with physical operations restricted to nearest neighbors \cite{Litinski_2019}. 
Encoding these physical qubits into surface code patches results in idling logical qubits requiring continuous stabilizer measurements and therefore incurs a cost comparable in scaling to performing logical operations. 
Moreover, implementing non-local logical operations requires large ancillary patches to bridge non-neighboring qubits. 
As a result, most of the computational volume in a generic circuit is from idle volume, making it a natural target for resource reduction.
The active volume architecture proposed in \cite{litinski2022activevolume} introduces non-local connectivity to largely eliminate this overhead.
It instead distinguishes ``data'' qubits used to store the computational state from ``workspace'' qubits that provide the ancillary surface-code patches needed to perform lattice-surgery operations.
By reusing this workspace as the computation proceeds, the volume associated with idle storage and routing is significantly reduced. 

The active volume architecture does not require all-to-all connectivity between $N$ logical qubits; rather, each qubit connects to only $O(\log N)$ others~\cite{litinski2022activevolume}.
In the surface-code picture, this corresponds to allowing transversal two-qubit operations only between specific pairs of surface-code patches. More broadly, active volume is compatible with a wide range of fault-tolerant schemes, including LDPC codes and concatenated codes.
Although initially motivated by silicon photonics \cite{litinski2022activevolume}, these limited long-range connections are key to efficiently route qubits and facilitate parallel execution of logical operations. 
As such, active volume provides a natural metric for platforms with flexible or non-local interactions, including photonic architectures (both fusion-based and GKP-encoded) and neutral-atom systems \cite{Webber_2022}.
It is also relevant for emerging superconducting approaches incorporating long-range couplers \cite{Yoder2025}.
In contrast, for predominantly local architectures, active volume remains a useful metric, particularly in capturing Clifford costs, but may lead to different optimization trade-offs compared to locality-dominated spacetime metrics (e.g. differing conclusions for Hamming weight phasing in \cref{subsec:hwp} compared to \cite{FLASQ}).

\paragraph{Active volume as a metric}

Each logical operation has an associated
active volume which is measured in units of logical blocks.
The total active volume of the circuit can be calculated by summing up the active
volumes of all constituent operations.
The number of logical blocks required to implement a gate will vary from operation to operation.
In this work, unless stated otherwise, the number of logical blocks associated with basic operators is taken from \cite[Table 1]{litinski2022activevolume}.
For example, a $\mathrm{CNOT}$ will require 4 logical blocks whereas a $\mathrm{Toffoli}$ requires 47 logical blocks (assuming conservative assumptions regarding the cost of CCZ distillation).
Though meaningful, for many applications this gap is not large enough to treat Clifford operations as ``free,'' demonstrating the need for Clifford-aware compilation.
This benefit is subroutine dependent, for example arithmetic circuits are especially active volume efficient, with adders exhibiting a significantly lower per-Toffoli active volume cost than data loading \cite{litinski2022activevolume}.
We will see examples of this in \cref{appen:av-subroutine}.
Therefore, even $\mathrm{T}$-count-preserving compilation choices can impact the active volume due to differing Clifford overheads.
 
Whereas the computational volume scales with the number of qubits, active volume does not since in practice a large proportion of the former is idle.
Key quantum algorithms have seen significant reductions in resources simply by leveraging active volume compilation: $44\times$ improvement in spacetime volume in RSA factorization \cite{litinski2022activevolume}; $300-700\times$ in elliptic curve
cryptography \cite{litinski2023compute256bitellipticcurve}, $2-20\times $ in binary curve
cryptography \cite{ECC_activevolume}, as well as $25\times$ in chemistry problems \cite{caesura2025faster}.
Materials science is one of the most promising applications of early fault tolerance; this work provides the first resource estimates for lattice Hamiltonians for an active volume architecture.

\section{Results}\label{sect:results}

We first present the impact of the architecture-aware compilation techniques proposed in this work.
We compare conventional circuit-volume estimates with active-volume estimates before quantifying the reductions achieved through architecture-aware compilation.
Throughout this section, we consider the total cost of all QPE queries required to estimate the ground-state energy to an extensive error of $0.0051tL_xL_y$.
Unlike a number of previous resource estimates that neglect the controlled implementation of the time-evolution circuits, our estimates explicitly include the control structure required for QPE.
To improve the reproducibility of our resource estimates, all circuits are implemented in the open-access PsiQuantum software package~\cite{psiquantum_qdk}, which compiles the circuits, reports gate counts, and computes active volume using an implementation of the lookup tables of Ref.~\cite{litinski2022activevolume}, extended as described in~\cref{sect:compilation}.

As a baseline, we consider the state-of-the-art circuit from Ref.~\cite{kan2024resource-optimized}\footnote{We restrict the number of QPE queries to $6\times2^l$ for integer $l$. Although Ref.~\cite{kan2024resource-optimized} reports query counts not of this form, the associated directionally controlled implementation is not specified. In addition, the numerical Holevo-variance results that we use to determine the required number of queries are reported only for the power-of-two schedules considered in Ref.~\cite{higgins2007entanglement}. Since we apply the same restriction to both the baseline and improved circuits it does not bias their comparison.}.
This baseline reflects compilation choices optimized to minimize T-count, providing a direct comparison against active-volume-aware compilation decisions made in this work.
While such analytical expressions are tractable for gate counts, deriving corresponding active volume expressions quickly becomes unwieldy, motivating software-assisted resource estimation used throughout this work. ~\cref{appen:qre_data} provides additional data for reproducibility, including a detailed specification of the baseline circuit and verification of our implementation against analytically computed Toffoli counts.

\begin{figure}[H]
\centering
\includegraphics[width=0.85\textwidth]{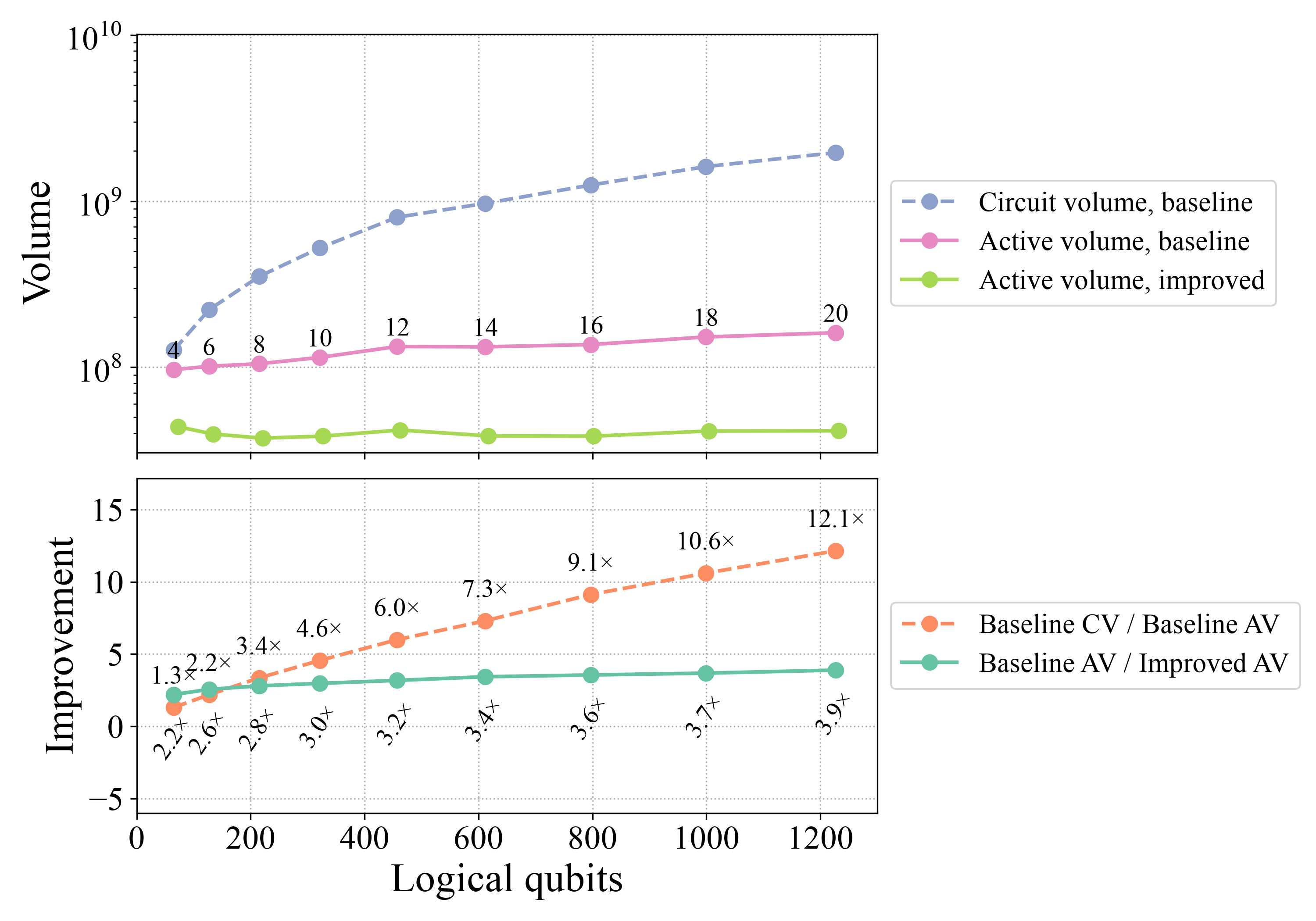}
\caption{
Comparison of computational volume of 2D square Fermi-Hubbard ground-state energy estimation for the baseline and improved implementations against the logical qubit count. The logical qubit count includes both the system qubits and the auxiliary qubits required to implement a single batch of Hamming-weight phasing. The numbers above the baseline active volume data points indicate the lattice size.
\textbf{Top}: Absolute circuit volume (as defined in Ref.~\cite{litinski2022activevolume}) and active volume resource estimates for the total cost of QPE algorithm achieving a $0.0051t L^2$ error.
\textbf{Bottom}: Relative speedups obtained by first replacing circuit volume with active volume as the resource metric (orange), and then by applying the compilation techniques proposed in this work (green). Circuit volume (CV) exceeds the baseline active volume (AV) by up to $12.1\times$ for the $20 \times 20$ lattice. The proposed compilation provides a further reduction in active volume of up to $3.9\times$. 
}
\label{fig:comparing_volumes}
\end{figure}

We first compare conventional circuit volume and active volume estimates in~\cref{fig:comparing_volumes}, to assess how the choice of resource metric affects the estimated cost of the algorithm.
As the logical qubit requirements of the algorithm increase, circuit volume increasingly overestimates the execution cost on an active volume architecture, reaching a factor of $12.1\times$ for the $20\times20$ lattice.
Consistent with previous active volume studies in quantum chemistry and cryptography, these results confirm circuit volume is a poor proxy for the execution cost on architectures with some non-local connectivity.
We therefore compare the baseline and improved circuits in terms of active volume, as shown by the remaining curves in~\cref{fig:comparing_volumes}.
We observe up to a $3.9\times$ reduction for the $20 \times 20$ lattice, demonstrating the further benefit of our compilation choices. 
This improvement is achieved through reductions in both Clifford and non-Clifford resource requirements, highlighting the benefits of optimizing directly for spacetime cost rather than a single gate-count metric.

\begin{figure}[H]
\centering
\includegraphics[width=0.95\textwidth]{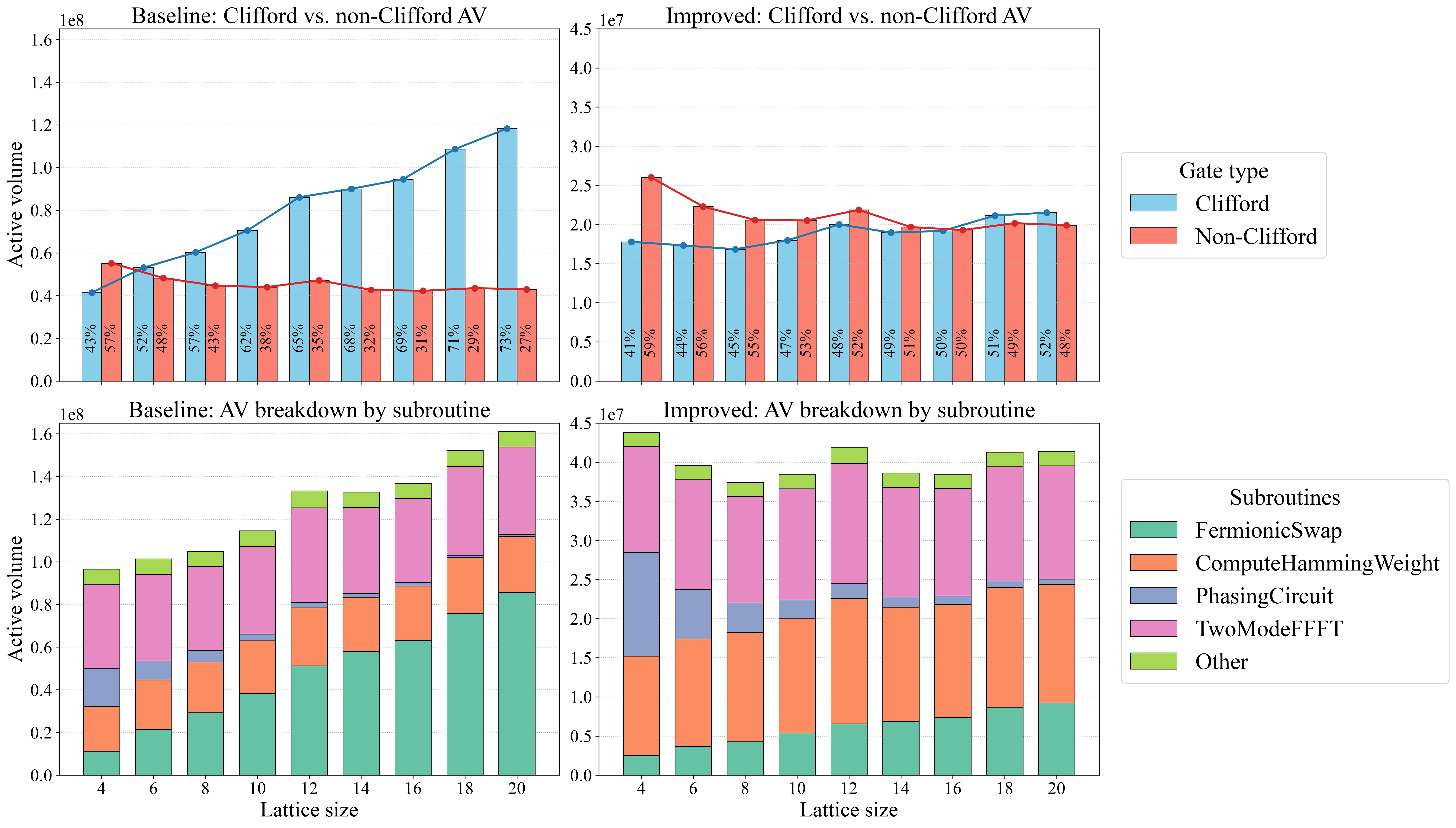}
\caption{
Active volume (AV) breakdown. 
\textbf{Top row:} Breakdown of the active volume into Clifford and non-Clifford contributions. The proposed compilation techniques reduce both components, resulting in a substantially more balanced distribution of execution cost.
Note that plots on the left column are active volume breakdowns of the baseline circuits while the plots of the right column are those of the improved circuits.
\textbf{Bottom row:} Active volume breakdown of the baseline and improved circuits by major subroutine. The dominant contributions are grouped according to fermionic-mode routing (``FermionicSwap''), implementing towers of rotations (``ComputeHammingWeight'' and ``PhasingCircuit''), and hopping-term basis changes (``TwoModeFFFT''). Remaining operations are grouped into ``Other.'' The reductions are distributed across the subroutines, with the largest improvement occurring in the fermionic swap subroutine. The breakdown is obtained using data generated by~\cite{psiquantum_qdk} for callgraphs (see \cite{caesura2025faster} for an example).
}
\label{fig:av_breakdowns}
\end{figure}

To understand the origin of the reduction, the top row of~\cref{fig:av_breakdowns} separates the total active volume into its Clifford and non-Clifford contributions.
Unlike gate-count metrics, active volume places Clifford and non-Clifford operations on a common footing by accounting for their relative execution cost (e.g., a Toffoli gate occupies 47 logical blocks compared with 4 for a CNOT).
This enables a direct comparison of their contributions to the quantum resources.
The proposed compilation techniques reduce both components.
The largest absolute reduction is obtained in the Clifford contribution, reflecting that the baseline circuit was primarily optimized to minimize T-count and therefore left greater opportunity to reduce the Clifford overhead.
Nevertheless, the non-Clifford active volume is also reduced through the combined effect of reducing the number of QPE queries, improved Trotter ordering, and more efficient error budgeting.
Ref. \cite{campbell2022early} observed that under a fixed extensive error target the Toffoli count becomes approximately independent of lattice size.
As the system grows, the increasing gate count per Trotter step is nearly offset by the reduction in the number of Trotter steps permitted by the larger error budget.
The baseline active volume does not exhibit this favorable scaling because the increasing Clifford overhead dominates the active volume.
After compilation improvement, however, this behavior is largely recovered, with the total active volume remaining nearly constant over the lattice sizes considered at around $4\times 10^7$ logical blocks.
For the $20\times 20$ lattice size the contributions to active volume from Clifford and non-Clifford operations approaches an approximately 1:1 split, indicating that further reductions in execution cost will require continued optimization of both the Clifford and non-Clifford components of the algorithm.

The bottom row of \cref{fig:av_breakdowns} attributes the active volume to the major subroutines.
Those reductions are distributed throughout the algorithm rather than arising from a single optimization.
The largest reduction is observed in the fermionic swap network, consistent with it being the dominant source of active volume in the baseline implementation.
However, appreciable reductions are also obtained in other subroutines, demonstrating that the improvements described in \cref{sect:compilation} contribute throughout the circuit.
For the baseline implementation, the active volume is increasingly dominated by the fermionic swap network, which ultimately determines the scaling of the execution cost.
In contrast, the optimized circuit exhibits a much more balanced distribution of active volume across the subroutines, with no single component overwhelmingly dominating the overall cost.

Lastly, we incorporate a common space-time trade-off in which operations are executed in batches so that ancillary qubits can be reused, reducing the peak logical-qubit requirement at the cost of additional circuit depth.
Although the majority of logical qubits are set by the lattice size, our realization of parallel equiangular rotations via Hamming weight phasing introduces an ancilla overhead that can be traded against depth.
We follow~\cite{kivlichan2020improved,campbell2022early,kan2024resource-optimized} in exploiting this by partitioning the parallel rotations into batches and processing them sequentially.
The number of batches is then a tunable compilation parameter.
Four batches reduces the logical qubit requirement for the $L=20$ square lattice by nearly $25\%$ while increasing the active volume by only around $3\%$, making batching an attractive option for qubit-constrained FTQC.
While active volume provides the primary optimization objective in this work, determining the best point on a space-time trade-off requires additional architectural information, as the relative value of qubits and execution time depends on the target machine and its logical-qubit capacity.
Determining the preferred batching strategy requires moving one stage further down the compilation stack to logical cycle scheduling~\cite{heavey2026scheduler}.
We return to this trade-off in detail later (\cref{subsec:batched_hwp}); first, we use the same scheduling framework to assess comparative runtimes of the compiled circuits.

\subsection{Trends in runtime estimates}

Although active volume is a useful resource proxy for optimizing quantum algorithms, wall-clock runtime is a more directly interpretable and widely used performance metric.
However, converting logical resources into runtimes requires explicit assumptions about the underlying hardware.
We therefore report runtime estimates for an illustrative photonic implementation of the active volume architecture~\cite{litinski2022activevolume}.
These estimates are intended to \emph{compare} compilation methods under a common hardware model, rather than to predict the performance of a particular device. 
See~\cref{sec:runtime_data} for more details on the example hardware model considered in this work that is based on Ref.~\cite{litinski2022activevolume}.

To quantify the runtime improvement obtained from the compilation methods introduced in this work, \cref{fig:runtimes} compares our results with the runtime results from the baseline circuit.
For each lattice size, we set the machine's total logical qubit capacity to $1.5$ times the circuit's peak data-qubit requirement. This choice prevents the number of workspace qubits from overwhelming the number of memory qubits, thereby avoiding the reaction-limited regime.\footnote{This refers to the regime in which sequential measurement-dependent basis updates cannot be completed within the time available for a logical cycle, so that reaction depth, rather than available workspace, limits the achievable runtime~\cite{heavey2026scheduler}.}
Taking the $L=8$ lattice as example, the improved circuit budget is 220 logical circuit qubits, compared with 214 for the baseline circuit.
With our qubit allocation, this corresponds to 110 and 107 workspace qubits, respectively, representing a $2.8\%$ increase for the improved circuit.
Despite this small difference, the improved circuit is $73\times$ faster, with a runtime of \(4.57\,\mathrm{s}\) rather than \(332\,\mathrm{s}\).
The qubit counts are closely matched and thus the vast majority of the speedup is due to compilation methods outlined in this paper.

\begin{figure}[H]
    \centering
    \begin{subfigure}[b]{0.47\textwidth}
        \centering
    \includegraphics[width=0.95\linewidth,
    ]{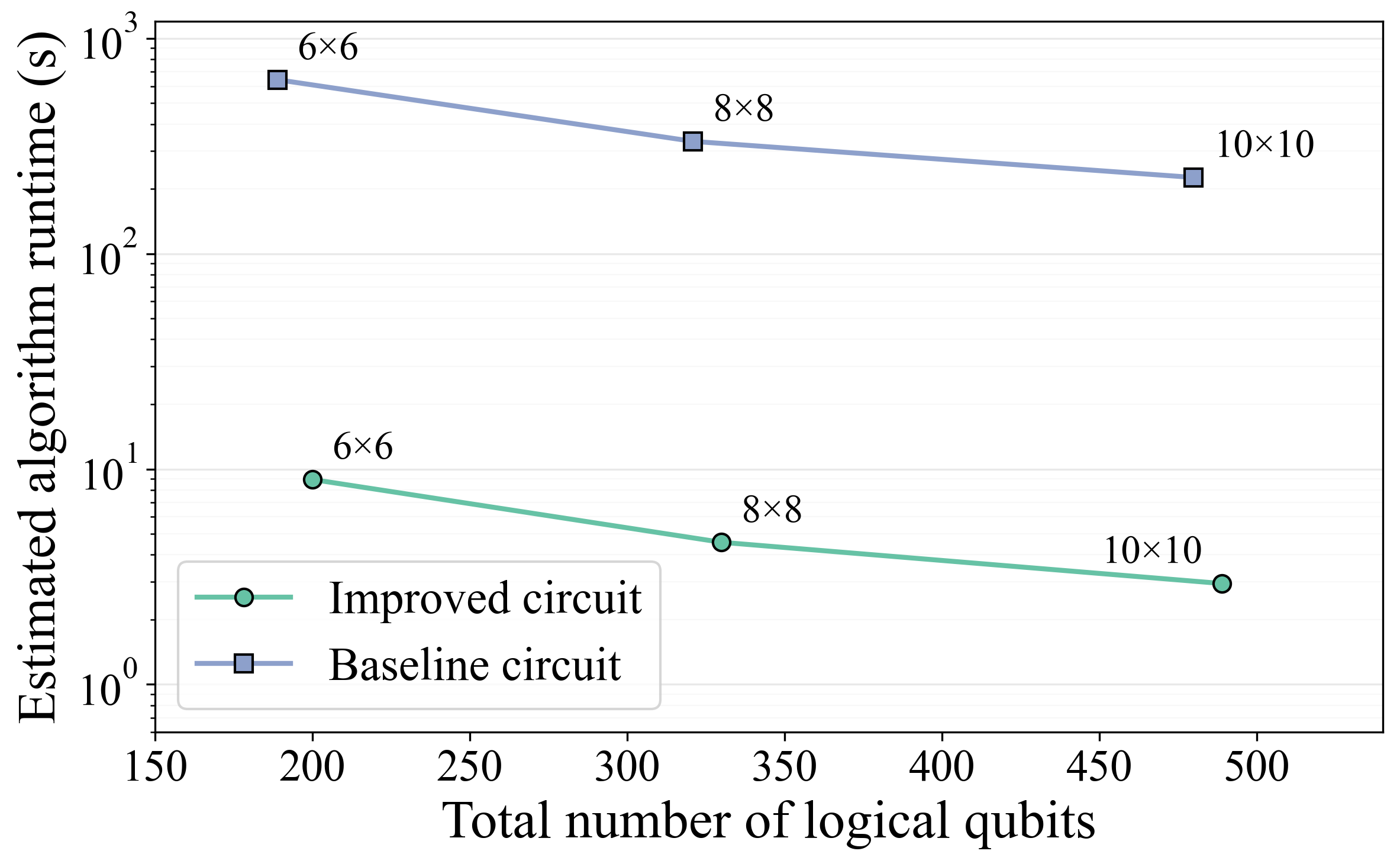}
    \label{fig:runtimes}
    \end{subfigure}
    \begin{subfigure}[b]{0.47\textwidth}
        \centering
    \includegraphics[width=0.95\linewidth,
    ]{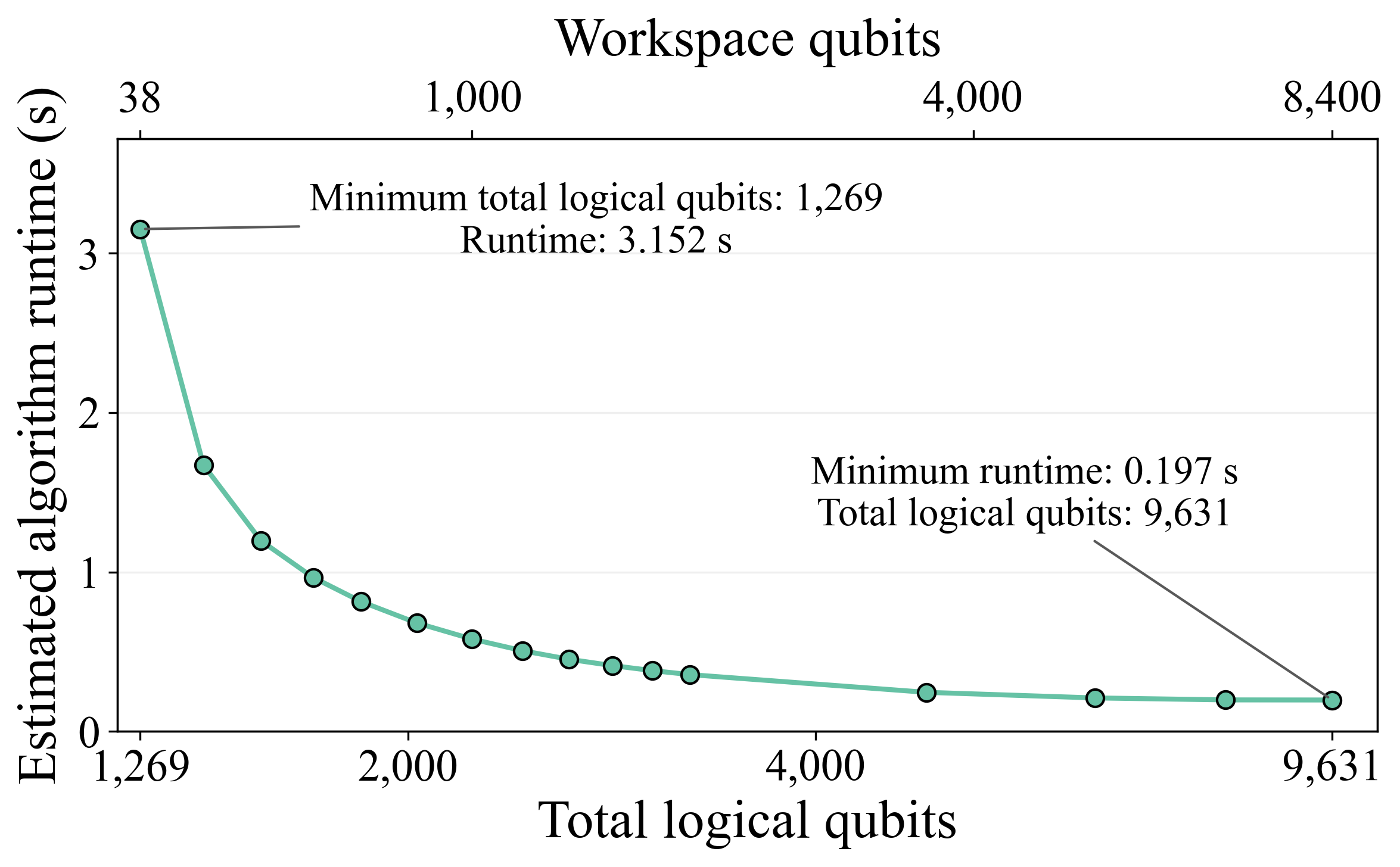}
    \label{fig:runtimes_workspace_scan}
    \end{subfigure}
    \caption{Runtime analyses using the block scheduler from Ref.~\cite{heavey2026scheduler}. \textbf{(a)} Comparison of estimated runtimes for the baseline and improved Fermi--Hubbard simulations for square lattice sizes $L=6$ to $10$, each using single batched Hamming weight phasing. 
    \textbf{(b)} Estimated runtimes for the improved $20\times20$ Fermi--Hubbard circuit as a function of logical qubit capacity. Increasing the available qubits enables greater parallelization until the reaction-limited regime is reached, revealing an approximately order-of-magnitude space-time trade-off between qubit count and runtime.
    }
    \label{fig:runtime_trend_analysis}
\end{figure}

To demonstrate how changing the qubit count affects the runtime, we also plot the runtime versus the qubit count for the case of simulating a $20 \times 20$ lattice in~\cref{fig:runtimes_workspace_scan}.
Here we can observe the spacetime trade-off facilitated by the scheduler, which allows one to roughly trade an order of magnitude in qubits for an order of magnitude in runtime.
While the absolute runtime values may depend significantly on the hardware parameters, the runtime trends observed for given numbers of workspace qubits are expected to hold.

\subsection{Reduction in Toffoli count}

Although active volume is the primary resource target for our work, Toffoli count remains the conventional metric used in many fault-tolerant resource estimates.
We therefore compare our circuits against previous Toffoli resource estimates from Refs.~\cite{campbell2022early} and~\cite{kan2024resource-optimized} in~\cref{fig:toff_comparison}.
Although active volume accounts for both Clifford and non-Clifford resources, Toffoli gates still contribute a substantial fraction of the total cost.
Consequently, compilation techniques that reduce the Toffoli count without introducing substantial Clifford overhead also reduce the active volume.
In this work including sine-windowed QPE, Suzuki--Trotter re-ordering, and re-optimized error budgeting reduce the Toffoli count, either by decreasing the number of QPE queries or by lowering the Toffoli gates per query.
For the $20 \times 20$ lattice, our circuit reduces the 
Toffoli count by $3.37 \times$ relative to \cite{campbell2022early} and $1.96 \times$ relative to \cite{kan2024resource-optimized}, establishing the lowest reported Toffoli count for this benchmark to date.\footnote{Note that Ref.~\cite{campbell2022early} reports a lower logical qubit count by avoiding phase catalysis in Hamming weight phasing, corresponding to a different point on the space-time trade-off.}

\begin{figure}[H]
\centering
\includegraphics[width=0.7\textwidth]{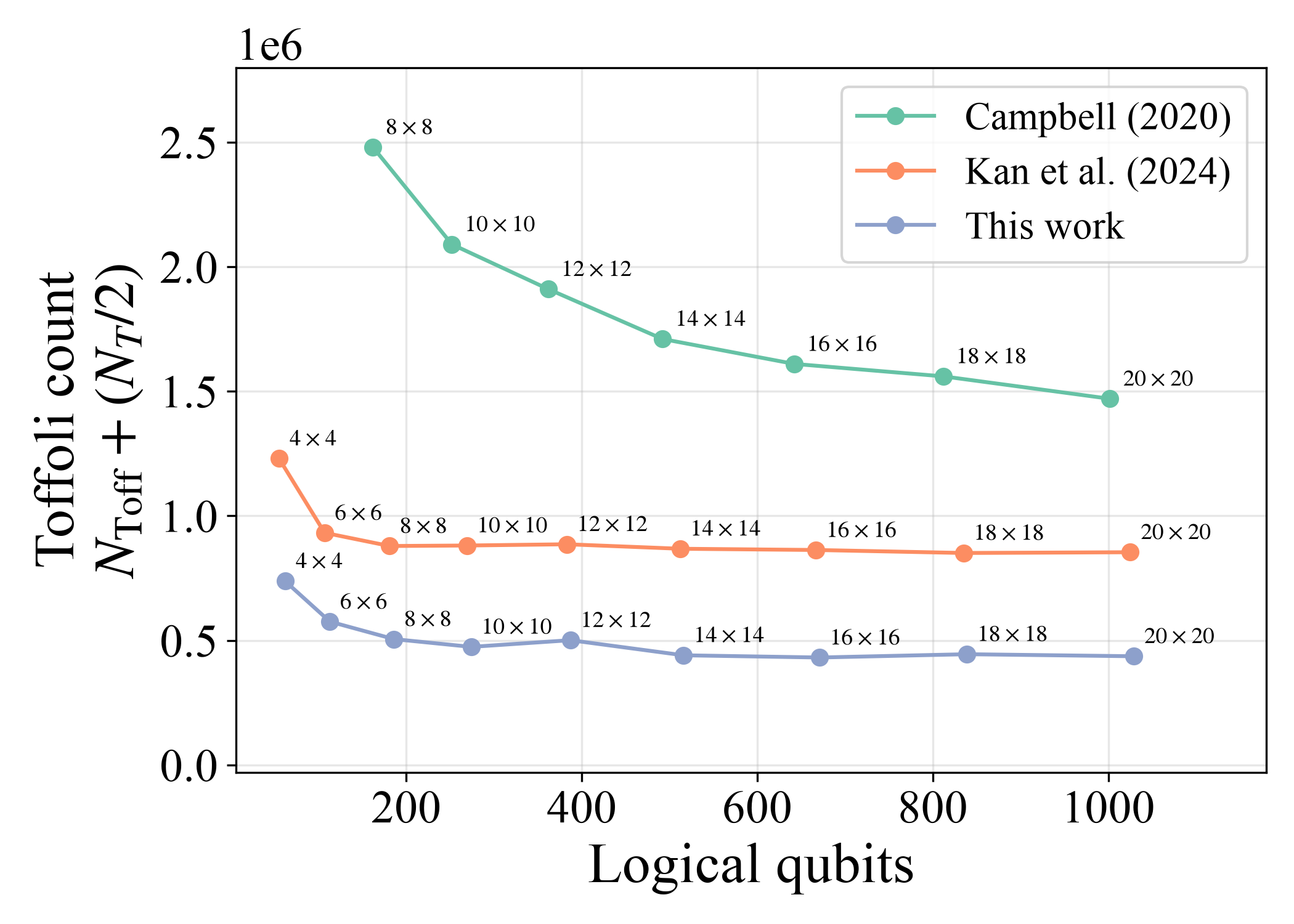}
\caption{Comparison of the total Toffoli counts between this work and those reported by Campbell~\cite[Table II]{campbell2022early} and Kan et al.~\cite[Supplementary Table 5]{kan2024resource-optimized}. All T states are assumed to be catalyzed using CCZ states for consistency.
All three works assume
$u = 8$ and $t=1$ Hamiltonian parameters, and these works assume Hamming weight phasing is applied to batches of $L^2/2$ equiangular rotations. We additionally note that Ref.~\cite{campbell2022early} rounded to two significant figures for the Toffoli counts.
}
\label{fig:toff_comparison}
\end{figure}

Moreover this comparison is intentionally conservative.
We explicitly compile the controlled QPE implementation and therefore restrict the number of queries to be of the form $2^{k-1}$ for $k$ phase qubits. 
Previous work instead reports costs for arbitrary query counts obtained by scaling a single query.
Relaxing this restriction would further reduce our Toffoli counts but would require a more automated compilation of the controlled QPE circuit.

\section{Compilation details}\label{sect:compilation}

In this section, we outline the compilation choices and strategies used to reduce the active volume of the ground-state energy estimation algorithm.
Some of these improvements act at the algorithmic level by reducing the number of Trotterized time-evolution queries required by QPE, while others act at the circuit level by lowering the cost of each query.
Despite being motivated by active volume reduction, nearly all of these techniques are general improvements that would also benefit other FTQC architectures.
Only the gate-level compilation using operations optimized through oriented ZX diagrams relies specifically on the active volume architecture.

As our focus is a comparative study of QPE costs for ground-state energy estimation, we account only for queries to the time-evolution operator and neglect initial-state preparation, window-state preparation, and the inverse QFT, for which we do not introduce new compilation improvements.
The last two costs are negligible compared to the query cost \cite{babbush2018encoding, najafi_optimum_2023}.
Initial state preparation (for lattice systems) may be more significant and is an active area of research \cite{fomichev2023initial, berry2024rapid, smith2024constantdepth, job2024cost, murta2024from, guo2025preparation}.
It is important to ensure that the overhead from the approximate eigenstate preparation remains low, as obtaining reliable, sufficient-quality initial states can become prohibitively expensive at the large lattice sizes we focus on, as in the case of e.g. tensor network states~\cite{haghshenas2022variational, lee2023evaluating, fomichev2023initial, job2024cost,  liu2025accurate, wu2025alternating}. Instead, we may employ a cheaper but coarser method such as mean-field states \cite{wecker2015solving, scholle2023comprehensive, roth2025superconductivity, langmann2026update, matsuyama2026bcs} with reasonable overlap.
Applying a low-resolution QPE as a filter \cite{fomichev2023initial} can then refine the eigenstate approximation, further motivating our focus on the query cost.

\subsection{Circuit overview}

\cref{fig:big_summary} provides a high-level overview of the circuit construction used for ground-state energy estimation, while the remainder of this section describes its compilation in detail.
We first specify the QPE implementation~(\cref{fig:big_summary} \textcolor{blue}{(I)}), including the use of the sine window and directionally controlled unitaries, in~\cref{subsubsec:sine_window}.
We then discuss the ordering of the Suzuki--Trotter terms depicted in~\cref{fig:big_summary} \textcolor{blue}{(II)} and the opportunities for merging consecutive evolutions in~\cref{subsection:trotter_reordering}.
We next introduce an efficient implementation of the controlled and directionally controlled Hamming weight phasing routine~\cref{fig:big_summary} \textcolor{blue}{(V)} tailored to these circuits, including the necessary phase fix-ups, in~\cref{subsec:hwp}.
The subsequent subsections describe three circuit-level optimizations.
We first reduce the cost of the fermionic SWAP network in~\cref{subsec:fswap}, corresponding to~\cref{fig:big_summary} \textcolor{blue}{(IIIc)}.
We then develop active-volume-efficient implementations of the gates appearing in~\cref{fig:big_summary} \textcolor{blue}{(IV)} in~\cref{subsec:preoptimized_circuits}, before analyzing the space-time trade-off offered by a batched implementation of Hamming weight phasing in~\cref{subsec:batched_hwp}.
Finally, we characterize the errors introduced by the compiled circuit and determine a favorable 
division in the error budgets in~\cref{subsubsec:error_budgeting_for_circuit}.

\clearpage
\begin{figure}[H]
\centering
\begin{tikzpicture}
\begin{scope}[on background layer]
         \node [inner sep=0pt] (img1) at (0, 0) {\includegraphics[trim=5 5 5 5,clip,width=0.9\linewidth]{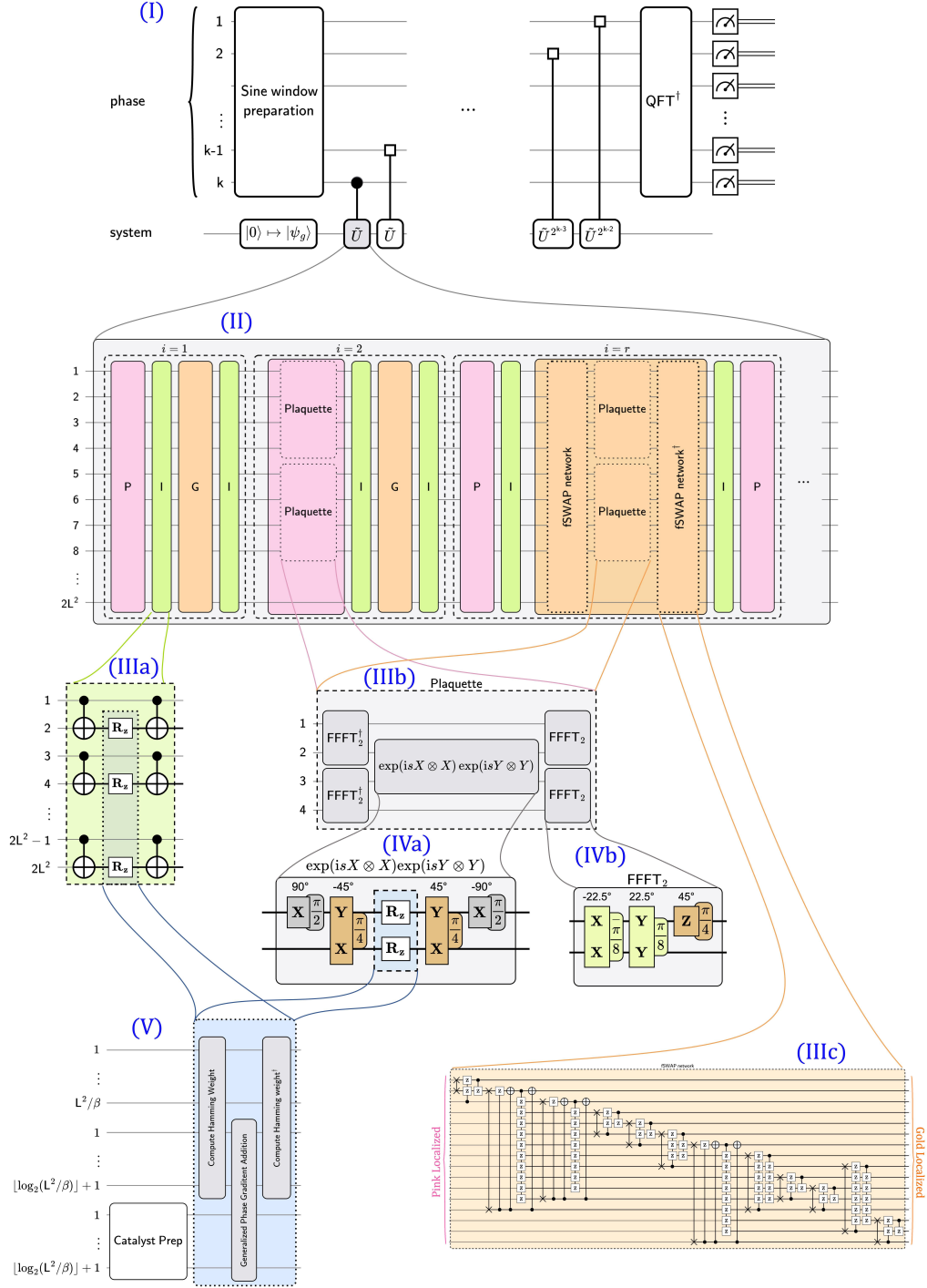}};
\end{scope}
\end{tikzpicture}
\vspace{-5mm}
\caption{Circuit for the ground-state energy estimation algorithm. 
The square control marker in the QPE represents the control structure of a directionally controlled unitary \cite{wecker2015solving, kivlichan2020improved}.}
 \label{fig:big_summary}
\end{figure}

\subsection{Adopting the sine window QPE variant}\label{subsubsec:sine_window}

Quantum phase estimation employs a finite-sized phase register (thus a finite number of queries) which leads to a phase error, $\epsilon_{\text{QPE}}\tau$ given time evolution unitary $\mathrm{e}^{\mathrm{i} H \tau}$.
\cref{fig:big_summary} \textcolor{blue}{(I)} illustrates a QPE with $k$ phase qubits. 
One way to quantify this error is through the variance of the phase distribution prior to measurement.
The Holevo variance provides a natural measure of phase uncertainty because it accounts for the periodicity of the phase \cite{berry2000optimal}.
We compare two QPEs that share the same asymptotic scaling in the Holevo variance but differ in constant factors.

\begin{figure}[H]
    \includegraphics[width=0.8\textwidth]{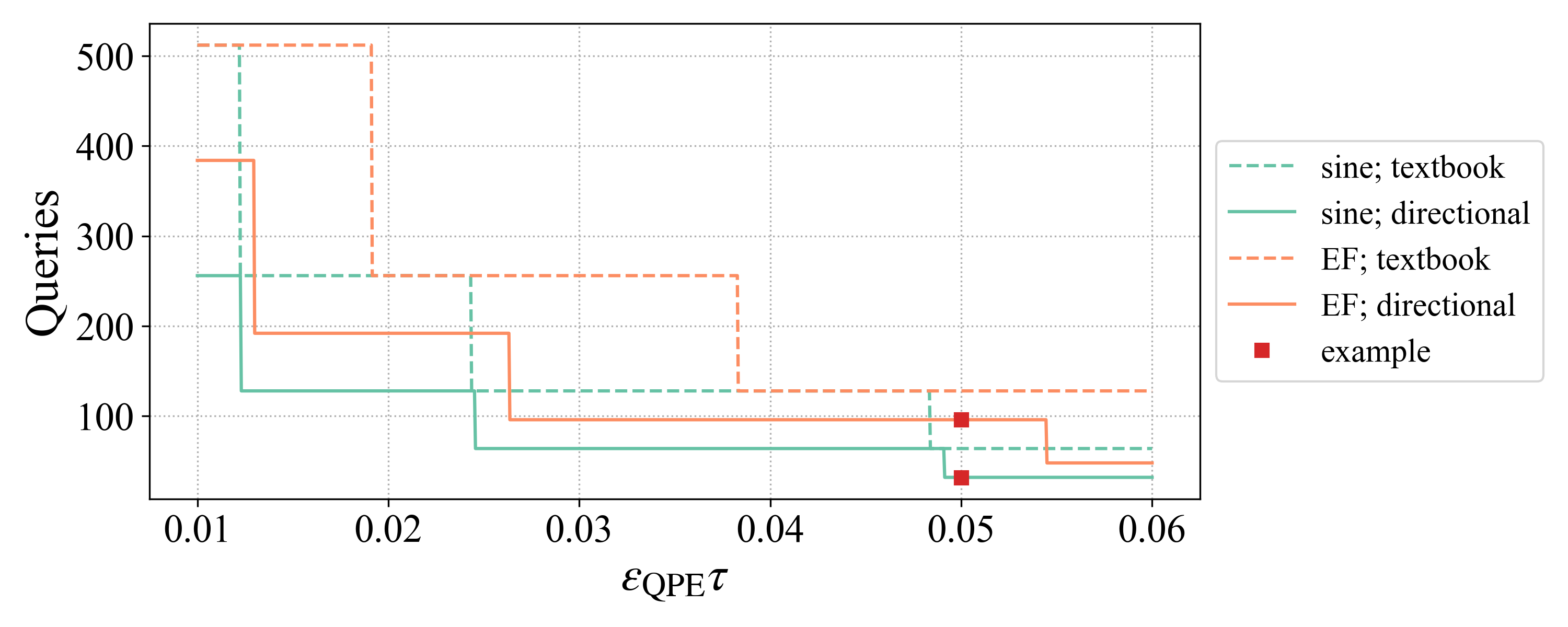}
    \centering
\caption{
Comparing number of queries between sine-windowed QPE \cite{babbush2018encoding, najafi_optimum_2023} versus the entanglement-free variant as a function of the allocated error budget for phase error. For each flavor of QPE, we consider applying the conventional control structure (referred to as ``textbook'' in the figure legend) versus applying directionally controlled time evolution operators \cite{kivlichan2020improved, babbush2018encoding}. The red squares mark an example point in which the product of the QPE energy error and evolution time equals $0.05$. At these parameters, the sine-windowed QPE requires 32 queries whereas the entanglement-free QPE requires 96 queries.
}
\label{fig:qpe_queries}
\end{figure}

Prior works on the quantum simulation of Fermi-Hubbard \cite{kivlichan2020improved,campbell2022early,kan2024resource-optimized} have considered a variant of iterative QPE that achieves Heisenberg scaling using an adaptive scheme \cite{higgins2007entanglement}. 
We will refer to this method as the ``entanglement-free'' (EF) QPE. 
In this QPE, the experiment or circuit to estimate each bit is repeated $M$ times: for each of the $k$ phase qubits, the same power of $U$ is applied $M$ times.
In addition to this repeating structure, after measuring each phase qubit, the estimate of the eigenphase is updated by $\delta\theta$ determined by a Bayesian algorithm that considers all previous measurement results.
In the original work $M=6$ was empirically determined to be sufficient to achieve Heisenberg scaling.
From this fit, they estimated the scaling to be off by a factor of around $1.56$ compared to the analytical bound on the uncertainty (standard deviation), i.e. 
\begin{align}\label{eq:uncertainty_ef_qpe}
    \Delta\theta_{\text{est, EF}} \approx 1.56 \pi / N
\end{align}
\noindent for $N \gg 1$ total queries.
One of the challenges with the method is that each additional phase qubit incurs a higher query cost to the unitary due to the multiplicative factor of $M$ in query counts compared to a standard coherent QPE. 

An alternative way to achieve Heisenberg scaling in QPE is by employing the sine window \cite{luis_optimum_1996, babbush2018encoding, najafi_optimum_2023}. 
The sine window state is prepared in the phase register prior to applying the controlled unitaries, as illustrated in~\cref{fig:big_summary} \textcolor{blue}{(I)}.
Sine-windowed QPE achieves the minimum Holevo variance $V_H$,
given $N$ queries \cite{najafi_optimum_2023}.
Holevo variance is approximately equal to the standard notion of variance in the regime of narrow distributions.
Thus we take the squared error in the QPE obtained phase, or $\Delta\theta_{\text{est}}^2$, to be equal to the Holevo variance.
In this work, we choose to employ a QFT-based QPE with a sine window, which does not significantly impact the qubit count.
To set a fair comparison, we can re-express the uncertainty using the sine window as:
\begin{align}\label{eq:uncertainty_sine_qpe}
    \Delta \theta_{\text{est, sine}} = \sqrt{V_H} = \tan \bigg( \frac{\pi}{N+1} \bigg),
\end{align}
for $N$ queries.
The phase error tolerance is related to the energy estimate by $\Delta \theta =\epsilon_\text{QPE} \tau$ where recall $\tau$ is the time used in the evolution. 
We set the target uncertainties of each QPE,~\cref{eq:uncertainty_ef_qpe} and~\cref{eq:uncertainty_sine_qpe}, to this error
and solve for the queries:
\begin{align}
    N_{\text{EF}} &= \frac{1.56 \pi}{\epsilon_\text{QPE} \tau} \label{eq:queries_berry}\\
    N_{\text{sine}} &= \frac{\pi}{\arctan(\epsilon_\text{QPE} \tau)} - 1  \label{eq:queries_sine}
\end{align}

From these total number of queries, we want to estimate the number of phase qubits in each QPE variant.
First, for EF-QPE, we solve for the number of phase qubits:
\begin{align}
    \frac{1.56 \pi}{\epsilon_\text{QPE} \tau} \leq M (2^{k_{\text{EF}}}-1),
\end{align}
where $M$ is, again, the number of repetitions of each power of $U$.
This results in:
\begin{align}
    k_{\text{EF}} = \bigg\lceil \log_2 \bigg( \frac{1.56\pi}{\epsilon_\text{QPE} \tau M} + 1 \bigg) \bigg\rceil,
\end{align}
which we round to an integer with ceiling.
With this number of phase qubits that achieves the desired error, we can now compute the number of queries if we were to employ a QPE with directional controlled unitaries \cite{babbush2018encoding}\footnote{ Ref.~\cite{babbush2018encoding} presents a QPE circuit that doubles the effect of phase kickback via use of ``directionally controlled'' unitaries \cite{wecker2015solving, kivlichan2020improved}, i.e. applying $U$ when the control qubit is $\ket{1}$ and applying $U^\dagger$ when the control qubit is $\ket{0}$. Such QPE requires $2^{k-1}$ queries for $k$ phase qubits.}:
\begin{align}\label{eq:queries_ef_directional}
    N_{\text{EF, directional}} = M \bigg( 2^{ k_{EF}-1} \bigg) = M \bigg( 2^{\lceil \log_2(1.56\pi/(\epsilon_\text{QPE} \tau M)) + 1 \rceil - 1}  \bigg).
\end{align}
For the number of queries for the textbook QPE, we can instead compute $N_{\text{EF, textbook}} = M ( 2^{ k_{EF}}-1 )$.
We repeat the same procedure for the sine windowed QPE and arrive at:
\begin{align}
    N_{\text{sine, directional}} &= 2^{ k_{ \text{sine}}-1} \\
    N_{\text{sine, textbook}} &= 2^{ k_{ \text{sine}}}-1,
\end{align}
where $k_{ \text{sine}} = \lceil \log_2(\pi/\arctan{(\epsilon_\text{QPE} \tau)}) \rceil$.
The number of queries for both the EF-QPE and the sine-windowed variants are compared in~\cref{fig:qpe_queries}.
These numerics confirm that using the sine-windowed QPE reduces the number of queries required to achieve a chosen target precision.
Therefore, we choose the sine-windowed QPE to maintain the Heisenberg scaling while also reducing the number of queries.

\subsection{Ordering of Suzuki-Trotter sequence and merging terms}\label{subsection:trotter_reordering}

We implement the time evolution unitary using a second-order Suzuki-Trotter formula, partitioning the Hamiltonian into three non-commuting components (see \cref{fig:fh_plaq}): an interaction term and two sets of non-overlapping plaquette terms.
Repeating the symmetric product formula means that evolutions with the same generator can often appear consecutively and thus can be merged via $\mathrm{e}^{\mathrm{i}H_i s_1}\mathrm{e}^{\mathrm{i}H_i s_2} = \mathrm{e}^{\mathrm{i}H_i (s_1+s_2)}$.
This is particularly advantageous when the implementation cost is largely independent of the evolution time, as in our case where the time determines the angle of rotation.
Specifically, when simulating $H_\text{FH} = H_1 + H_2 + H_3$ for a time $\tau$ using $r$ steps, the resulting sequence admits the following merging,
\begin{subequations}
\begin{align}
\mathrm{e}^{\mathrm{i}H_\text{FH}\tau} &\approx  (\mathrm{e}^{\mathrm{i}H_1\tau/2r}\mathrm{e}^{\mathrm{i} H_2 \tau/2r} \mathrm{e}^{\mathrm{i}H_3\tau/r}\mathrm{e}^{\mathrm{i} H_2\tau/2r}\mathrm{e}^{\mathrm{i}H_1\tau/2r})^r\\
& = \mathrm{e}^{\mathrm{i}H_1\tau/2r}(\mathrm{e}^{\mathrm{i}H_2\tau/2r} \mathrm{e}^{\mathrm{i}H_3\tau/r}\mathrm{e}^{\mathrm{i}H_2\tau/2r}\mathrm{e}^{\mathrm{i}H_1\tau/r})^{r-1}\mathrm{e}^{\mathrm{i}H_2\tau/2r}\mathrm{e}^{\mathrm{i}H_3\tau/r}\mathrm{e}^{\mathrm{i}H_2\tau/2r}\mathrm{e}^{\mathrm{i}H_1\tau/2r}.
\end{align}
\end{subequations}
Then in QPE repeated applications of the unitary $\mathrm{e}^{\mathrm{i}H_\text{FH}\tau}$ to extract phase information enables further merging between queries. 
For a $k$-qubit phase register and a unitary implemented with $r$ second-order Trotter steps the number of queries to each evolution is given in \cref{tab:no_queries_merging}.

\begin{table}[h!]
\renewcommand{\arraystretch}{1.2}
\centering
\begin{tabular}{|c | c | c  | c |}
\hline
 \textbf{QPE control structure} & $\mathrm{e}^{\mathrm{i}H_1 (\cdot)}$ & $\mathrm{e}^{\mathrm{i}H_2 (\cdot)}$ & $\mathrm{e}^{\mathrm{i}H_3 (\cdot)}$ \\ 
 \hline
  \textit{Standard/Textbook} & $r(2^k-1) +k$ & $2r(2^k -1)$  &  $r(2^k-1)$ \\
  \hline
  \textit{Directional} & $r2^{k-1}+k$ & $r2^{k}$ & $r2^{k-1}$   \\ 
  \hline
\end{tabular}
\caption{Query count for the three Hamiltonian components after merging consecutive evolutions between Trotter steps and across queries controlled on the same phase qubit.
Standard refers to textbook QPE control structure.
Note that this query count will include evolutions of different times, e.g. $\mathrm{e}^{\mathrm{i}H_1 \tau}$ and $\mathrm{e}^{\mathrm{i}H_1 \tau/2}$ both contribute to the cited count.
For completeness they are divided as follows: in the case of standard QPE, there are $r(2^{k}-1) - k$ calls to $\mathrm{e}^{\mathrm{i}H_1 \tau}$ and $2k$ calls to $\mathrm{e}^{\mathrm{i}H_1 \tau/2}$.
In the case of QPE using directional controls, there are  $r2^{k-1}-k$ calls to $\mathrm{e}^{\mathrm{i}H_1 \tau}$ and $2k$ calls to $\mathrm{e}^{\mathrm{i}H_1 \tau/2}$.}
\label{tab:no_queries_merging}
\end{table}

The ordering of Hamiltonian partitions affects both the Trotter error and the circuit cost via differing query counts, introducing an additional compilation degree of freedom.
The query counts satisfy $n_Q(H_2)>n_Q(H_1)>n_Q(H_3)$, so from a compilation standpoint it is advantageous to assign the most resource-intensive evolution to $H_3$ and the cheapest evolution to $H_2$. 
For our decomposition the cost of evolutions is ordered $\mathcal{C}(H_h^g)>\mathcal{C}(H_h^p)>\mathcal{C}(H_I)$:  all evolutions involve a tower of rotations; hopping terms additionally require two-mode FFFT gates; and `gold' hopping further incurs a cost of the $\mathrm{fSWAP}$ network, as shown in \cref{fig:big_summary}.

Previous resource estimates \cite{campbell2022early, kan2024resource-optimized} used the ordering $H_1 = H_I$, $H_2 = H_h^p$, $H_3 = H_h^g$ (the `IPG ordering'), motivated by analytically tractable commutator bounds derived in \cite{campbell2022early}.
However, we will see that a comparable commutator bound can be achieved for our suggested ordering and so the choice can instead be guided by the query count.
This motivates the Trotter ordering $H_1 = H_h^p$, $H_2 = H_I$, $H_3 = H_h^g$ (the `PIG ordering') illustrated in \cref{fig:merging}.

\begin{figure}[H] 
    \centering
    \includegraphics[trim=10 10 10 10,clip,width=0.8\textwidth]{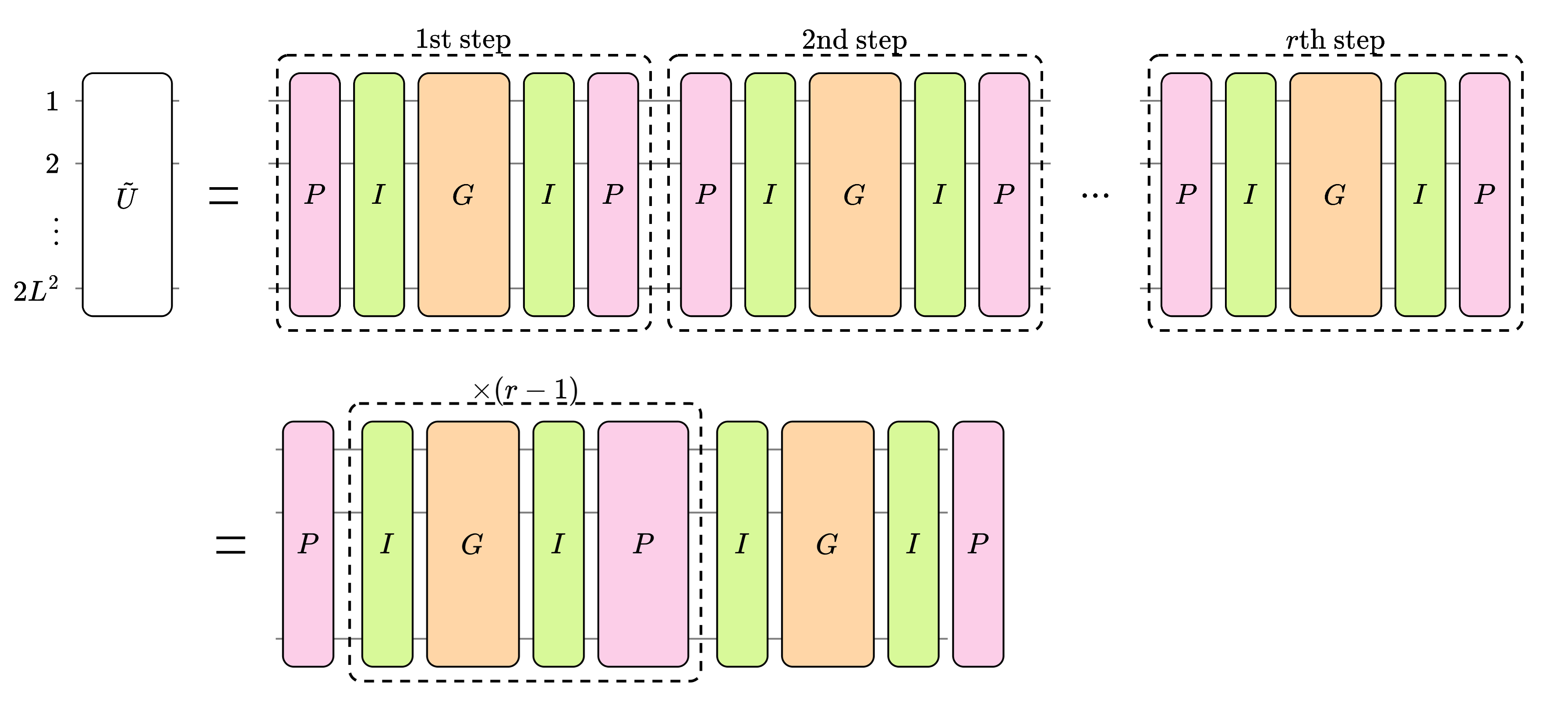}
    \caption{Circuit diagrams demonstrating Trotter ordering and merging within a single evolution. The `PIG' order allows merging of the more resource-intensive hopping terms. In addition to merging within a query, we allow merging between sequential queries controlled on the same phase qubit.
    }
    \label{fig:merging}
\end{figure}

\subsubsection{Controlling the Suzuki-Trotter sequence}\label{subsubsec:controlling_trotter_sequence}

In the QPE algorithm, the Trotterized time evolution operators $\tilde{U}$ are controlled, as shown in~\cref{fig:big_summary} \textcolor{blue}{(I)}.
For QPE leveraging directional phase kickback we require both closed-controlled and directionally controlled versions of $\tilde{U}$.
Given the `PIG' ordering used in $\tilde{U}$, the two types of controlled-$\tilde{U}$ operations can be compiled to leverage the symmetry in the second-order Suzuki-Trotter sequence. 

\begin{figure}[H]
\centering
    \includegraphics[trim=10 10 10 10,clip,width=0.85\linewidth]{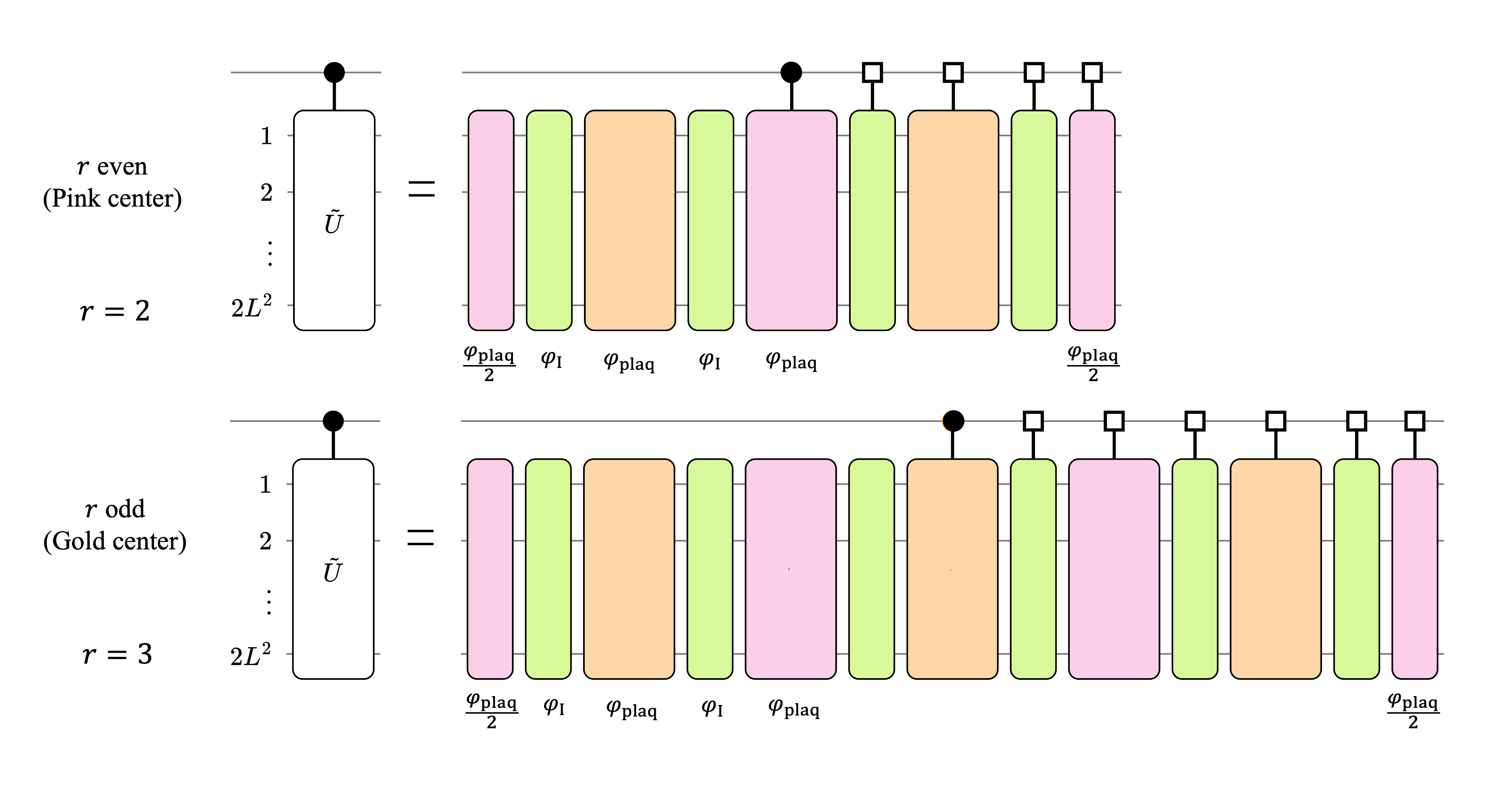}
\caption{Controlling the time-evolution operator. Controlled evolution operator for small numbers of Trotter steps (i.e. $r=2$ and $r=3$ to represent even and odd cases) to demonstrate the control structure. 
    The symbols under each type of operation indicate the angle in the rotation towers.
    For pink plaquette terms, the two operations at the ends have half the angle but all intermediate pink and gold terms should have the same angle.
    This method of controlling half the numbers of Trotter terms is from Ref.~\cite{simon2025halving}. Using the `PIG' ordering, the controlled operation always corresponds to a pink or gold plaquette term with (merged) angle $\varphi_{\text{plaq}}$.}
\label{fig:controlling_trotter_sequence}
\end{figure}

For the closed-control case on the first phase qubit, only the center term has to be controlled.
For an even number of Trotter steps, $r$, the center term is a pink plaquette term while for an odd $r$, the center term is a gold plaquette term.
Using the compilation strategy from Ref.~\cite{simon2025halving}, we can avoid controlling around half of the terms and only control the center term and directionally control the following terms, see~\cref{fig:controlling_trotter_sequence}.
That is, there are $4r + 1$ terms, in which one of them is closed-controlled and $2r$ of them are directionally controlled.
In the following section, we propose a method for efficient directional control of Trotter terms in this work.

For the directional-control case on subsequent phase qubits, all constituent terms in $\tilde{U}$ are directionally controlled.
Again, we can compile this control structure with reduced overhead -- for each of the three types of Trotter terms, circuit symmetry allows the required (directional) control to be applied only to the parallel set of equiangular rotations.
In the following section, we discuss an efficient way of implementing this directional control structure.

\subsection{Hamming weight phasing}\label{subsec:hwp}

The evolution of both the plaquette term (\cref{fig:big_summary} \textcolor{blue}{(IVa)}) and the interaction term (\cref{fig:big_summary} \textcolor{blue}{(IIIa)}) each contains towers of $m = L^2$ equiangular rotation gates.
Hamming weight phasing (HWP) is a compilation technique to synthesize such a tower of $m$ rotations using logarithmically many ancilla qubits \cite{gidney2019efficient, kivlichan2020improved}.
While HWP has been widely applied to  compilation schemes for the Fermi-Hubbard model \cite{campbell2022early, kivlichan2020improved, kan2024resource-optimized} due to its favorable Toffoli complexity, its practical benefit is architecture dependent.
On architectures with limited connectivity recent works have opted to directly synthesize these equiangular rotations~\cite{chung2026partiallyfaulttolerantquantumcomputation, FLASQ, webster2026pinnaclearchitecturereducingcost}.
For instance, in Ref.~\cite{FLASQ} computing the Hamming weight requires routing between data and ancilla registers via additional patches, consuming workspace that could otherwise be used for executing logical operations. 
In contrast, the non-local connectivity of the active volume architecture largely mitigates this overhead, and the spacetime tradeoff of HWP can be advantageous when sufficient memory is available \cite{litinski2022activevolume}.
We will return to further analyze this tradeoff in \cref{subsec:batched_hwp}.

Refs.~\cite{wang2023option, kan2024resource-optimized} introduced a Toffoli-optimized version of HWP 
costing
$m + \lfloor \log_2(m) \rfloor - w(m) + 1$ Toffoli gates and only one rotation (referred to as the ``payload'' rotation), where $w(m)$ is the Hamming weight of the binary representation of $m$.
The circuit first computes the Hamming weight of the $m$-qubit register onto a logarithmically-sized ancillary register, then applies the generalized phase-gradient addition circuit (GPGA) \cite{gidney2019efficient, wang2023option} to synthesize the required phase gradient using a reusable $(\lfloor \log_2(m) \rfloor + 1)$-qubit catalyst.
Then the Hamming weight is uncomputed.

There is additionally a one-time cost of preparing the catalyst register used in GPGA.
There are two distinct HWP circuits corresponding to the interaction and plaquette terms since their corresponding rotation angles generally differ, i.e. $u\neq t$.\footnote{For certain relationships between the magnitudes and signs of $t$ and $u$, additional catalyst states may be shared e.g. by replacing some additions with subtractions via negations. Although $t=1$ and $u=8$ permits this type of fine-grained optimization, we retain the more general construction to accommodate other parameter choices.}
Each unique HWP circuit contributes $\lfloor \log_2(L^2)\rfloor + 1$ catalyst rotations.
As we merge pink plaquette evolutions in the `PIG' term ordering, the change in the evolution time doubles the rotation angle with respect to the un-merged evolution and thus requires a different phase gradient state.
Instead of creating a full new catalyst state, all but one catalyst qubit can be reused.
Hence the total number of catalyst rotations is $2\lfloor \log_2(L^2)\rfloor + 3$.

Due to the symmetry of operations in each term, controlling each term corresponds to implementing the respective control structure for the HWP subroutine.
Our main contribution to HWP is reducing the cost of implementing the control structure and providing explicit circuits for this.
Previous resource estimates often neglect the cost of controlling the evolution, and instead scale the cost of a single (uncontrolled) evolution by the number of queries \cite{campbell2022early, kan2024resource-optimized}.
Naive control methods can be inefficient and materially increase the complexity of the subroutine, e.g. controlling an adder doubles the complexity.
Ref.~\cite[Appendix A]{sanders2020compilation} describes a less naive method to directionally control HWP.
Our proposed implementation further reduces the cost of the control structure.
We incorporate this overhead explicitly into our resource estimates for a more complete instantiation of the circuit and confirm that controlling the unitary adds insignificant overhead, especially after our optimization.

\subsubsection{Directionally controlled Hamming weight phasing}\label{subsec:bi_hwp}

A controlled HWP can be implemented by controlling the payload rotation and the uncompute of the out-of-place adder-like operation 
in the GPGA circuit as shown in~\cref{fig:controlled_hwp}.
One subtlety is that our target operation is a tower of $R_z$ rotations whereas the HWP framework traditionally assumes a tower of \emph{phase} gates (e.g.~\cite[Fig.~14]{gidney2019efficient}).\footnote{One could instead replace the payload rotation with an $R_z$ gate, but for either choice, the controlled and directionally controlled HWP circuits require phase fix-up operations.}
While these two gates differ only by a global phase when uncontrolled, upon introducing control, this phase becomes relative and must be corrected via a `fix-up' rotation on the control qubit.

\begin{figure}[H]
\centering
\includegraphics[trim=15 15 15 15,clip,width=0.95\linewidth]{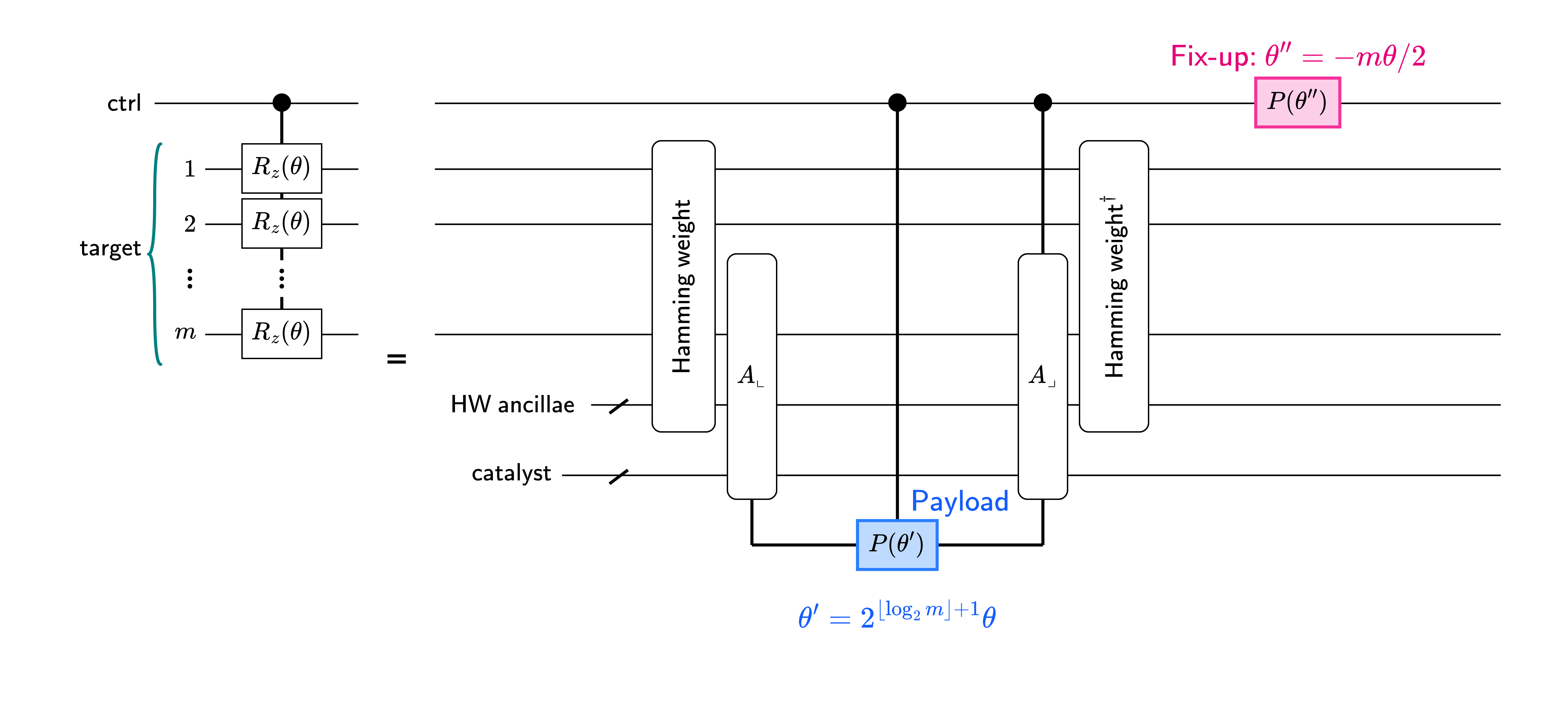}
    \caption{Controlled Hamming weight phasing. 
    A controlled HWP is implemented by controlling only the payload phase rotation and operation $A_{\lrcorner}$ (see~\cref{fig:phasing_circuit_av_counts} for its gate decomposition).
    The controlled payload decomposes into one left elbow (temporary AND compute), one uncontrolled rotation, and one right elbow (temporary AND uncompute), together with a local phase fix-up rotation on the control qubit.
    Recall that adder-like circuits $A_{\llcorner}$ and $A_{\lrcorner}$ act on qubits encoding the Hamming weight, which may be qubits from the target as well as the ancilla registers \cite{kan2024resource-optimized}. Controlling the GPGA circuit replaces $\lfloor \log_2 (L^2) \rfloor + 1$ of the CNOTs in $A_{\lrcorner}$ with Toffoli gates \cite{gidney2018halving}.
    }
    \label{fig:controlled_hwp}
\end{figure}

Our key compilation technique is to realize directional control of HWP through applying operations on the GPGA catalyst register rather than on the target register.
A straightforward approach to achieve this is to apply open-controlled CNOT fanouts to the target register before and after every HWP invocation, thereby switching between the addition and subtraction-like circuits \cite{sanders2020compilation}.
However, because the target register is acted upon by intermediate operations such as fSWAPs, two-mode FFFTs, and CNOTs in the interaction and plaquette terms, these fanouts must be repeated for every HWP call.

Instead, we apply open-controlled CNOT fanouts to the GPGA \emph{catalyst register}, leaving the target register untouched, as shown in~\cref{fig:directionally_controlled_hwp}.
The payload phase rotation is similarly conjugated by open-controlled CNOTs to reverse the sign of the applied phase.
Together, these transformations replace separately controlled implementations of the forward and inverse $R_z$ rotation towers with a single directionally controlled HWP construction.
The catalyst and payload transformations again introduce local phase corrections on the control qubit.
These corrections combine into a single $R_z$ rotation and the resulting circuit is equivalent to the directionally controlled rotation towers up to a global phase.\footnote{Note that we are implementing $\mathrm{e}^{\mathrm{i}\gamma}C(U)$ for a global phase $\gamma$ instead of $C(U)$, where $C(U)$ is the controlled unitary in QPE. 
This is different from preparing $C(\mathrm{e}^{\mathrm{i}\gamma}U)$, which would impact the phase estimate.}
After the first conventionally controlled HWP on the first QPE phase qubit (see~\cref{fig:big_summary} \textcolor{blue}{(I)}), all subsequent HWP invocations are directionally controlled.
Consequently, the catalyst register need only be negated once before the first directionally controlled HWP and restored after the final one.
This reduces the catalyst overhead to just two logarithmic-sized CNOT fanouts per QPE phase qubit, rather than incurring this cost for every HWP call.

\begin{figure}[H]
\centering
\includegraphics[trim=10 10 10 10,clip,width=\linewidth]{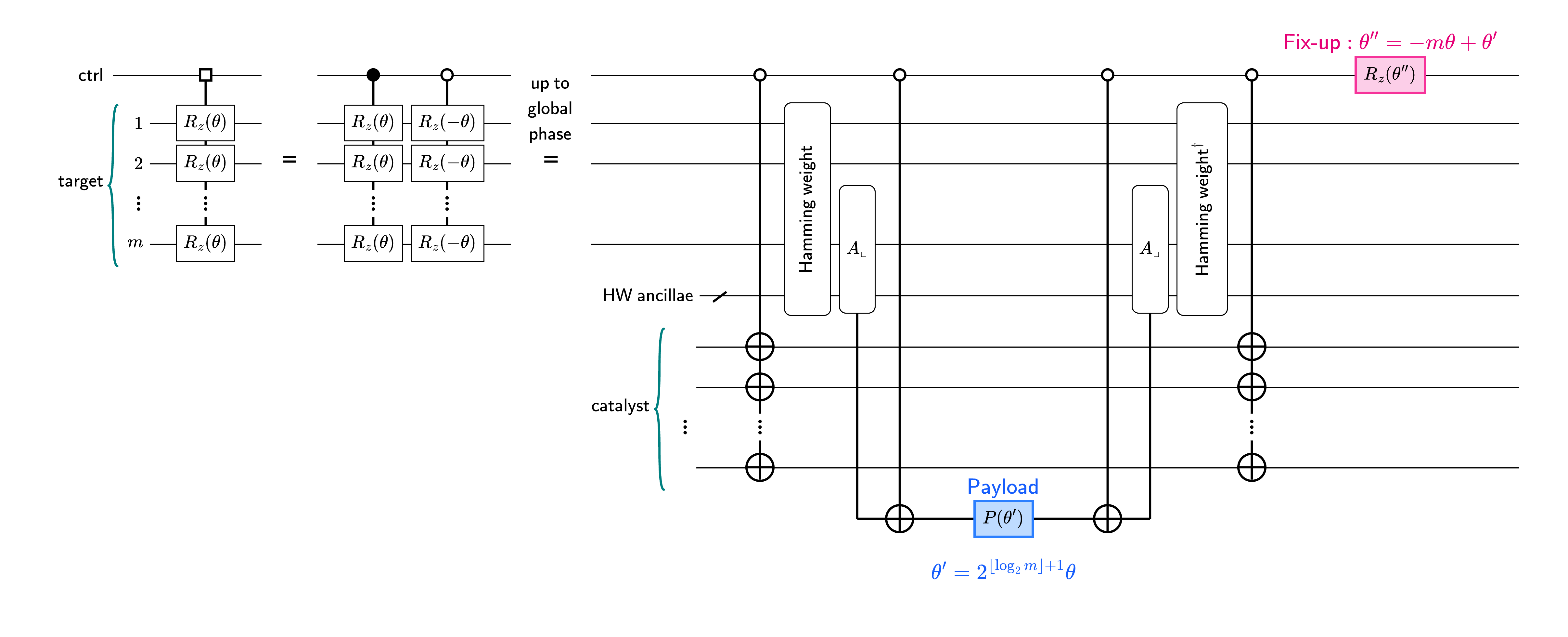}
\caption{
Directionally controlled Hamming weight phasing implemented by acting on the logarithmic-sized GPGA catalyst register rather than the target register. Open-controlled CNOT fanouts toggle the GPGA circuit between addition and subtraction-like circuits, while conjugating the payload phase rotation reverses its sign. A local phase correction is required on the control qubit.}
\label{fig:directionally_controlled_hwp}
\end{figure}

The local phase corrections from successive directionally controlled HWP calls commute and therefore accumulate into a single $R_z$ rotation on each of the remaining $k-1$ phase qubits.
As shown in~\cref{appen:phase-fix-ups}, these accumulated corrections form a generalized phase gradient and can therefore be synthesized directly using GPGA rather than separate rotations.
Together with the amortized catalyst fanouts described above, this keeps the control overhead small.
\cref{tab:control_cost_breakdown} gives the complete resource breakdown.
For the $20 \times 20$ lattice, the active volume from the control structure is only around $7\times 10^3-8\times 10^3$ blocks, compared to around $10^7$ blocks from the rest of the circuit. 

{
\renewcommand{\arraystretch}{1.2}
\begin{table}[h!]
\centering
\begin{tabular}{|>{\centering\arraybackslash}p{3.2cm}|>{\raggedright\arraybackslash}p{10cm}|}
\hline
\textbf{QPE phase qubit(s)} &
\multicolumn{1}{c|}{\textbf{Resources}} \\
\hline

\multirow{4}{3.2cm}{\centering First phase qubit}
& One left and right elbow for controlling the payload rotation in the center term \\
\cline{2-2}
& $\lfloor \log_2(L^2) \rfloor + 1$ Toffoli gates (controlling adder-like circuit) \\
\cline{2-2}
& One $R_z$ rotation (phase fix-up) \\
\cline{2-2}
& $2r \times 2$ CNOTs (directional control of payload rotations in HWP calls) \\
\cline{2-2}
& Two CNOT fanouts with target size $2\lfloor \log_2(L^2) \rfloor + 3$ qubits \\
\hline

\multirow{3}{3.2cm}{\centering $(k-1)$ phase qubits}
& $(r2^{k+1} - 4r + k + 1) \times 2$ CNOTs
(directional control of payload rotations in HWP calls) \\
\cline{2-2}
& $2(k-1)$ CNOT fanouts with target size
$2\lfloor \log_2(L^2) \rfloor + 3$ qubits \\
\cline{2-2}
& $(k-1)$ Toffoli gates, one phase rotation, and a catalyst of size $(k-1)$
for synthesizing the generalized phase gradient comprising phase fix-up
$R_z$ rotations \\
\hline
\end{tabular}
\caption{Additional resources required to enable directional phase kickback for QPE on the plaquette Trotterized time evolution operator using $k$ total phase qubits.
}
\label{tab:control_cost_breakdown}
\end{table}
}

\subsection{Reducing Clifford costs from the fermionic swap network}\label{subsec:fswap}

The Fermi–Hubbard model is naturally expressed in terms of fermionic operators, as shown in \cref{eq:fhm}. To simulate this system on a quantum computer, it must be mapped to qubits via a fermion-to-qubit encoding.
We apply the Jordan-Wigner encoding \cite{JordanWigner1928, Nielsen2005TheFC}, in which fermionic hopping operators are mapped to Pauli strings with non-local parity:
\begin{equation}
c_j^\dagger c_k
=
\frac{1}{4}
\left(X_j - i Y_j\right)
\left(\prod_{\ell=j+1}^{k-1} Z_\ell \right)
\left(X_k + i Y_k\right),
\quad (j < k).
\end{equation}
A key implication of this encoding is that locality depends on the ordering of fermionic modes rather than their geometric arrangement.
As a result, the geometrically local hopping terms in the Fermi-Hubbard model  can map to high-weight Pauli operators if the corresponding modes are not adjacent in the Jordan–Wigner ordering.
Instead of directly implementing these high-weight Pauli operators, our compilation scheme uses fermionic swap (fSWAP) gates to dynamically reorder modes \cite{campbell2022early, kan2024resource-optimized}.
Interacting fermionic sites become adjacent in the Jordan-Wigner ordering thus localizing the hopping term before implementation.
Although these operations are Cliffords and therefore negligible in a T-count cost model, a naive fSWAP implementation can dominate the active volume of a single Trotter step, necessitating careful optimization.

\subsubsection{Jordan-Wigner ordering strategy}

\begin{figure}[h!]
\centering
\includegraphics[width=0.8\linewidth]{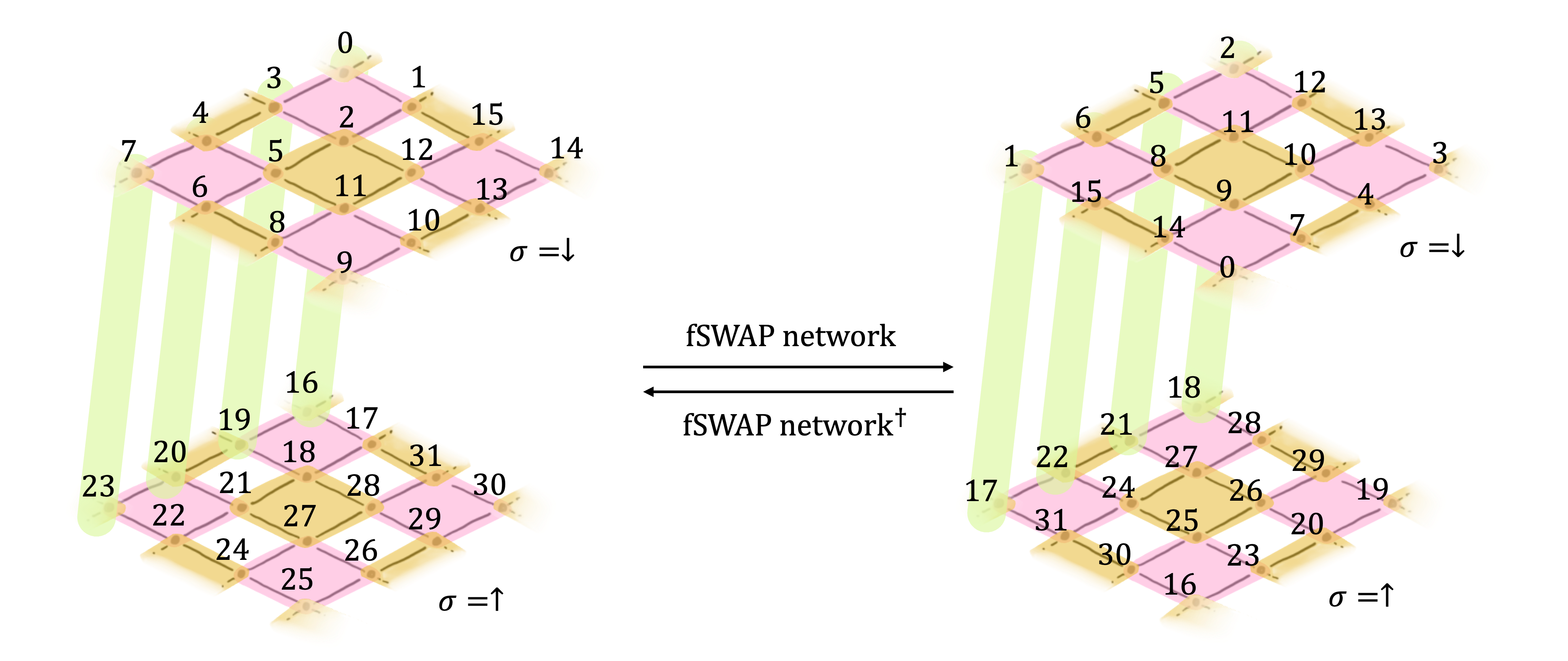}
\caption{Jordan-Wigner ordering of fermionic modes onto the 2D square lattice. The pink and gold highlighted squares are the two disjoint subgraphs of the hopping terms, the green highlight corresponds to the interaction between spin sectors on a given lattice site. The spin down sector is enumerated with the first $L^2$ qubits and the spin up sector is enumerated with the second $L^2$ qubits. The left-hand side depicts the initial enumeration where the pink plaquettes are localized: sites are labeled with consecutive modes around the plaquette.
To simulate the gold terms, a fermionic swap network is used to rearrange the fermion modes on the qubits so that the gold plaquettes are localized.
}
\label{fig:enumeration}
\end{figure}

The Jordan-Wigner ordering of fermionic modes on the 2D lattice is a free parameter in the compilation strategy.
In the Jordan-Wigner representation, the interaction operator does not acquire a parity string.
We therefore assign the first $L^2$ fermionic modes to the spin-down sector and the second $L^2$ modes to the spin-up sector, such that fSWAP operations occur within each sector independently.
This is sometimes referred to as the `block' ordering \cite{arute2020observationseparateddynamicscharge, Stanisic_2022, alam2025programmabledigitalquantumsimulation}.
In this scheme the interaction term acts between qubits separated by a distance of $L$ in the Jordan-Wigner ordering.
On sparse-connectivity architectures such as heavy-hex this leads to substantial routing overhead, motivating an interleaved or pair-interleaved spin ordering \cite{Chowdhury_2026, hartnett2026fastaccuratehighresolutionsimulation}.
However, the non-local CNOT gates from the shifted interaction term do not incur additional cost over local CNOT gates in an active volume architecture.

For a single spin sector, previous works have explored optimal Jordan-Wigner orderings for NISQ-era simulation algorithms and specific architectures \cite{Chiew_2023}. 
Optimality in this context strongly depends on the circuit structure and architecture connectivity.
Recall we implement the fermionic hopping terms in ``pink'' and ``gold'' groups.
While there exist enumerations that localize each set of hoppings individually, no single enumeration simultaneously localizes both.
Within each spin sector, we choose the initial Jordan-Wigner ordering to localize the pink plaquettes.
This contrasts with other approaches (i.e. \cite{kan2024resource-optimized}) that adopt a standard serpentine ordering often considered in DMRG \cite{scardicchio2026, abedi2025}.\footnote{This choice may additionally require a one-time fSWAP network to map between the qubit ordering used for initial-state preparation and that used for time evolution, since some preparation methods may favor a different ordering (e.g., MPS-based state preparation).}
As illustrated in \cref{fig:enumeration}, the required fSWAP network is the sequence of swaps that transforms between the pink-localized and gold-localized Jordan-Wigner ordering.

\subsubsection{fSWAP network implementation}

\begin{figure}[h!]
\centering
\begin{subfigure}[c]{0.4\textwidth}
    \centering
    \includegraphics[trim=5 5 5 5,clip,width=\linewidth]{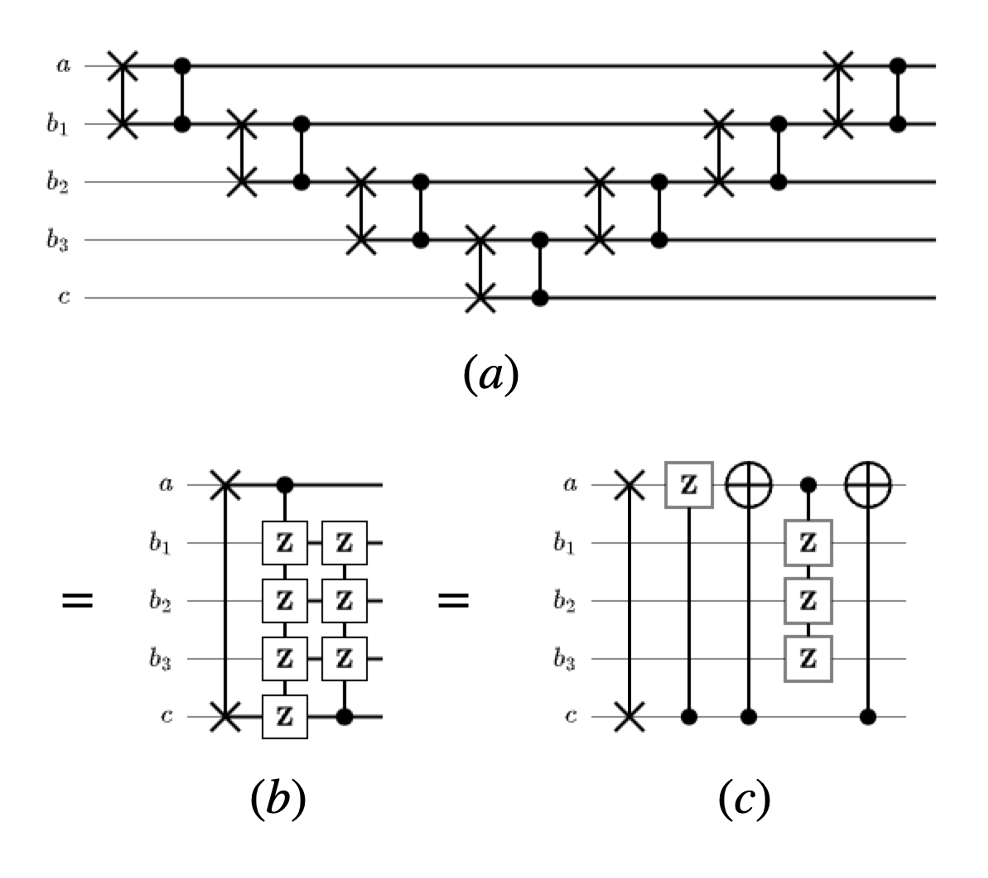}
\end{subfigure}
\hfill
\begin{subfigure}[c]{0.55\textwidth}
    \centering
    \includegraphics[width=\linewidth]{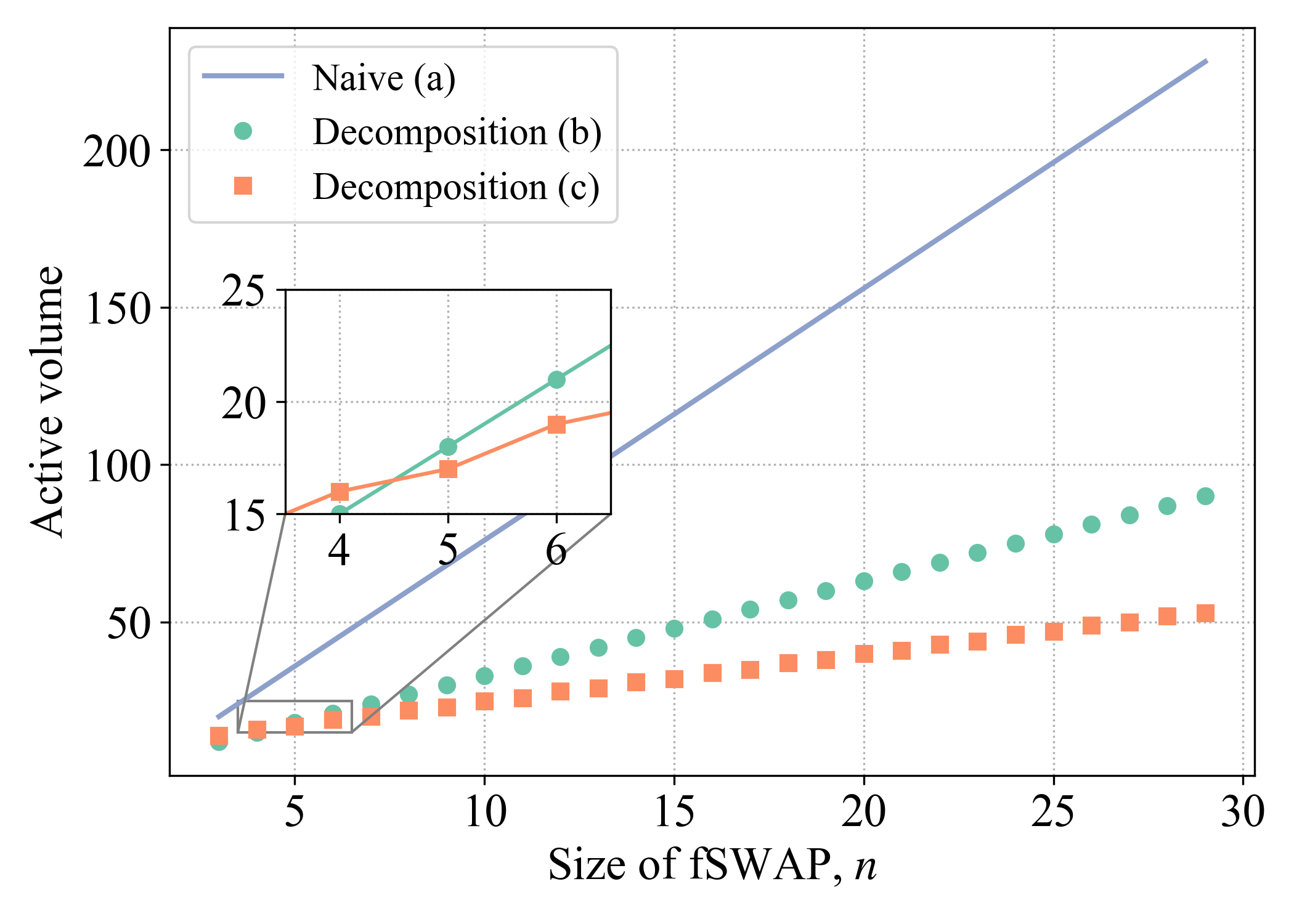}
\end{subfigure}
\caption{Decompositions and costs of non-local fSWAP operations. \textbf{(Left)} An example of a fermionic swap (fSWAP) between non-adjacent modes $a$ and $c = a + n$ where $n=4$. Three equivalent circuit decompositions are shown.
\textbf{(Right)} Comparing active volume costs of different decompositions of a non-adjacent fSWAP of varying sizes $n$ (number of fermions between the two non-local fermions to swap).
This naïve gate decomposition of a fSWAP costs $4(2n -1)$ logical blocks.
The active volume cost of the first decomposition (b) is $\lceil\frac{3}{2}(n-1)\rceil + \lceil\frac{3}{2}n\rceil + 4$ logical blocks. 
The cost of (c) is $\lceil\frac{3}{2}(n-1)\rceil + 2 + 5 + 4$ logical blocks. See \cref{subsubsec:fanouts} and \cref{app:cx_cz} for optimized derivations of the component active volume.
The inset figure shows the crossover point at $n=5$. 
}\label{fig:fswap_comparisons}
\end{figure}

Having fixed the Jordan-Wigner ordering strategy, we now consider the cost of implementing the resulting fSWAP network.
An  individual fSWAP operation consists of a sequence of SWAP gates and controlled-$\mathrm{Z}$ operations.
In an active volume architecture which supports non-local connectivity, 
the SWAP gates need not be physically implemented and can instead be absorbed into a relabeling of qubits.
As a result, the cost of the fSWAP is entirely determined by the required $\mathrm{CZ}$ gates.
For non-adjacent fermionic modes ($i$ and $i+n$ where $n > 1$), an fSWAP is typically realized via a double cascade of SWAPs and CZ gates \cite{seki2022spatial, kan2024resource-optimized, jiang2018quantum}.
Depicted in \cref{fig:fswap_comparisons} (a), intermediate qubits are temporarily displaced and then restored, which results in a symmetric forward-and-reverse cascade.

Active volume is sub-additive so the multi-target gates in decompositions (b) and (c) are cheaper than multiple singly targeted gates, see~\cref{tab:av_costs}.
~\cref{fig:fswap_comparisons} compares the three decompositions, demonstrating that at $n=5$ there is a crossover where if $n \geq 5$ (c) is preferable and (b) is favorable for smaller fSWAP sizes. 
After applying the two decompositions (depending on $n$) an example fSWAP network is given in~\cref{fig:big_summary} \textcolor{blue}{(IIIc)}.
While there are further works studying fermionic routing~\cite{maskara2025fastsimulationfermionsreconfigurable, constantinides2025lowdepthfermionroutingancillas}, after considering the adjustments in our work, the active volume contribution of the fSWAP network becomes a relatively minor component of the total cost of a Trotter step -- see \cref{fig:av_breakdowns} for a breakdown of active volume by subroutine.

\subsection{Optimizing circuit fragments for active volume}\label{subsec:preoptimized_circuits}

One way to estimate the active volume of a subroutine is to sum the cost of each gate operation using lookup tables from Refs.~\cite{litinski2022activevolume, heavey2025improved}.
However, this results in a loose upper bound since the costs are sub-additive.
Tighter estimates can be obtained by optimizing active volume counts for sets of gate operations using simplified oriented ZX (OZX) diagrams as in~\cite{litinski2022activevolume, heavey2025improved}.
In this section, we derive active volume counts for CNOT and CZ fanout operations, replace gate operations in the plaquette terms with equivalent Pauli product rotation (PPR) operations, and group adder segments within Hamming weight phasing.
The resulting optimized costs are summarized in~\cref{tab:av_costs} with detailed derivations provided in~\cref{appen:av-subroutine}.

\begin{table}[H]
\centering
\renewcommand{\arraystretch}{1.1}
\setlength{\tabcolsep}{6pt}
\begin{tabular}{|p{0.4\linewidth}|c|p{0.15\linewidth}|}
\hline
\textbf{Operation} & \textbf{Active volume} & \textbf{Notes} \\
\hline
CNOT fanout ($m$ targets) & $\left\lceil \frac{3}{2}m \right\rceil + 3$ & ~\cref{fig:fanout_av}\textcolor{blue}{a-c} \\
CZ fanout ($m$ targets) & $\left\lceil \frac{3}{2} m \right\rceil + 2$ & ~\cref{fig:fanout_av}\textcolor{blue}{d-f} \\
$YX$ Clifford PPR & 5 & ~\cref{fig:av_yx_pi4_ppr} \\
$XX$ and $YY$ $\pi/8$ PPR & 64 & ~\cref{fig:xx_yy_pprs} \\
CNOT and CZ & 5 & ~\cref{fig:cx_cz} \\
Half-adder compute--uncompute pair & $54$ & ~\cref{fig:hw_adder_av_counts} \\
Full-adder compute--uncompute pair & $74$ & ~\cref{fig:hw_adder_av_counts} \\
Adder-like pair in GPGA & $57$ & ~\cref{fig:phasing_circuit_av_counts} \\
\hline
\end{tabular}
\caption{Summary table for active volume of subroutines used in this work.}
\label{tab:av_costs}
\end{table}

\subsubsection{Reducing the active volume of CNOT and CZ fanouts}\label{subsubsec:fanouts}

Multi-target CNOT and CZ gates, i.e. collections of CNOT or CZ gates sharing a common control qubit but acting on different targets, are widely used in quantum subroutines such as adders and QROMs.
Estimating the active volume of a fanout by summing the cost of individual two-qubit gates gives a loose upper bound.
Naively, a single CNOT gate requires $4$ logical blocks in an active volume architecture~\cite{litinski2022activevolume}, so a CNOT fanout with $m$ targets would be assigned a cost of $4m$ logical blocks.
Instead we derive optimized OZX representations; see~\cref{app:fanout_av}.
The resulting fanout costs remain linear in $m$, but the leading coefficient is reduced from $4$ to $3/2$, with additive constants of $3$ and $2$ for CNOT and CZ fanouts, respectively.
These counts will be used in later sections to estimate the cost of the adders in Hamming weight phasing as well as CZ fanouts used in fermionic swap networks.

\subsubsection{Replacing gate operations with PPRs}\label{subsubsec:PPR}

Pauli product rotations (PPRs) can often substitute several elementary quantum gate operations, e.g. $\exp[\mathrm{i} \phi(Z \otimes Z)]$ replaces two CNOTs and an $R_z(2 \phi)$.
Additionally, active volume of general PPRs with different $X$ and $Z$ weights have been optimized and costed in Ref.~\cite{litinski2022activevolume}.
We leverage these developments to reduce active volume counts for two of the more costly routines in the plaquette evolution term: the evolution under $XX + YY$ and the two-mode FFFT gate.

\paragraph{Evolution under $\boldsymbol{(X_A\otimes X_B)+(Y_A\otimes Y_B)}$:}
We begin by considering the operator $\exp[\mathrm{i} s (X_A \otimes X_B)] \cdot \exp[\mathrm{i} s (Y_A \otimes Y_B)]$, which admits the matrix representation
\begin{align}
    \exp[\mathrm{i} s (X_A \otimes X_B)] \cdot \exp[\mathrm{i} s (Y_A \otimes Y_B)] \cong
    \begin{bmatrix}
        1 & 0 & 0 & 0 \\
        0 & \cos(2s) & \mathrm{i} \sin(2s) & 0 \\
        0 & \mathrm{i} \sin(2s) & \cos(2s) & 0 \\
        0 & 0 & 0 & 1
    \end{bmatrix}.
\end{align}
Circuit decompositions of this operation are shown in~\cref{fig:XXYY_term} using standard gates and in~\cref{fig:XXYY_term_with_ppr} using PPRs.
The PPR decomposition is already substantially cheaper in active volume prior to OZX optimization: excluding rotations, the standard-gate decomposition costs $34$ logical blocks\footnote{Recall that a $\sqrt{X}$ gate and its inverse each cost $5$ logical blocks (see ~\cref{app:root_x_av}), a Hadamard costs $3$, and a CNOT costs $4$ \cite{litinski2022activevolume}.} compared to $10$ blocks for the PPR construction.
We further reduce this cost by decomposing the $Y \otimes X$ Clifford PPR in terms of a joint $Y \otimes X \otimes Y$ PPR and an ancilla qubit prepared in the $\ket{+}$ state and optimizing the corresponding OZX in~\cref{app:yxz_ppm_av}.

\begin{figure}[H]
    \centering
    \begin{subfigure}[b]{0.44\textwidth}
        \centering
        \includegraphics[trim=15 15 15 20,clip,width=\textwidth]{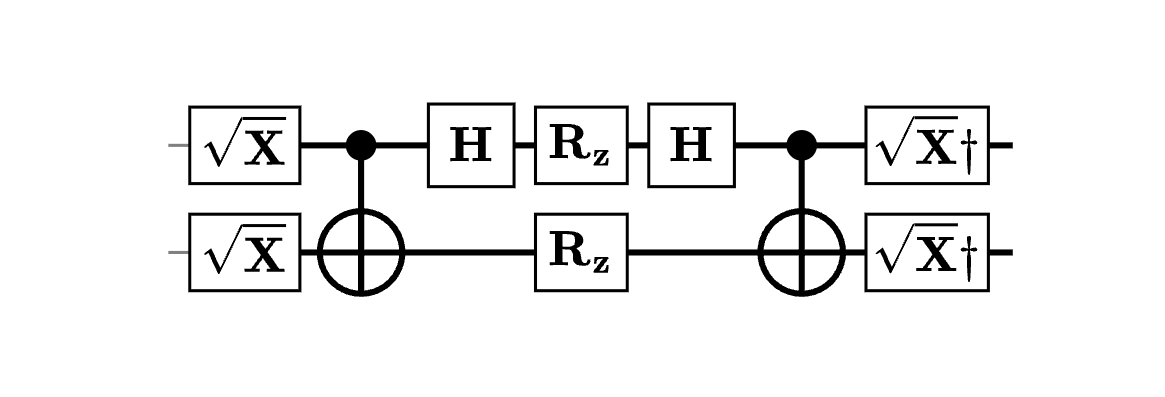}
        \caption{Decomposition using standard gates. Note that each $\sqrt{X}$ gate can be decomposed into $H S H$ gates.}
        \label{fig:XXYY_term}
    \end{subfigure}
    \hfill
    \begin{subfigure}[b]{0.44\textwidth}
        \centering
        \includegraphics[trim=20 20 20 25,clip,width=\textwidth]{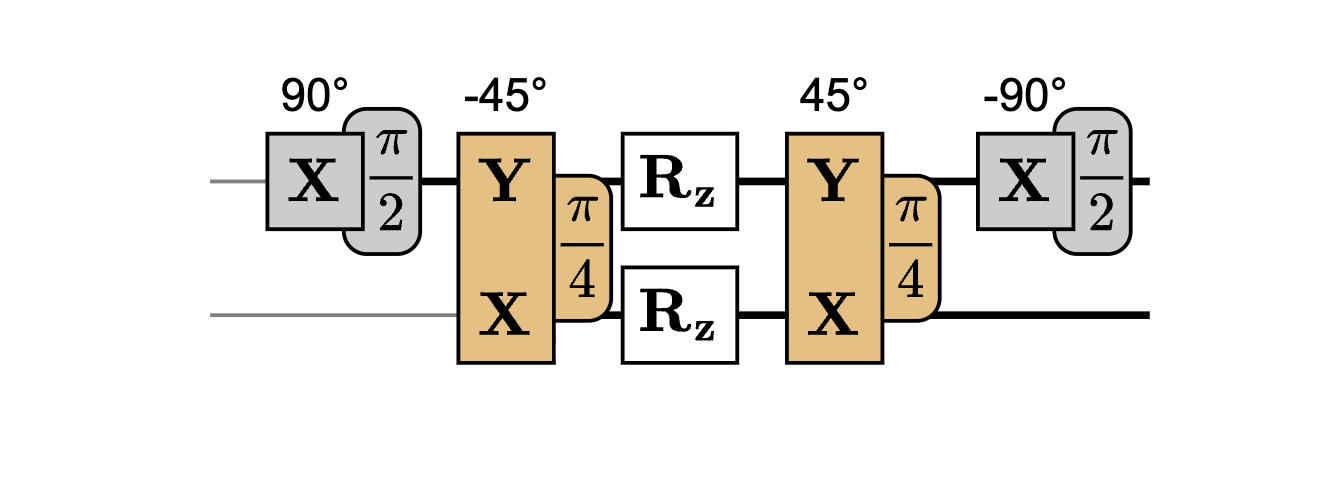}
        \caption{Decomposition using PPRs. The Pauli gates, i.e. rotation angle of $\pm \pi/2$, have no active volume cost.}
        \label{fig:XXYY_term_with_ppr}
    \end{subfigure}
    \caption{Standard and PPR circuit decomposition for the $\exp[\mathrm{i} s(X_A \otimes X_B)] \cdot \exp[\mathrm{i} s (Y_A \otimes Y_B)]$.}\label{fig:expXXYY}
\end{figure}

\paragraph{Two-mode fast fermionic Fourier transform:} Another significant source of non-Clifford cost comes from the two-mode fast fermionic Fourier transform (FFFT) gate.\footnote{Referred to in Ref.~\cite{campbell2022early} as $F$ operators.}
A two-mode FFFT has the following matrix representation in the occupation number basis:
\begin{align*}
    \begin{bmatrix}
    1 & 0 & 0 & 0 \\
    0 & 1/\sqrt{2} & 1/\sqrt{2} & 0 \\
    0 & 1/\sqrt{2} & -1/\sqrt{2} & 0\\
    0 & 0 & 0 & -1
\end{bmatrix}.
\end{align*}
The Clifford+T decomposition of this operation gives an active volume cost of $89$ blocks\footnote{We assume a decomposition of a controlled Hadamard gate using $S, H, T$ gates followed by a CNOT then $T^\dagger, H, S^\dagger$ gates.}.
Our first improvement is to re-express the circuit in terms of PPRs, and conjugate by Cliffords chosen to reduce the active volume cost. 
In particular, PPRs with an odd number of $Y$s require reactive $Y$ operations and are therefore more expensive.
The PPR decomposition we consider in~\cref{fig:fh_plaq_diagonalizing_unitaries} (b) is conjugated with single-qubit Cliffords so that the remaining non-Clifford operations contain an even number of $Y$s.
These inner Cliffords then cancel, giving the circuit in~\cref{fig:big_summary} \textcolor{blue}{(IVb)}. 
We expect the outer Cliffords to also cancel between plaquettes through the fSWAP network, but conservatively retain the active volume cost here.

\begin{figure}[ht]
\centering
\includegraphics[trim=10 10 10 10,clip,width=0.7\linewidth]{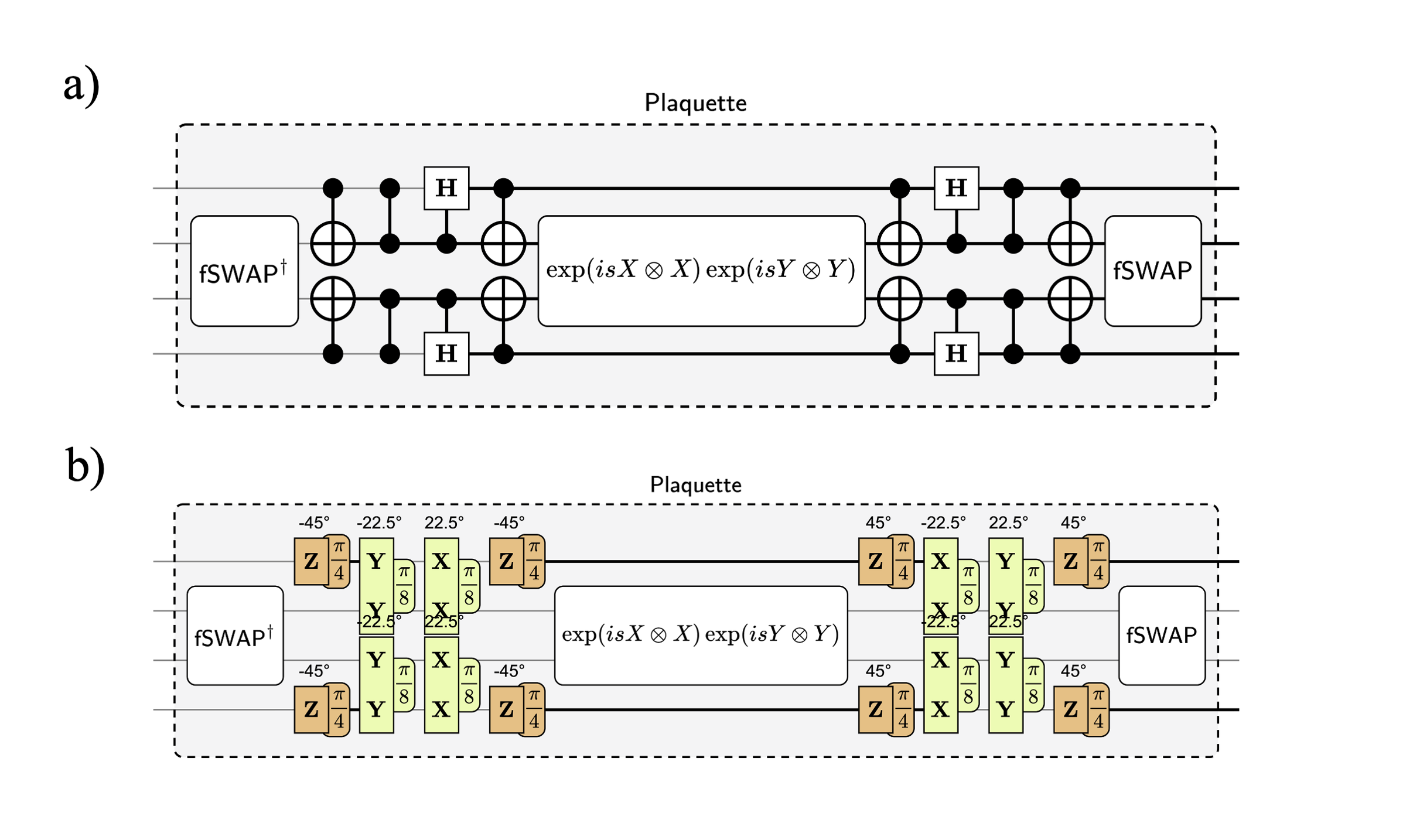}
\caption[]{Active volume optimization of two-mode FFFT within a plaquette evolution.
\textbf{(a)} Standard gate decomposition of two-mode FFFT in a plaquette term;
\textbf{(b)} PPR decomposition of two-mode FFFT, each using two $\pi/8$-PPRs (with even $Y$ weight) and two Clifford PPRs (Cliffords are in orange). These FFFTs are conjugated by single-qubit Cliffords, to guarantee that the non-Clifford operations contain an even number of $Y$ operations and therefore do not need extra blocks for a reactive $Y$. The inner Cliffords cancel resulting in the final decomposition.
} 
\label{fig:fh_plaq_diagonalizing_unitaries}
\end{figure}

After these manipulations, each two-mode FFFT reduces to applying a pair of PPRs with Paulis $X \otimes X$ and $Y \otimes Y$ and angle $\frac{\pi}{8}$.
The decomposition of this operation into Pauli product measurements (PPMs), together with the compilation down to logical blocks is in~\cref{app:zxppms}.
The resulting active volume cost is $64$ logical blocks, including magic-state preparation and reactive measurements.
Adding the outer $\pi/4$ $Z$ PPR gives a total cost of $69$ blocks.

\subsubsection{Using optimized adder segments in Hamming weight phasing}\label{subsub:adders}
In Hamming weight phasing, there are two main subroutines: (un)computing the Hamming weight of the target register and a phasing circuit that is effectively a sequence of adder-like segments with a payload phase operation.
Half (two-bit) and full (three-bit) adders are the building blocks of both subroutines, and the corresponding oriented ZX diagrams have been optimized in Ref.~\cite{litinski2022activevolume}. 
To compute the Hamming weight of an $n$-qubit register, we use $n - \lceil \log_2(n+1) \rceil$ full adders and $\lfloor \log_2(n+1) - w(n)\rfloor$ half adders, for a total of $n - w(n)$ adders \cite{perriello2024quantum}.
Using the optimized constructions reduces the active volume cost of a half-adder compute--uncompute pair from $57$ (1 left elbow including the resource state cost, 1 right elbow and 2 CNOTs) to $54$ blocks and that of a full-adder pair from $85$ (4 CNOTs with two qubit targets, 2 CNOTs, 1 left elbow with resource state, and 1 right elbow) to $74$ blocks.
The adder-like segments appearing in the phasing circuit are similarly reduced from $75$ (2 CNOTs with two qubit targets, 3 CNOTs, 1 left elbow with resource state, and 1 right elbow) to $57$ blocks using the optimized values from Ref.~\cite{litinski2022activevolume}.
The circuit fragments in question and the corresponding optimized OZX diagrams demonstrating these costs are expanded on in~\cref{appen:adder_av}.

\subsection{Batched Hamming weight phasing}\label{subsec:batched_hwp}

Up to this point, our compilation choices have targeted reductions in active volume.
Here, we consider a space-time trade-off in which a modest increase in active volume can substantially reduce the logical-qubit requirement, an important constraint for early-generation quantum computers.
While Hamming weight phasing is gate-efficient, there is a space-time tradeoff as this subroutine requires $O(m)$ ancillae.
To leverage HWP with reduced qubit costs, we consider a batched version \cite{kivlichan2020improved, campbell2022early}, in which instead of treating an entire tower of $m$ equiangular rotations, the tower is split into $\beta$ batches.
Each batch of $m/\beta$ rotations is then treated using HWP.
For simplicity, we consider $\beta$ to be a power of two.
Computing the Hamming weight uses $m/\beta - w(m/\beta)$ ancillae the generalized phase-gradient addition circuit uses an additional $\lfloor \log_2(m/\beta) \rfloor + 1$.
For small numbers of batches, the number of logical qubits can be significantly reduced while minimally increasing the active volume. 

\begin{figure}[H]
\centering
\includegraphics[trim=10 10 10 10,clip,width=0.9\textwidth]{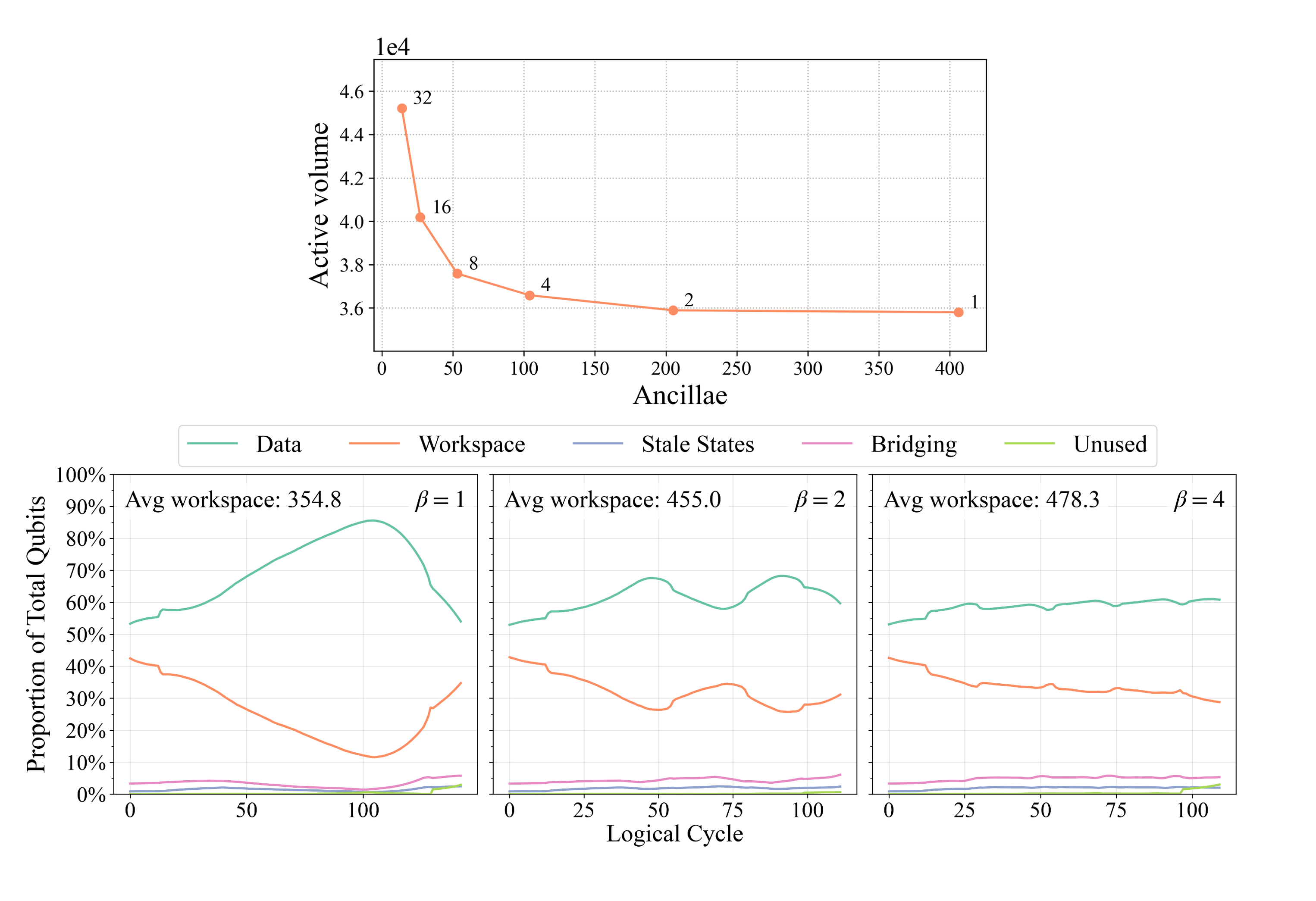}
\caption{Resource counts at $20 \times 20$ versus the number of ancillae in batched Hamming weight phasing for a single interaction term in the Trotterized time evolution operator. 
The total error from rotations is set to $10^{-2}$.
\textbf{Top:} The increase in the active volume and decrease in number of ancillae qubits as we increase the number of batches. The number of batches $\beta$ is annotated next to each point.
\textbf{Bottom:} The partitioning of qubits among the five active volume architecture qubit roles throughout execution of the interaction term using $\beta = \{1,2,4\}$. Results were obtained using the block scheduler of \cite{heavey2026scheduler}. Data qubit (memory) allocation steadily increases during the computation of each batch, which corresponds to a reduction in the workspace size, creating a lemon shape. The height of these lemons decreases and flattens out as we increase the number of batches, which results in a larger average workspace size.
}
\label{fig:batching_demo}
\end{figure}

We demonstrate this by considering the interaction term of Trotterized time evolution unitary, in which we have a tower of $L^2$ rotations conjugated by CNOTs (\cref{fig:big_summary} \textcolor{blue}{(IIIa)}).
For $L=20$, we consider different batch sizes and track the active volume and ancilla requirements in~\cref{fig:batching_demo}.
Increasing the number of batches incurs only a modest active volume overhead while substantially reducing the qubit requirement, particularly at small batch counts.
However, these aggregate metrics by themselves do not help identify the preferred compilation choice: the benefit of the space-time tradeoff will depend on the total machine size and on how qubits are allocated over time, not only the maximum qubit allocation.
In an active volume architecture, the total number of qubits is partitioned into different roles at each logical cycle, distinguishing data and workspace qubits from temporary overheads associated with stale-state storage and bridging~\cite{litinski2022activevolume}, as well as currently unused qubits.
To evaluate this trade-off we consider logical-cycle-level schedules generated by the block scheduler of \cite{heavey2026scheduler}.
Larger batch counts require greater active volume, but they also leave more workspace available on average, allowing more logical work to be performed in each cycle.

The block scheduler’s qubit-usage data allows us to identify the optimal batch size for a given logical-qubit budget.
Let $\Delta AV$ and $\Delta \bar{W}$ denote the percentage changes in active volume and average workspace size respectively, when moving from the currently selected batch size to a candidate batch size.
The candidate is advantageous when
\begin{equation}
    \Delta \bar{W} > \Delta AV 
\end{equation}
i.e. the relative increase in available workspace exceeds the relative active volume overhead.
We demonstrate this empirically for the 1323 logical qubit case in~\cref{fig:optimal_batch_size_vs_computer_size} where $\beta=8$ achieves the highest workspace efficiency and therefore the shortest logical cycle count. 
As the available number of logical qubits increases, this space-time trade-off shifts in favor of smaller batch sizes.
Therefore, although our primary objective is to reduce active volume, selecting among available space-time trade-offs requires proceeding further down the compilation stack and evaluating the logical-cycle schedule under the target machine’s resource constraints.

\begin{figure}[H]
\centering
\includegraphics[width=0.95\textwidth]{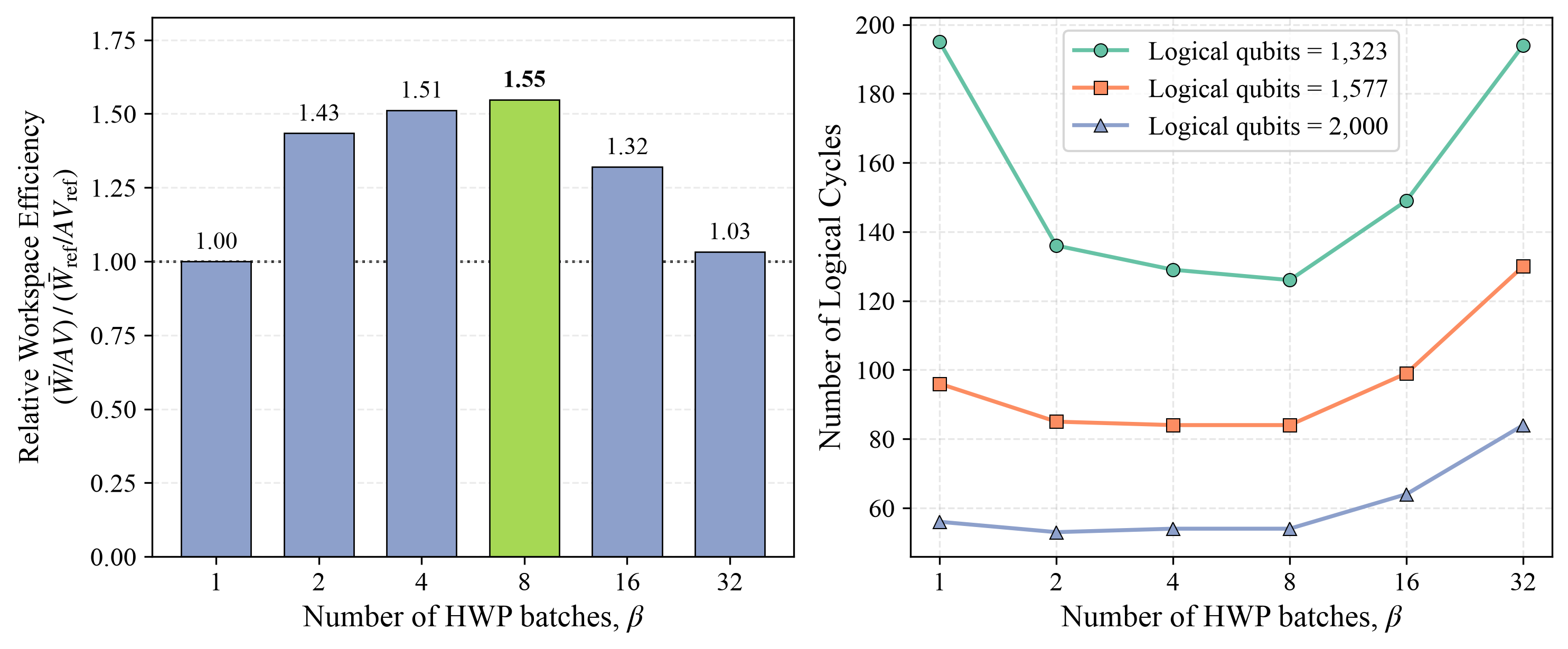}
\caption{
Batch-size optimization for catalyzed Hamming-weight phasing of a single interaction term on the $20\times20$ lattice. \textbf{Left:} For a 1323 logical qubit machine, the ratio of average
workspace size to active volume ($\bar{W}/AV$), normalized to the single-batch implementation, is maximized at $\beta=8$. \textbf{Right:} Logical-cycle counts obtained using the block scheduler of \cite{heavey2026scheduler} for machines with 1323, 1577, and 2000 logical qubits. These qubit counts are arbitrarily chosen but are large enough to run the computation without an out-of-memory error. The preferred batch count decreases as more qubits become available, from $\beta=8$ to $\beta=4$ and $\beta=2$, respectively. For the 1323 qubit case, the variation in the number of logical cycles across batch sizes follows the trend shown in the left panel, supporting
the proposed explanation for why some batch sizes outperform others. 
}
\label{fig:optimal_batch_size_vs_computer_size}
\end{figure}

\subsection{Circuit error analysis}\label{subsubsec:error_budgeting_for_circuit}

The relative energy error of $0.51\%$ represents the total energy error budget, which must be apportioned among the distinct sources of error in the algorithm:
\begin{enumerate}
    \item \textbf{QPE error}, $\epsilon_\text{QPE}$;  We consider the intrinsic error of QPE due to a finite phase register. This error contribution is isolated from other implementation errors by considering an exact time evolution unitary and ground state preparation. 
    This discretization error determines the number of phase qubits required in the algorithm.
    \item \textbf{Time-evolution implementation error}; The error due to approximating the time-evolution unitary $U(\tau)=\mathrm{e}^{\mathrm{i}H\tau}$ by a unitary $\tilde{U}_\tau$. This approximation error has two distinct origins:
    \begin{enumerate}
        \item \textit{Suzuki–Trotter error}; Implementing the time evolution via the second-order product formula described in \cref{eqn:trotter-intro} with the plaquette Hamiltonian partitioning from \cite{campbell2022early} introduces an error that is controlled by increasing the number of Trotter steps, $r$. 
        The error of a single Trotter step of duration $s$ can be bounded in terms of a commutator norm $W$, as shown in~\cite[Proposition 10]{childs2021theory} and \cite[Equation 6]{kivlichan2020improved}:
        \begin{align}\label{eq:2ndtrotter_s}
        \norm{U_{s}  - U(s)} \leq \frac{Ws^3}{r^3}.
        \end{align}
        We note that
        $W$ depends on the system, Hamiltonian splitting, and term ordering
        and can be costly to estimate; we defer the discussion of how this is computed to the next section.
        For the full evolution over time $\tau$ broken into $r$ steps, this yields
        \begin{align}\label{eq:2ndtrotter}
        \norm{\left(U_{\tau/r} \right)^r - U(\tau)}& \leq r \norm{U_{\tau/r} - U(\tau/r)}\leq r W\left(\frac{\tau}{r} \right)^3 = \frac{W\tau^3}{r^2},
        \end{align}
        where $(U_{\tau/r})^r$ describes the Suzuki-Trotter approximation of the time evolution operator at time $\tau$ with otherwise perfect gate synthesis.
        \item \textit{Rotation synthesis error}; The circuit contains arbitrary rotation (payload) gates appearing from the HWP subroutine.
        We implement the rotation synthesis using 
        repeat-until-success mixed fallback
        rotation synthesis of \cite{shorter_circs}.
        For a diamond norm difference of $\Delta$ in the rotation gate, the estimated mean number of T-gates required per arbitrary rotation \cite[Table 1]{shorter_circs} is
        \begin{equation}\label{eqn:rottoffs}
            \Delta<10^{-4}: \qquad N_{\mathrm{rot-T}}(\Delta) = 0.53\log_2(1/\Delta) + 4.86.
        \end{equation}
        We assume T-gates are catalyzed using CCZ states \cite{gidney2019efficient}.\footnote{The implied 1:2 conversion rate in resource states has been used in prior resource estimates for Fermi-Hubbard \cite{campbell2022early, kan2024resource-optimized} as a reasonable cost proxy. However, in practice, this ratio will depend on the architecture and the types of magic states distilled.
        We use this cost proxy for the error analysis but assume distillation of both Toffoli and T states when counting the active volume.
        }
        Given $\beta$ batches and $r$ Trotter steps there are $\beta(4r+1)$ arbitrary rotations in one query. 
        The total unitary error given each rotation is implemented to error $\Delta$ is given by repeated use of the triangle inequality:
        \begin{equation}\label{eqn:splitting_rot_error}
            \|(U_{\tau/r})^r - (\tilde{U}_{\tau/r})^r \|\leq  \beta(4r+1)\Delta.
        \end{equation}
    \end{enumerate}
    
    Using the triangle inequality, the total unitary error can be bounded by the sum of two unitary errors
    \begin{subequations}
    \begin{align}
    \|U(\tau) -(\tilde{U}_{\tau/r})^r \|
    &= \|U(\tau) - (U_{\tau/r})^r + (U_{\tau/r})^r -(\tilde{U}_{\tau/r})^r \| \\
    &\leq \|U(\tau) - (U_{\tau/r})^r\|
    + \|(U_{\tau/r})^r - (\tilde{U}_{\tau/r})^r \| .
    \end{align}
    \end{subequations}
    The total unitary error must subsequently be translated into an energy error, which we will denote by $\epsilon_U$. We will later apportion this energy error into two pieces, an energy error budget resulting from the Suzuki-Trotter approximation, $\epsilon_\text{TS}$ and an energy error budget for imperfect rotation synthesis, $\epsilon_\text{rot}$.
    \item \textbf{QPE implementation/synthesis error}; Independent of the approximation of the time evolution $\tilde{U}$ described above, the third contribution corresponds to the approximation of the QPE circuit itself.
    If the perfect QPE circuit for unitary $U$ is denoted $V_{\text{QPE}}$, then the QPE circuit with implementation error is denoted $\tilde{V}_{\text{QPE}}$.
    At this stage, we neglect errors associated with state preparation, the phase state window preparation, and the inverse QFT.
    Under these assumptions, the implementation error in QPE stems from rotation synthesis in the preparation of the catalyst state.
    
    While we envisage using the same repeat-until-success synthesis method as \cref{eqn:rottoffs}, the catalyst rotations are given an independent error contribution.
    Recall from \cref{subsec:hwp} that the number of catalyst qubits for the `PIG' ordering with $\beta$ batches of HWP is given by $2\lfloor \log_2(L^2/\beta)\rfloor + 3 $ where we apportion the same error to each rotation as in \cref{eqn:splitting_rot_error}.
    As with the previous contributions, this unitary error is later mapped to an energy error, which we denote by $\epsilon_{\text{cat}}$.
\end{enumerate}

\subsubsection{Error budgeting}
Combining all three contributions, the total energy error must satisfy
\begin{align}
\epsilon \leq \epsilon_{\text{QPE}} + \epsilon_{U} + \epsilon_{\text{cat}} .
\end{align}
To relate the corresponding unitary implementation errors to deviations in the estimated energy, we use an effective-Hamiltonian analysis, which bounds eigenenergy shifts in terms of the operator-norm distance between the ideal and approximate unitaries~\cite{bhatia1984bound,kivlichan2020improved}.
The resulting bounds determine the allowable time-evolution and catalyst-synthesis errors for a given allocation of the total energy budget; the derivation is given in~\cref{appen:error_budgets}.

A resource-efficient allocation of the error budget is obtained by minimizing the Toffoli count while constraining the sum of error contributions to be within the total error budget.
We further decompose the time-evolution error contribution as
\begin{align}
\epsilon_U \approx \epsilon_{\text{ST}}+\epsilon_{\text{rot}},
\end{align}
where $\epsilon_{\text{ST}}$ and $\epsilon_{\text{rot}}$ denote the contributions from Suzuki--Trotter and rotation-synthesis errors, respectively\footnote{This is a first-order approximation that allows a linear attribution of energy error to each source. The approximation and its relation to the underlying unitary-error bounds are detailed in the appendix.}.
This gives the working error budget
\begin{align}\label{eqn:error_budgetmain}
\epsilon \geq
\epsilon_{\text{QPE}}
+\epsilon_{\text{ST}}
+\epsilon_{\text{rot}}
+\epsilon_{\text{cat}}.
\end{align}

We optimize the allocation of~\cref{eqn:error_budgetmain} jointly with the evolution time $\tau$.
The allocation is parameterized by ${\delta_{\text{QPE}},\delta_{\text{ST}},\delta_{\text{rot}}}$, where each $\delta_i\in(0,1)$ specifies the fraction of the remaining error budget assigned to the corresponding error source; by construction, this parameterization ensures the error inequality.
The evolution time is optimized subject to an upper bound $\tau_{\max}$, which we choose to be less restrictive than the conservative no-aliasing bound $2\pi/\lVert H_{\mathrm{FH}}\rVert_1$.
Since classical estimates already determine the ground-state energy to a coarse precision~\cite{leblanc2015solutions,sorella2023systematically,Levy2024,roth2025superconductivity,liang2026investigating,Viteritti2026}, this prior information can be used to increase $\tau$, with the quantum algorithm only required to resolve the remaining uncertain bits.
Increasing $\tau$ reduces the QPE error at the cost of increased Trotter error.
We conservatively assume a classical precision of $5\%t$ per site; further details and justification are given in~\cref{appen:error_budgets}.
For a given allocation, the optimization parameters determine the principal circuit parameters: the number of QPE queries $N$, the number of Trotter steps $r$, and the synthesis precisions $\Delta_{\text{rot}}$ and $\Delta_{\text{cat}}$ of the payload and catalyst rotations.
The optimization procedure is summarized in~\cref{alg:opt_sim}.

\begin{algorithm}[h!]
\caption{Resource-optimized error budgeting }
\renewcommand{\arraystretch}{1.2}
\label{alg:opt_sim}
\begin{algorithmic}[1]
\Require{System size $\{L_x,L_y\}$, kinetic coefficient $t$, potential coefficient $u$, total error budget $\epsilon = 0.0051 t L_x L_y$ and number of batches for HWP $\beta$. }
\vspace{1mm}
\State Start with an initial guess for the error budgeting parameters and evolution time $\{\delta_\text{QPE}, \delta_\text{ST}, \delta_\text{rot}, \tau\}$ where $\delta_i \in (0, 1)$ and $\tau \in (0,\tau_\text{max})$. From the $\delta$'s compute the energy error allocation, $\epsilon$'s.
\State Compute circuit parameters: number of queries $N$, number of trotter steps $r$, error per rotation $\Delta_\text{rot}$, error per catalyst rotation $\Delta_\text{cat}$:
\begin{align*}
N(\epsilon_\text{QPE}, \tau) &= 2^{\ceil{\log_2(\pi/\arctan(\epsilon_{\text{QPE}}\tau))}-1}
\\
r(\tau, \epsilon_\text{TS}, u, t, L_x, L_y) &= \Bigg\lceil \sqrt{\frac{W_\text{PLAQ}(L_x,L_y,u,t) \tau^3}{2\sin(\epsilon_{\text{TS}}\tau/2)}}\Bigg\rceil\\
\Delta_\text{rot}(\epsilon_\text{rot},\beta, \tau, r)&= \frac{2}{\beta(4r+1)}\sin\left(\frac{\epsilon_\text{rot}\tau}{2} \right)\\
\Delta_\text{cat}(\epsilon_\text{cat},\beta, \tau, L_x, L_y)& = \frac{\sin(\epsilon_\text{cat}\tau/2)}{\lfloor \log_2(L_xL_y/\beta)\rfloor} + 3/2.
\end{align*}
\State Evaluate an estimate for the Toffoli count for a single query of the time evolution unitary:
\begin{align*}
\mathrm{Toff}_{\text{FFFT}} &= (4r+2)L_xL_y \\
\mathrm{Toff}_{\text{HW}} &= \beta(4r+1)\left(\frac{L_xL_y}{\beta} + \lfloor \log_2(L_xL_y/\beta)\rfloor - w(L_xL_y/\beta) + 1\right) \\
\mathrm{Toff}_{\text{rot}} & = \beta (4r+1)(0.265\log_2(1/\Delta_\text{rot}) + 2.43)\\
\mathrm{Toff}_{\text{cat}} & = (2\lfloor \log_2(L_xL_y/\beta)\rfloor + 3) \times (0.265\log_2(1/\Delta_\text{cat}) + 2.43)
\end{align*}
\Statex Update initial parameters to minimize over this cost function
\[
\min_{\delta_\text{QPE}, \delta_\text{ST}, \delta_\text{rot}, \tau} \quad
N\times(\mathrm{Toff}_{\text{FFFT}} + \mathrm{Toff}_{\text{HW}} + \mathrm{Toff}_{\text{rot}} ) + \mathrm{Toff}_{\text{cat}}
\]
\Statex \textbf{Output:} Optimal parameters $\{\delta_\text{QPE}^\star, \delta_\text{ST}^\star, \delta_\text{rot}^\star, \tau^\star\}$ and corresponding chosen circuit parameters $\{ N, r, \Delta_\text{rot}, \Delta_\text{cat}\}$.
\end{algorithmic}
\end{algorithm}

We use Toffoli count as a cost metric for the error budgeting because it admits simple analytic estimates, enabling fast evaluation during the optimization.
A more representative metric, such as active volume, would require explicit circuit compilation and detailed architectural considerations, making it significantly more 
computationally intensive.
As indicated by the Clifford/non-Clifford breakdown (\cref{fig:av_breakdowns}), the Toffoli count remains a reasonable proxy\footnote{Additionally for the error budgeting we ignore merging between unitary calls to simplify the cost function and multiply by the number of queries.} for the error budget.
Estimates for the Toffoli cost of the various circuit components are taken from \cite{kan2024resource-optimized}.
The output of the optimizer for all lattice sizes is reported in~\cref{appen:qre_data}.

\subsubsection{Computing the Trotter error}\label{sect:pig_trotter_error}

Changing the Trotter ordering requires recomputing the commutator bound as previous works used the `IPG' ordering.
The usual second order commutator bound formula is given by
\begin{equation}\label{eqn:trotterpig}
            \norm{\mathrm{e}^{\mathrm{i} (H_h + H_I)\tau}-\mathrm{e}^{\mathrm{i}  (\tau/2)H_h^p}\mathrm{e}^{\mathrm{i}  (\tau/2)H_I}\mathrm{e}^{\mathrm{i}  \tau H_h^g}\mathrm{e}^{\mathrm{i} (\tau/2)H_I}\mathrm{e}^{\mathrm{i}  (\tau/2)H_h^p}}\leq W_{\text{PLAQ}}\tau^3.
\end{equation}
\begin{equation}
    W_{\text{PLAQ}} = \frac{1}{12} \sum_{b=1,2}\left\|\sum_{c>b, a>b}\left[\left[H_b, H_c\right], H_a\right]\right\| +\frac{1}{24} \sum_{b=1,2}\left\|\sum_{c>b}\left[\left[H_b, H_c\right], H_b\right]\right\|,
\end{equation} \label{eq:W_PLAQ}
with $H_1 = H_h^p$, $H_2 = H_I$ and $H_3 = H_h^g$. 
We investigate two methods of computing $W$ since the tight bound is prohibitively expensive as involves computing the spectral norm of Hilbert-space-sized matrices. 

The first method is inspired by the work of~\cite{Schubert.2023}.
The key idea is to expand the above nested commutators and then group the resulting expressions into clusters acting on at most $C$ sites.
We choose $C$ to be small enough that we can exactly diagonalize each of the clusters and thus compute their spectral norm.
We then upper bound the overall spectral norm by applying triangle inequalities between clusters.
This is a slight adaptation from Ref.~\cite{Schubert.2023} which bounds the commutator norms for a sub-lattice of terms which, upon translations, yield bounds on the overall sums.
In contrast, we optimize over partitioning the full, expanded commutator sums into as few clusters as possible to minimize the use of triangle inequalities.

The second approach uses the algebraic structure of the Hamiltonian to express the double commutators in terms of operators whose norms can be bounded directly.
This generalizes the result of \cite{campbell2022early} (which was for the `IPG' Trotter ordering) and allows large portions of the norm calculation to be performed analytically.
Since fewer triangle inequalities are introduced than in the clustering method, the bound is expected to be slightly tighter, as observed in~\cref{fig:wplaq_plot}, and is therefore used in our numerical resource estimates.
A detailed derivation is given in~\cref{appen:trotter_error}.
Although the `IPG' ordering still gives lower values for $W_\text{PLAQ}$, our improved bound leaves only a small gap between the two orderings.
This modest increase is more than compensated by the substantially larger compilation savings obtained by reordering the components.

\begin{figure} 
\centering
\includegraphics[width=0.6\linewidth]{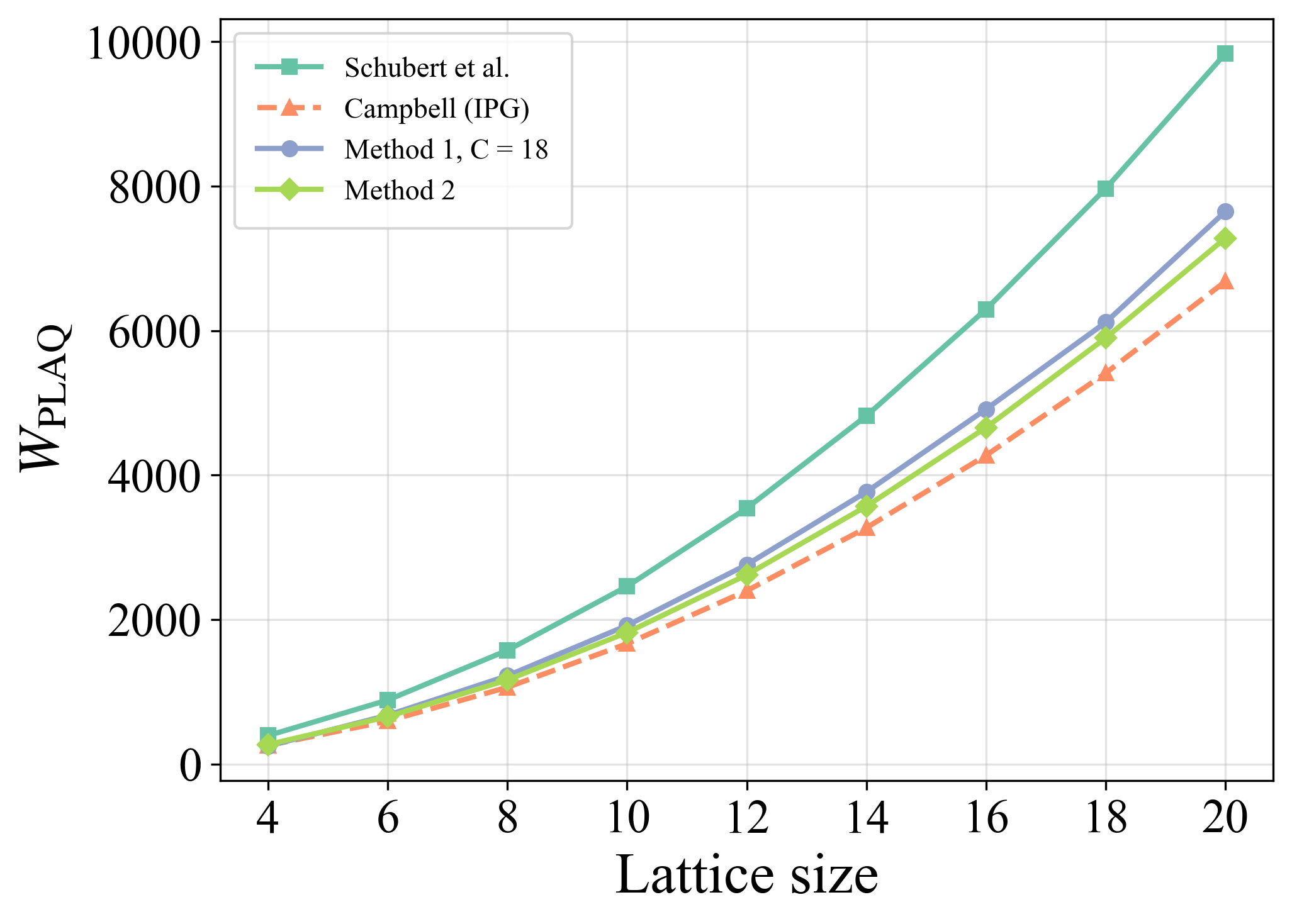}
\caption{Comparison of Trotter commutator bounds for $W_{\mathrm{PLAQ}}$ for $u=8$ and $t=1$.
We introduce: Method~1, an extended clustering approach, and Method~2, a generalized algebraic approach and compare against prior work.
Campbell is available only for the IPG ordering \cite{campbell2022early}, whereas Schubert \cite{Schubert.2023}, Method~1, and Method~2 can be directly compared, computing $W_{\mathrm{PLAQ}}$ for the PIG ordering.
}
\label{fig:wplaq_plot}
\end{figure}

\section{Future directions and conclusions}\label{sect:conclusions}

In this work, we presented the first active-volume-based compilation of the Fermi-Hubbard ground-state energy estimation algorithm for an early generation FTQC.
We confirmed that conventional circuit volume increasingly overestimates the execution cost on architectures with some non-local connectivity and showed that directly optimizing for spacetime cost yields substantial additional reductions beyond those achieved through T-count optimization alone.
These improvements arise from optimizations across multiple levels of the quantum software stack, ranging from algorithmic choices, through circuit-level compilation such as Trotter reordering to architecture-aware logical operation synthesis using OZX-diagram compression.
Although developed for active volume architectures, we emphasize that many of our compilation techniques are broadly applicable to reducing resources on general surface-code-based architectures.
Combined with recent advances in active volume scheduling~\cite{heavey2026scheduler}, these results substantially reduce the estimated runtime of large-scale Fermi--Hubbard simulation on early fault-tolerant quantum computers. 
Although the absolute runtime values depend significantly on the hardware parameters, we observe the expected runtime trends, both in the comparison between the baseline and improved circuits and in the impact of increasing the number of workspace qubits.
We conclude by outlining several directions for future work.

The largest opportunity for further reducing the active volume is likely to come from tighter estimates of the Trotter error.
The resource estimates presented in this work rely on rigorous operator-norm error bounds, which are expected to be highly conservative for the low-energy subspace relevant to ground-state energy estimation.
In addition, the commutator bounds used to estimate the Trotter error introduce further pessimism.
While NISQ experiments typically estimate Trotter error through numerical convergence against classical simulations, such approaches are difficult to scale to large lattices and do not provide the rigorous error guarantees associated with FTQC.
Developing substantially tighter, yet reliable, error estimates for the low-energy sector relevant to ground-state energy estimation would therefore have a direct and potentially dramatic impact on the cost of quantum simulations.
Further reductions in active volume could be achieved by reformulating 
the error-budget optimization to directly minimize active volume rather than Toffoli count, while relaxing the restriction on the number of QPE queries may provide a smoother optimization landscape through the controlled implementations proposed in Refs.~\cite{lee2021even,dutkiewicz2025error}. 
At the circuit level, additional reductions may be obtained through improved fermionic swap network constructions that reduce the number of fSWAP operations~\cite{constantinides2025lowdepthfermionroutingancillas} while exposing larger circuit fragments to ZX-diagram optimization.

Finally, as architecture-aware compilation continues to mature, further reductions in execution time are likely to require closer co-design between circuit compilation and runtime scheduling.
While active volume provides an effective optimization objective, as with T-gate count, minimizing active volume alone does not necessarily minimize execution time.
For example, aggressive OZX optimization can merge many operations into large active volume blocks which, despite having lower active volume, are more difficult to pack efficiently within the limited workspace available during scheduling.
Future compiler optimizations should therefore be developed alongside scheduling strategies to ensure that reductions in active volume translate into corresponding reductions in runtime.

\section*{Code availability}
We use PsiQuantum's open-access software Workbench to program quantum circuits and generate quantum resource estimates on the level of Toffoli gates, logical qubits, and active volume \cite{psiquantum_qdk}. The code to instantiate the specific circuits described in this paper is publicly available in the GitHub repository \cite{PsiQFHCodeRepo}.

\section*{AI use statement}
The scientific ideas in this manuscript originate from the authors.
Code to implement algorithmic subroutines was developed by the authors.
The manuscript was drafted and revised by the authors.
An AI chat tool (ChatGPT) was used in the review process of the manuscript and (Cursor) to clean up code documentation, remove unused code, follow code standards as well as to help write presentational plotting functions.

\section*{Author contributions}
The author list is ordered alphabetically. 
HA and SS contributed equally to this work and co-led the project and manuscript preparation.
HA, AK, JL, SP, SS, MS, RL contributed to the conception of the work. JL was the broader project manager for lattice simulations.
HA, AC, SH, AK, JL, SP, JP, WP, SS, MS, GU designed the quantum circuits to reduce the active volume. HA, CH, SS, WS, GU implemented the circuits in Workbench. 
AC and SH ran the block scheduler for runtime estimates.
All authors contributed to the writing and editing of the
manuscript.

\section*{Acknowledgments}
We thank Modjtaba Shokrian Zini, Zheting Jin, and Kasra Hejazi in particular for useful discussions and Cristian L. Cortes for early code contributions.
The quantum circuit diagrams in this work were generated using Circuit Designer (freely available at 
\url{https://construct.psiquantum.com/cd/}), and the quantum circuits were implemented and their resources were counted using Workbench 
(\url{https://construct.psiquantum.com/qdk})
from PsiQuantum's Construct software suite.

\printbibliography

\appendix
\addtocontents{toc}{\protect\setcounter{tocdepth}{1}}
\crefalias{section}{appendix}
\crefalias{subsection}{appendix}

\section{Deriving the active volume of a subroutine}\label{appen:av-subroutine}
In this section we outline how the active volume cost of an operation can be derived for a surface code active volume architecture that uses lattice surgery, dislocations, $Z$ basis measurements, and $\ket{0}$ state initialization as primitive operations, as was considered in Ref.~\cite{litinski2022activevolume}.

Each logical qubit is encoded as a $d \times d$ surface-code patch that undergoes $d$ rounds of syndrome measurements per logical cycle, defining a ($d \times d \times d$) spacetime block.
A fault-tolerant circuit can therefore be represented as a network of such blocks, with the active volume determined by the number of spacetime blocks required to realize the computation.

To derive the cost of a subroutine, we first rewrite it as a ZX diagram~\cite{Coecke_2011}.
For circuits written using the \{$S$, $H$, CNOT, $T$\} gate set, where
\begin{equation}
S = 
\begin{pmatrix}
1 & 0 \\
0 & i
\end{pmatrix},
\quad
H = \frac{1}{\sqrt{2}}
\begin{pmatrix}
1 & 1 \\
1 & -1
\end{pmatrix},
\quad
T =
\begin{pmatrix}
1 & 0 \\
0 & e^{i\frac{\pi}{4}}
\end{pmatrix},
\quad
CNOT =
\begin{pmatrix}
1 & 0 & 0 & 0 \\
0 & 1 & 0 & 0 \\
0 & 0 & 0 & 1 \\
0 & 0 & 1 & 0
\end{pmatrix}
\end{equation}
one can use the identities provided in~\cref{fig:gate_identities}.
We require the ZX diagram to be phaseless so that it maps directly to the supported surface code scheme.
Phase gates are therefore implemented by magic-state injection rather than by phased spiders: 
$\ket{Y} = \frac{1}{\sqrt{2}}(\ket{0}+\mathrm{i}\ket{1})$ and $\ket{T} = \frac{1}{\sqrt{2}}(\ket{0}+\mathrm{e}^{\mathrm{i}\frac{\pi}{4}}\ket{1})$ magic states supply the phases required for the $S$ and $T$ gates, respectively.
These magic states must be created in magic state factories and are used to execute gates such as Toffoli gates, elbows, and Pauli product rotations; see Ref.~\cite{litinski2022activevolume}.

\begin{figure}[h!]
    \centering
    \includegraphics[width=0.65\linewidth]{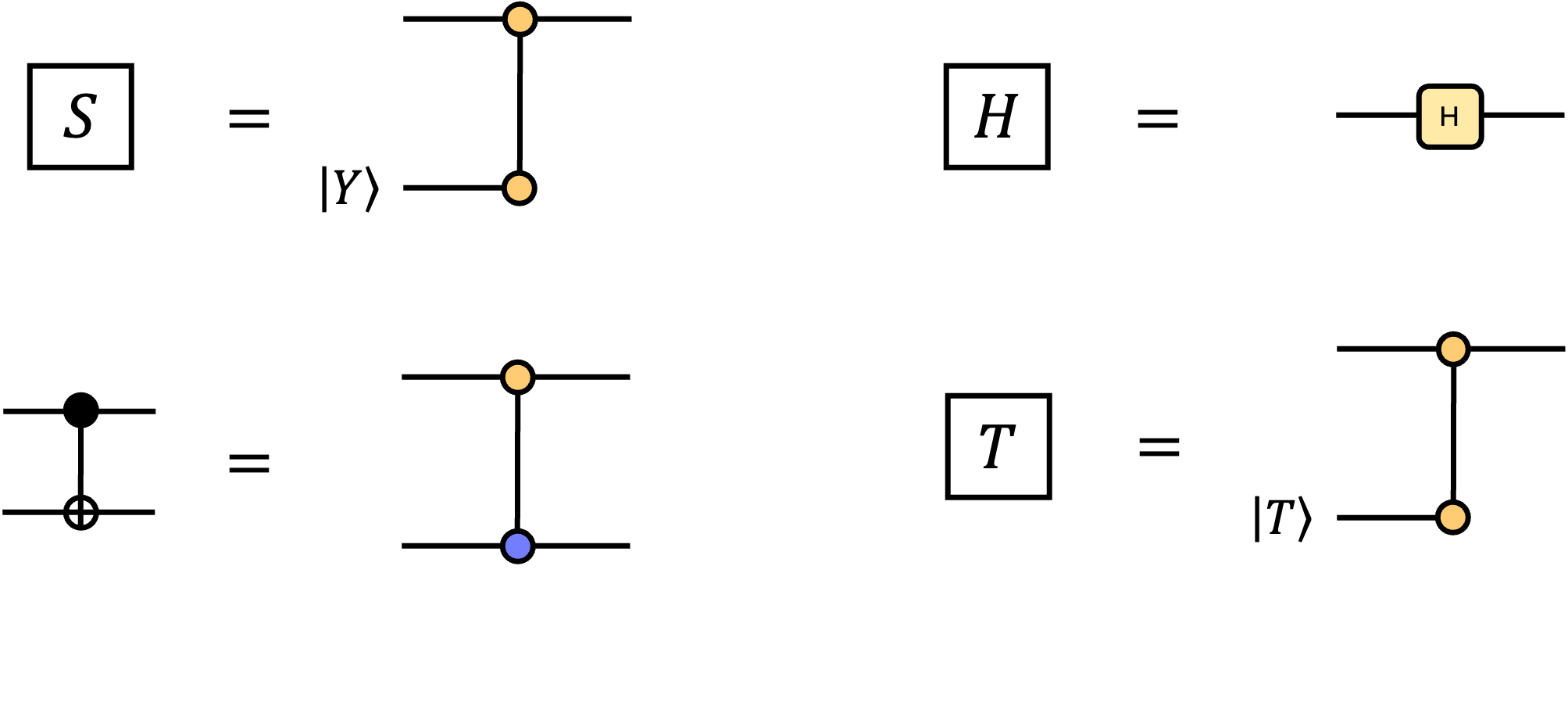}
    \caption{Phaseless ZX identities for the \{$S$, $H$, CNOT, $T$\} gate set. Orange (blue) spiders are Z (X) type. All spiders are phaseless with $\ket{Y}$ and $\ket{T}$ state injection supplying the phase for the $S$ and $T$ gates, respectively.
}
    \label{fig:gate_identities}
\end{figure}

\subsection{ZX to oriented ZX Protocol}

The ZX diagram specifies the required logical connectivity but not whether that connectivity can be realized by the available faces of each surface-code spacetime block.
We address these geometric constraints by assigning orientations to the ZX diagram to obtain an oriented ZX (OZX) representation.
Once we have a ZX diagram, we convert it into an OZX diagram using the following protocol:

\begin{enumerate}

  \item \textbf{Ports and spacetime faces.}
  For each spider, introduce six labeled ports: North ($N$), South ($S$), East ($E$), West ($W$), Up ($U$), and Down ($D$), representing the six faces of the associated spacetime block~\cite{demarti2024decoding} (\cref{fig:ports}). Every edge between two spiders is assigned to a port, and it must use the same port type on both spiders it connects. Each port of each spider can accommodate at most one edge. If the corresponding block does not provide enough ports to support all of a spider’s edges, you may need to split that spider into multiple spiders.

  \begin{figure}[h!]
      \centering
      \includegraphics[width=0.14\textwidth]{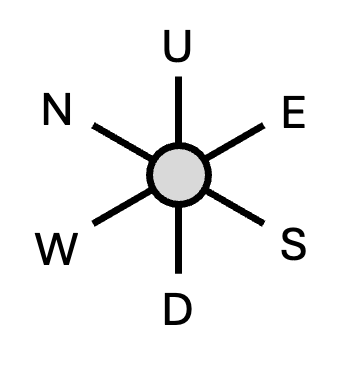}
      \caption{Locations of the North ($N$), South ($S$), East ($E$), West ($W$), Up ($U$), and Down ($D$) ports.}
      \label{fig:ports}
  \end{figure}

  \item \textbf{Spider orientations and syndrome collection.}
  Give each spider an \emph{orientation} by disabling exactly one pair of opposite ports, chosen from $N/S$, $E/W$, or $U/D$ (\cref{fig:ozx_examples}). The disabled pair is interpreted as the pair of edges reserved for syndrome collection in lattice surgery, ensuring fault-tolerant syndrome extraction~\cite{horsman2012surface, demarti2024decoding}. Concretely:
  \begin{itemize}
    \item An $N$-oriented spider has no connections on its $N/S$ ports.
    \item An $E$-oriented spider has no connections on its $E/W$ ports.
    \item A $U$-oriented spider has no connections on its $U/D$ ports.
  \end{itemize}
  \begin{figure}[h!]
      \centering
      \includegraphics[width=0.45\textwidth]{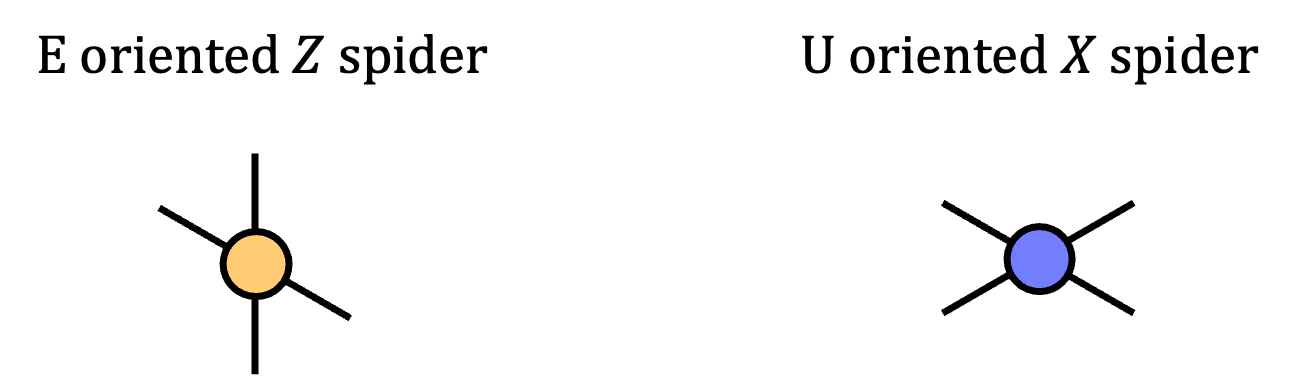}
      \caption{Oriented spider examples. \textbf{Left:} A $E$-oriented $Z$ spider. \textbf{Right:} A $U$-oriented $X$ spider.}
      \label{fig:ozx_examples}
  \end{figure}

  \item \textbf{Connectivity constraints for non-Hadamard edges.}
  Each edge must satisfy one of the two constraints:
  \begin{itemize}
    \item Spiders of the same type (both $Z$ or both $X$) that are directly connected must have the same orientation.
    \item Spiders of opposite type ($Z$ vs.\ $X$) that are directly connected must have different orientations.
  \end{itemize}
  See \cref{fig:table_of_rules} for examples.
  These rules ensure that rough edges connect to rough edges and smooth edges connect to smooth edges in lattice surgery~\cite{de2020zx}. 

  \item \textbf{Hadamard edges and dislocations.}
  Hadamard edges are interpreted as lattice dislocations and are restricted to $E/W$ ports (to simplify hardware; allowing $H$ on other ports is possible but not assumed). For Hadamard edges:
  \begin{itemize}
    \item Spiders of the same type that are directly connected must have different orientations.
    \item Spiders of the opposite type that are directly connected must have the same orientation.
  \end{itemize}
  This ensures rough and smooth edges are always connecting when performing lattice surgery on a dislocation~\cite{de2020zx}.

  \item \textbf{Input/output conventions and idle qubits.}
  Inputs to the diagram enter via $D$ ports while outputs exit via $U$ ports. Inputs and outputs can only attach to $E$-oriented $Z$ spiders or $N$-oriented $X$ spiders. This convention fixes a common orientation for all idle qubits, simplifying their connections to active volume blocks. Often the subadditivity of active volume arises from enforcing this extra constraint between active volume diagrams which need not be followed if the diagrams are combined.

  \begin{figure}[h!]
      \centering
      \includegraphics[width=0.8\textwidth]{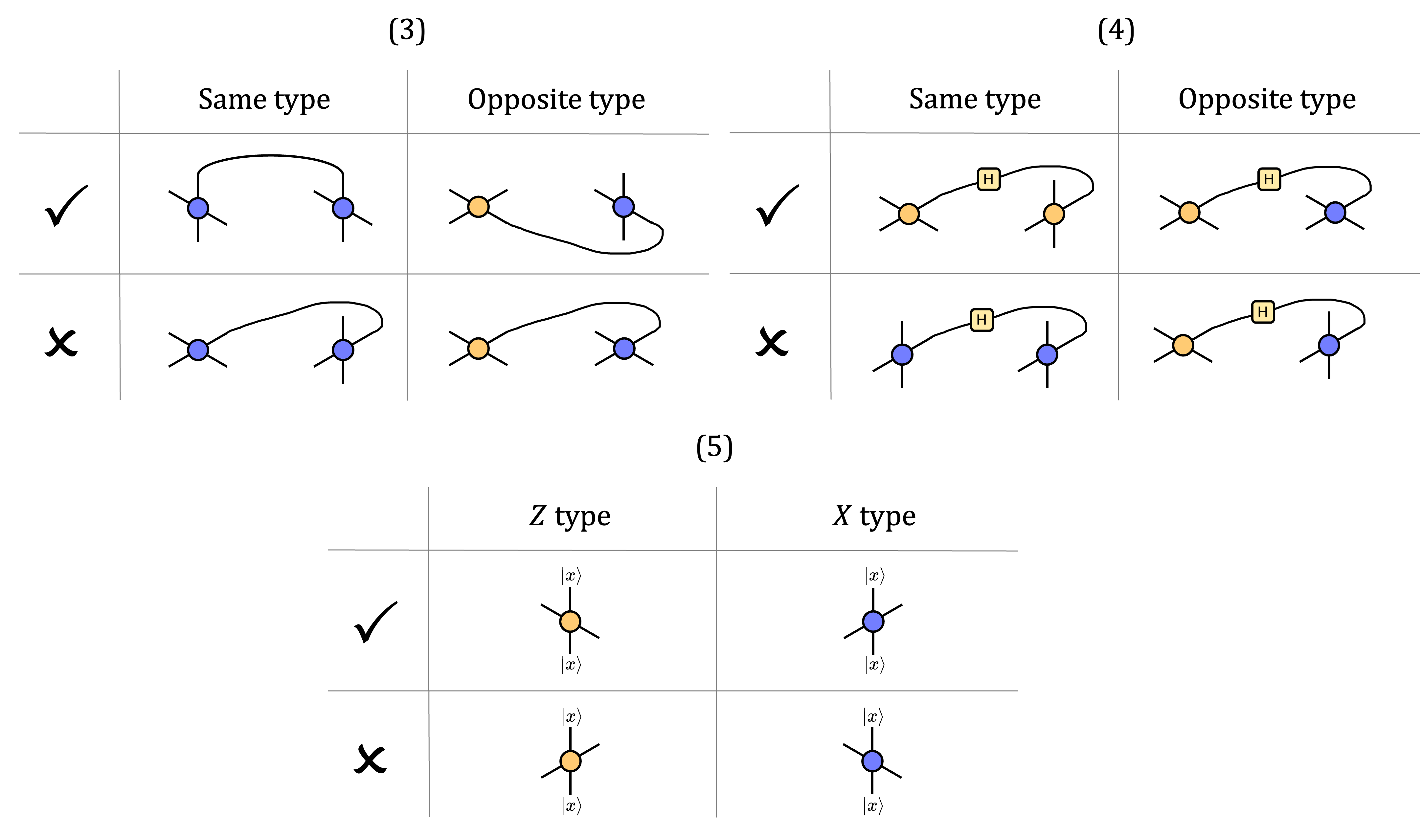}
      \caption{Tables that illustrate rules 3, 4, and 5 by showing both valid and invalid examples.}
      \label{fig:table_of_rules}
  \end{figure}

  \item \textbf{Drawing conventions.}
  For clarity of presentation, some OZX diagrams may use same-axis-to-same-axis rather than same-port-to-same-port connections when symmetry allows (e.g., connecting an $N$ port to an $S$ port). This is done purely for aesthetic purposes and diagrams using this convention should always have a valid corresponding same-port-to-same-port version.

\end{enumerate}
A summary of the port/orientation rules is given in~\cref{table:OZX_rules}.

Once a valid OZX diagram is obtained, the active volume of the subroutine is given by the number of spiders in the diagram together with the distillation cost of any injected magic states.
In practice, we often decompose large circuits into smaller subroutines, derive the active-volume cost of each, and sum the results.
Composing subroutines may introduce a small reorientation overhead: the $U$-port output of one piece must connect to the $D$-port input of the next, which is free when an appropriate port is available but otherwise requires an additional block.

\begin{table}[t]
  \centering
  \renewcommand{\arraystretch}{1.2}
  \begin{tabular}{|l|l|l|l|l|}
    \hline
    \textbf{Spider types} & \textbf{Edge type} & \textbf{Orientations} & \textbf{Allowed ports} & \textbf{Rule} \\ \hline
    Same ($Z$--$Z$ or $X$--$X$)
      & Non-Hadamard
      & $(N,N)$
      & $E/W/U/D$
      & 2, 3 \\
    Same ($Z$--$Z$ or $X$--$X$)
      & Non-Hadamard
      & $(E,E)$
      & $N/S/U/D$
      & 2, 3 \\
    Same ($Z$--$Z$ or $X$--$X$)
      & Non-Hadamard
      & $(U,U)$
      & $N/S/E/W$
      & 2, 3 \\[2pt]
    Opposite ($Z$--$X$)
      & Non-Hadamard
      & $(N,E)$
      & $U/D$
      & 2, 3 \\
    Opposite ($Z$--$X$)
      & Non-Hadamard
      & $(N,U)$
      & $E/W$
      & 2, 3 \\
    Opposite ($Z$--$X$)
      & Non-Hadamard
      & $(E,U)$
      & $N/S$
      & 2, 3 \\ \hline
    Same ($Z$--$Z$ or $X$--$X$)
      & Hadamard
      & $(N,U)$
      & $E/W$
      & 2, 4 \\[2pt]
    Opposite ($Z$--$X$)
      & Hadamard
      & $(N,N)$
      & $E/W$
      & 2, 4 \\
    Opposite ($Z$--$X$)
      & Hadamard
      & $(U,U)$
      & $E/W$
      & 2, 4 \\ \hline
    Input Z ($Z$--Input)
      & Non-Hadamard
      & $(E)$
      & $D$
      & 5 \\
    Input X ($X$--Input)
      & Non-Hadamard
      & $(N)$
      & $D$
      & 5 \\ [2pt]
    Output Z ($Z$--Output)
      & Non-Hadamard
      & $(E)$
      & $U$
      & 5 \\
    Output X ($X$--Output)
      & Non-Hadamard
      & $(N)$
      & $U$
      & 5 \\ \hline
  \end{tabular}
\caption{A table specifying how to convert a ZX diagram into an OZX diagram by assigning orientations to spiders and ports to edges. For each edge in the original ZX diagram, first orient the incident spiders (except input/output nodes, which serve only as placeholders for memory qubits and have no orientation) and then assign the edge to a port consistent with the table, given the spider types, edge type, and chosen orientations. If a combination of spider type, edge type, and orientation is not in this table, then it is not allowed. Two edges which connect to the same spider cannot have the same port type. One could view the process of going from ZX to OZX diagrams as assigning an orientation to each spider and an allowed port to each edge in a ZX diagram. The difficulty lies in finding the smallest equivalent ZX diagram which is amenable to this process.}
\label{table:OZX_rules}
\end{table}

\subsection{Example derivation: CY gate}

Here we demonstrate the active volume derivation process for the CY gate.
First, we shall decompose CY into \{$S$, $H$, CNOT, $T$\}, so we can use the identities from~\cref{fig:gate_identities}.

\begin{equation}
    Y = SXS^{\dagger},
\end{equation}
\noindent where $X$, $Y$, and $Z$ are the Pauli matrices. It follows that:
\begin{equation}
    CY = S \ CNOT \ S^{\dagger}.
\end{equation}
We first write this as a circuit diagram, then convert to ZX, convert to OZX, and finally count the number of spiders in the OZX diagram and add the cost of any magic state inputs to get the total logical block count.
This process is displayed in~\cref{fig:CY_derivation}.
We find that a CY can be executed in $13$ blocks including two $\ket{Y}$ states.
However, this can be reduced to just $9$ blocks by using a PPM decomposition (see~\cref{fig:CY_derivation_cheaper}).
Although both decompositions implement the same operation, their compiled costs can differ because $Y$-state injection effectively hides portions of the underlying ZX structure and thereby restrict the available simplifications.
This exemplifies a useful heuristic: decompositions with fewer resource-state inputs (e.g. $Y$ or $T$ states) are often preferable, both because they reduce distillation cost and because they expose more of the ZX diagram to subsequent optimization.

\begin{figure}[h!]
    \centering
    \includegraphics[trim=10 10 10 10,clip,width=0.9\linewidth]{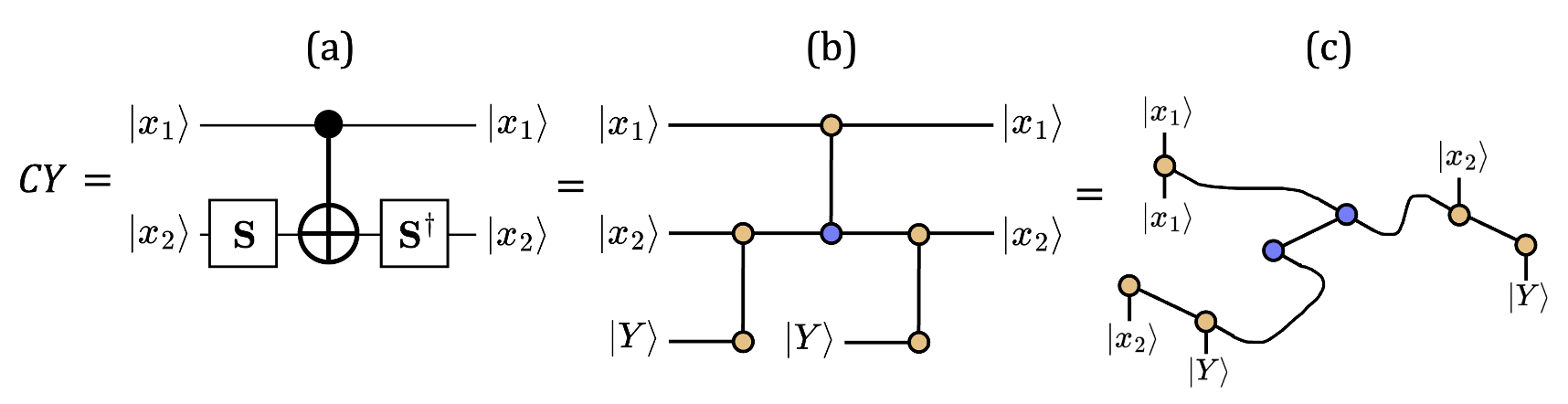}
    \caption{The active volume derivation of a CY gate. \textbf{(a)} First we decompose CY into two $S$ and one CNOT gate. \textbf{(b)} We use the ZX identities from~\cref{fig:gate_identities} to convert the circuit diagram into a ZX diagram. There are no meaningful diagram simplifications at this stage. \textbf{(c)} The OZX diagram requires the addition of one blue spider to obey connectivity requirements. There are $7$ spiders in the diagram and two $\ket{Y}$ states, therefore, the active volume for this decomposition is $13$ blocks.
    }
    \label{fig:CY_derivation}
\end{figure}

\begin{figure}[h!]
    \centering
    \includegraphics[trim=10 10 10 10,clip,width=0.8\linewidth]{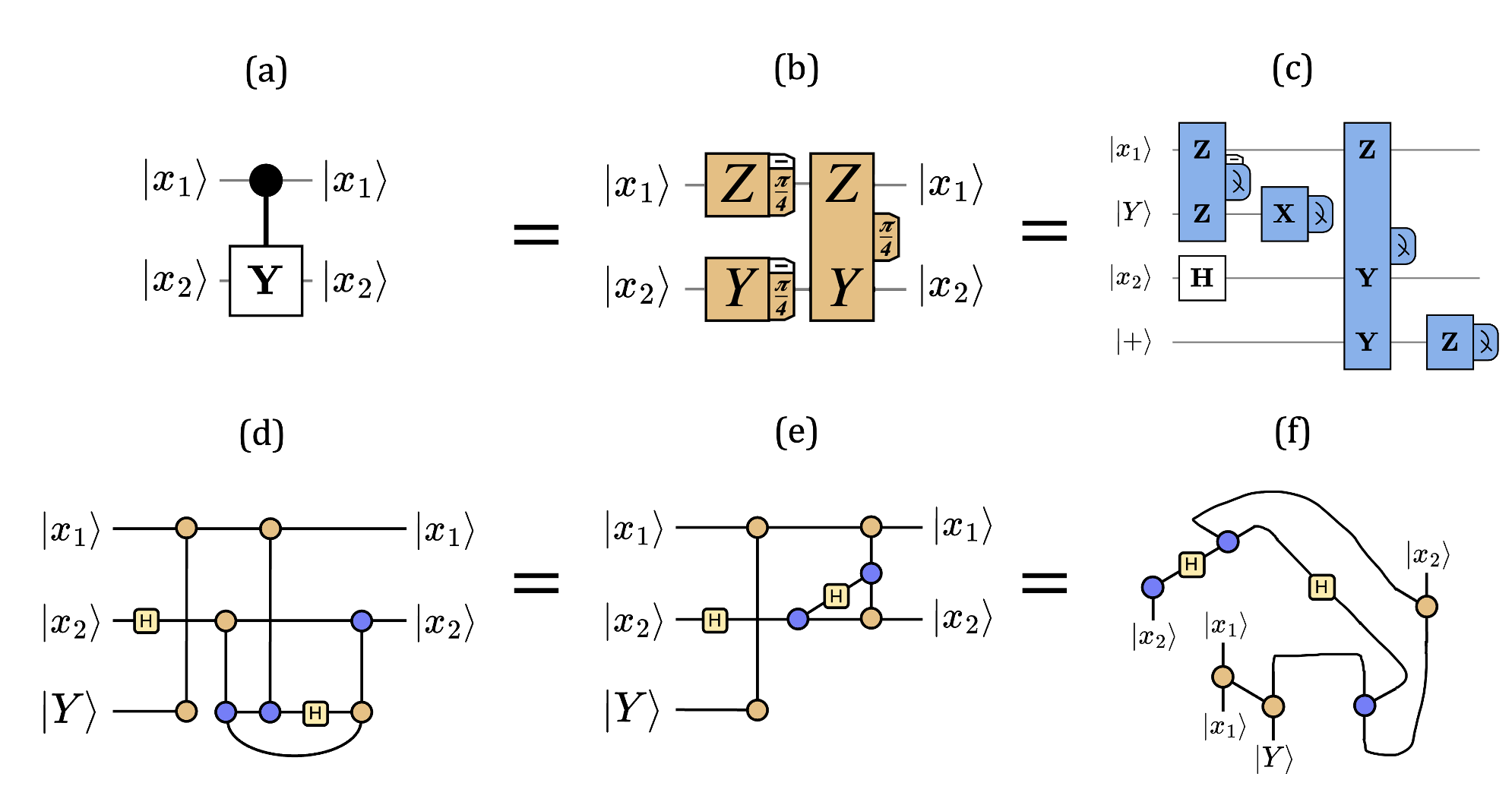}
    \caption{The active volume derivation of a CY gate via PPMs. \textbf{(a)} The circuit diagram for a CY gate. \textbf{(b)} A series of PPRs which are equivalent to a CY gate up to a global $\exp(\mathrm{i}\frac{\pi}{4})$ phase. \textbf{(c)} The PPR circuit rewritten as PPMs with state injection. A $Y$ $\frac{\pi}{4}$ PPR is equivalent to a $Z$ Pauli, which can be implemented through classical interpretation, and an $H$ gate. \textbf{(d)} The ZX diagram for the PPM circuit from \textbf{(c)}, obtained by following the prescription given in \cite[Figure 12]{litinski2022activevolume}. \textbf{(e)} A simplified ZX diagram representing the same circuit. \textbf{(f)} The OZX diagram for a CY gate. It has 6 spiders which correspond to 6 active volume blocks, it also consumes a $\ket{Y}$ state which costs 3 active volume blocks \cite{litinski2022activevolume}. Therefore the total cost of a CY gate is 9 blocks.
    }
    \label{fig:CY_derivation_cheaper}
\end{figure}

\clearpage
\subsection{Deriving the active volume of $\boldsymbol{\mathrm{CNOT}}$ and $\boldsymbol{\mathrm{CZ}}$ fanout}\label{app:fanout_av}
\begin{figure}[H]
    \centering
    \includegraphics[trim=10 10 10 10,clip,width=0.5\linewidth]{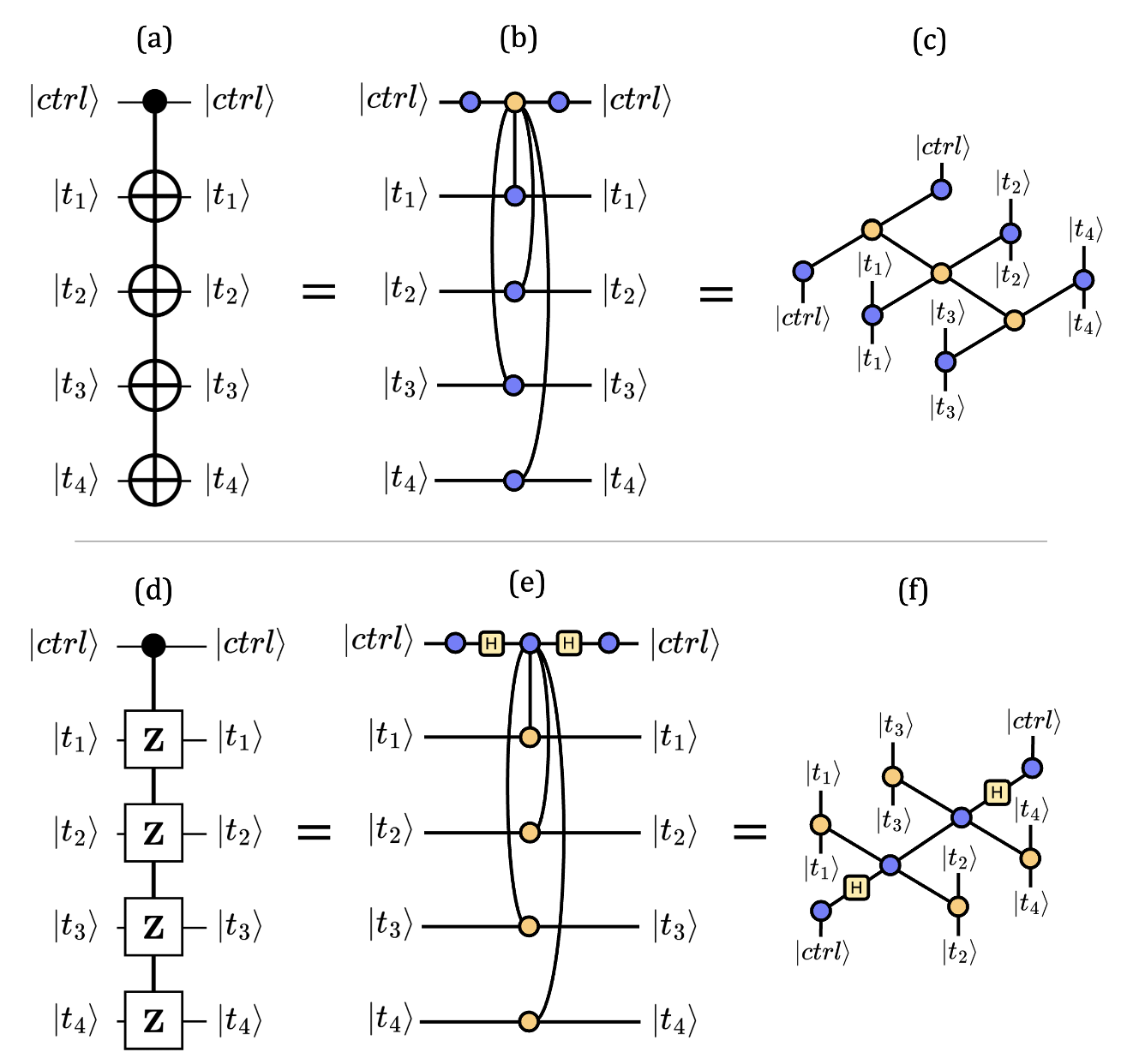}
    \caption{The active volume derivations of a multi-target (fanout) CNOT and CZ gate. \textbf{(a)} shows the circuit diagram representation of a multi-target CNOT. \textbf{(b)} The multi-target CNOT written as a ZX diagram. Here we use the identity provided in \cref{fig:gate_identities}, merge the control spiders, and insert $X$-type identity spider at the start and end of the control rail. \textbf{(c)} The OZX diagram for (b). More targets can be added by repeating the blue, orange, blue pattern. The block cost of an $m$-target multi-target CNOT is $\lceil \frac{3}{2} m \rceil + 3$. \textbf{(d)}, \textbf{(e)}, and \textbf{(f)} follow the same procedure but for a multi-target CZ. In (e), we use that a CZ is equivalent to a CNOT with Hadamards before and after the target. Since CZ is symmetric, we can treat the $\ket{ctrl}$ as the target and take $\ket{t_i}$ as the CNOT control.  By noticing the pattern shown in (f) we find the block cost for a $m$-target multi-target CZ is $\lceil \frac{3}{2} m \rceil + 2$.}
    \label{fig:fanout_av}
\end{figure}

\subsection{Deriving the active volume of $\boldsymbol{\sqrt{X}}$}\label{app:root_x_av}
\begin{figure}[H]
    \centering
    \includegraphics[trim=10 10 10 10,clip,width=0.85\linewidth]{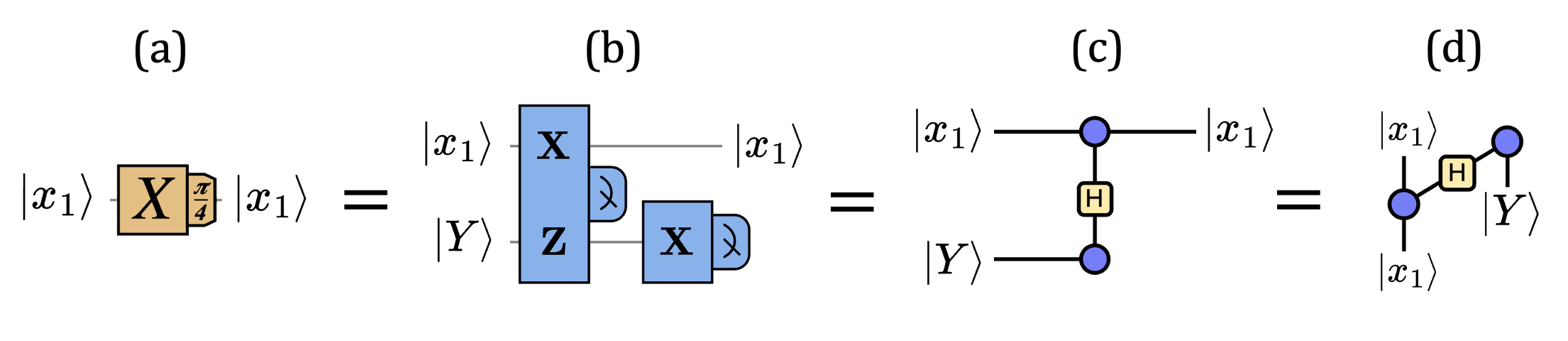}
    \caption{The active volume derivation of $\sqrt{X}$. \textbf{(a)} $\sqrt{X}$ is equivalent to an $X$ $\frac{\pi}{4}$ PPR up to global $\mathrm{e}^{\mathrm{i}\frac{\pi}{4}}$ phase. \textbf{(b)} $X$ $\frac{\pi}{4}$ PPRs can be implemented via $\ket{Y}$ state injection with an $X$ Pauli. The measurement outcomes of the $X \otimes Z$ PPM and $X$ determine the presence of an $X$ Pauli correction on $\ket{x_1}$. This correction can be implemented classically and is excluded from the diagram for simplicity. \textbf{(c)} The simplified ZX diagram for the circuit. \textbf{(d)} The OZX diagram for the circuit. It contains 2 spiders and consumes a $\ket{Y}$ state, which costs 3 blocks to create. Therefore, $\sqrt{X}$ costs 5 blocks.}
    \label{fig:root_x_av}
\end{figure}

\subsection{Deriving the active volume of $\boldsymbol{\mathrm{CNOT} \cdot \mathrm{CZ}}$}\label{app:cx_cz}
This gate combination appears in the fSWAP network and is assigned a cost of 5 logical blocks according to~\cref{fig:cx_cz}.
\begin{figure}[H]
    \centering
    \includegraphics[trim=10 10 10 10,clip,width=0.6\linewidth]{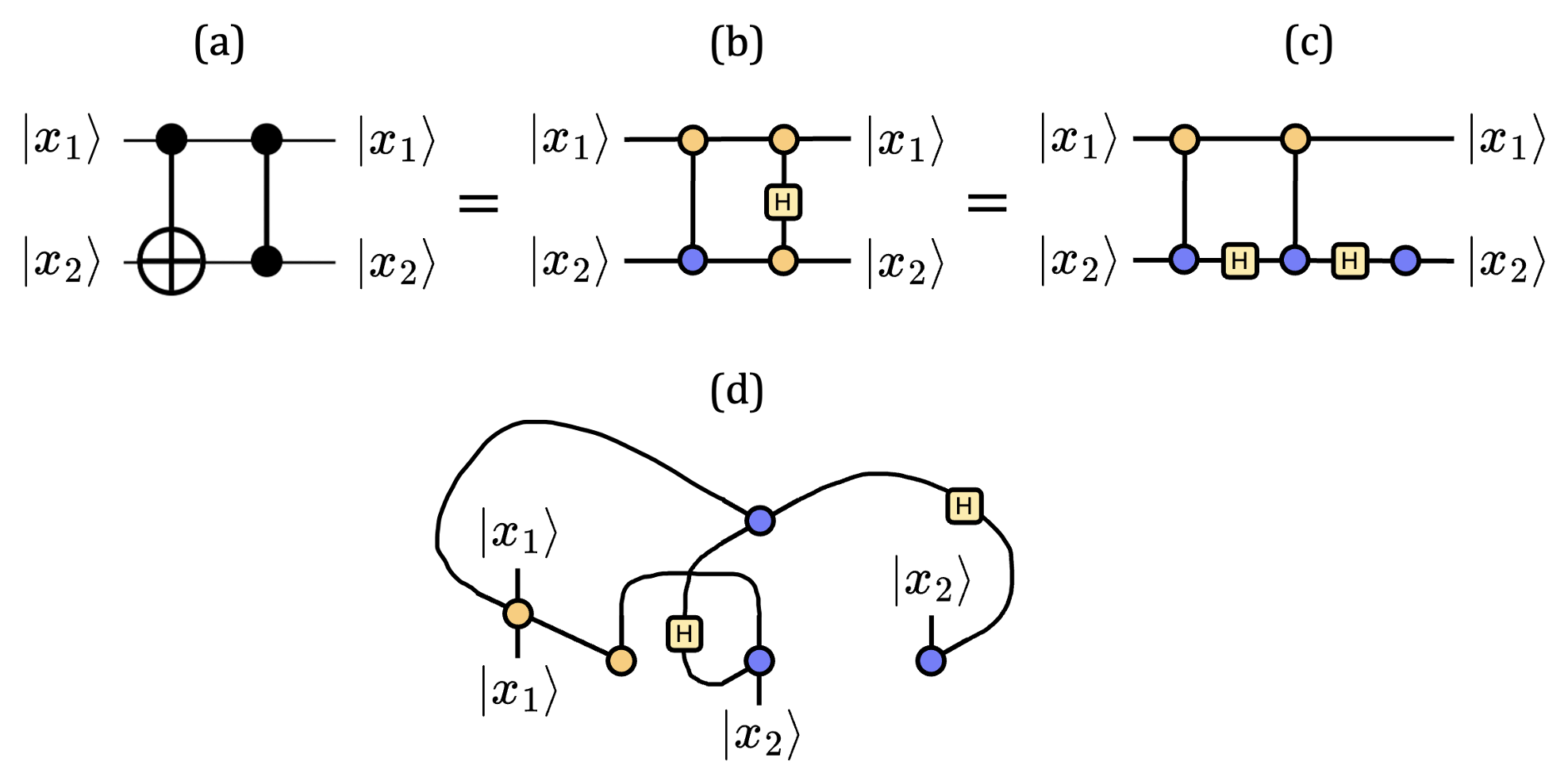}
    \caption{Active volume derivation for $\mathrm{CNOT} \cdot \mathrm{CZ}$. \textbf{(a)} The circuit diagram for a CNOT followed by a CZ gate. \textbf{(b)} The same circuit but written as a ZX diagram. \textbf{(c)} A rewrite of the ZX diagram where we have flipped the color of the bottom right orange spider by applying Hadamards to its edges. We also insert a blue identity spider between the Hadamard edge and the $\ket{x_2}$ output. \textbf{(d)} The same circuit transformed into an oriented ZX diagram. The diagram contains 5 spiders, therefore, a CNOT followed by a CZ costs 5 spacetime blocks.}
    \label{fig:cx_cz}
\end{figure}

\subsection{Deriving the active volume of $\boldsymbol{Y \otimes X \otimes Y}$ PPM for plaquette evolution operation}\label{app:yxz_ppm_av}

We derive the active volume of a $Y \otimes X \otimes Y$ PPM, which is used as part of the circuit for evolving under the Pauli operators $XX + YY$ in~\cref{subsubsec:PPR}.
A PPR can be decomposed as a joint PPM on the input register and on an ancilla qubit prepared in the state $\ket{+}$, derived in~\cref{fig:litinski_clifford_ppr}.
\begin{figure}[H]
    \centering
    \includegraphics[trim=10 10 10 10,clip,width=0.6\linewidth]{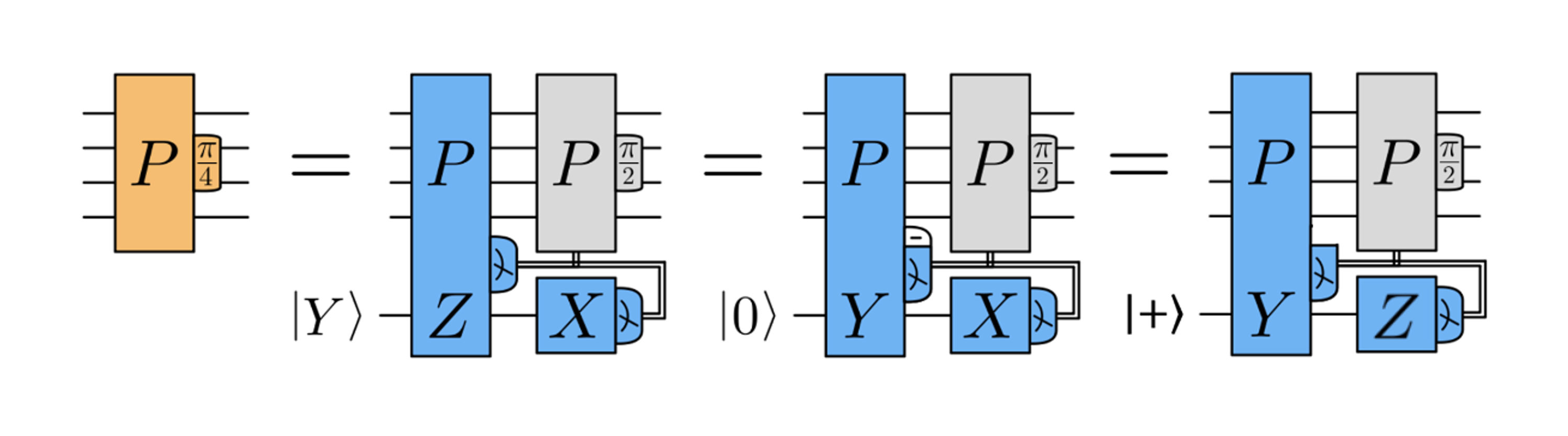}
    \caption{PPR decomposed in terms of PPMs. The first three equalities are copied from \cite{Litinski_2019}, the RHS is obtained by inserting a $H$ gate on the bottom rail and absorbing it into the $\ket{0}$ state and the $X$ measurement. When Pauli string $P$ has even (odd) $Y$ parity the middle-left (RHS) decomposition should be used. This is because PPMs with odd $Y$ parity induce a phase, $\mathrm{i}$. To prevent this, we can make the $Y$ parity even by adding a $Y$ measurement on a $Y$ state to the PPM. However, it is cheaper, in the case of $\frac{\pi}{4}$ PPRs, to absorb the injected $Y$ state into the PPM to obtain even $Y$ parity in the PPM.
    }
    \label{fig:litinski_clifford_ppr}
\end{figure}

In this case $P=Y \otimes X$, we must cost the implementation of a PPM with Pauli $Y \otimes X \otimes Y$.
Following the prescription in \cite[Figure 12]{litinski2022activevolume}, we show in~\cref{fig:av_yx_pi4_ppr} that the cost of a $Y \otimes X$ $\pi/4$ PPR is 5 blocks. 
\begin{figure}[H]
    \centering
    \includegraphics[trim=10 10 10 10,clip,width=0.7\linewidth]{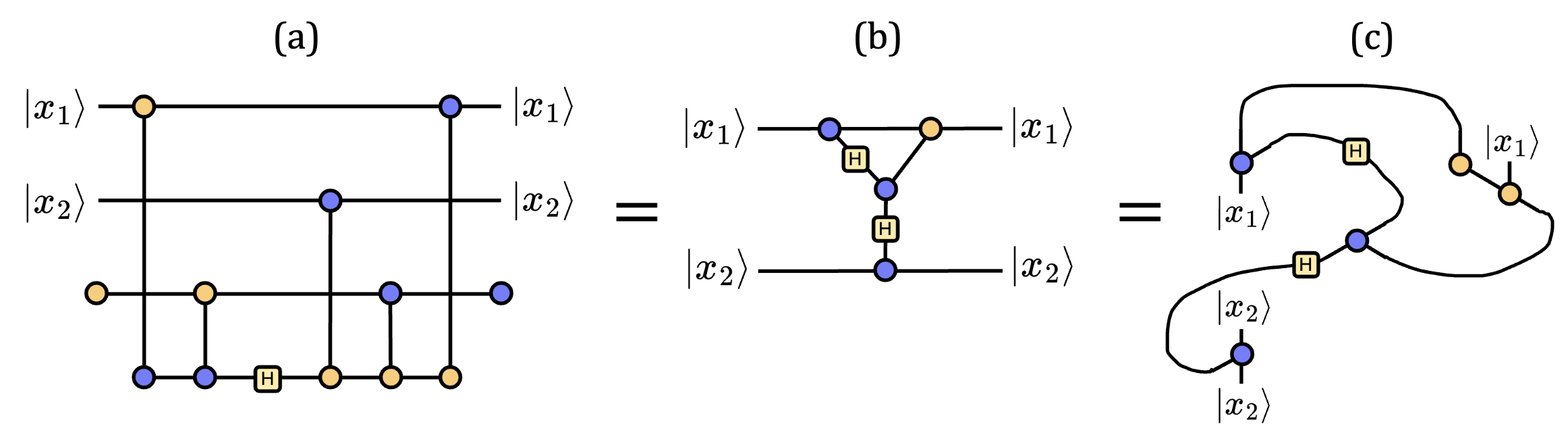}
    \caption{Active volume derivation of $Y \otimes X$ $\pi/4$ PPR. \textbf{(a)} The ZX diagram for the PPM in \cref{fig:litinski_clifford_ppr}. \textbf{(b)} The simplified ZX diagram. \textbf{(c)} The logically equivalent OZX diagram containing 5 spiders.}
    \label{fig:av_yx_pi4_ppr}
\end{figure}

\subsection{Deriving the active volume of an $\boldsymbol{(X \otimes X)(Y \otimes Y)}$ $\boldsymbol{\pi/8}$ PPR for two-mode FFFT}\label{app:zxppms}

Consider the $\pi/8$ PPRs decomposed into PPMs given in~\cref{fig:FH_PPMs} that appear in the 
two-mode FFFTs.
\begin{figure}[h!]
        \centering
        \includegraphics[trim=10 10 10 10,clip,width=0.55\textwidth]{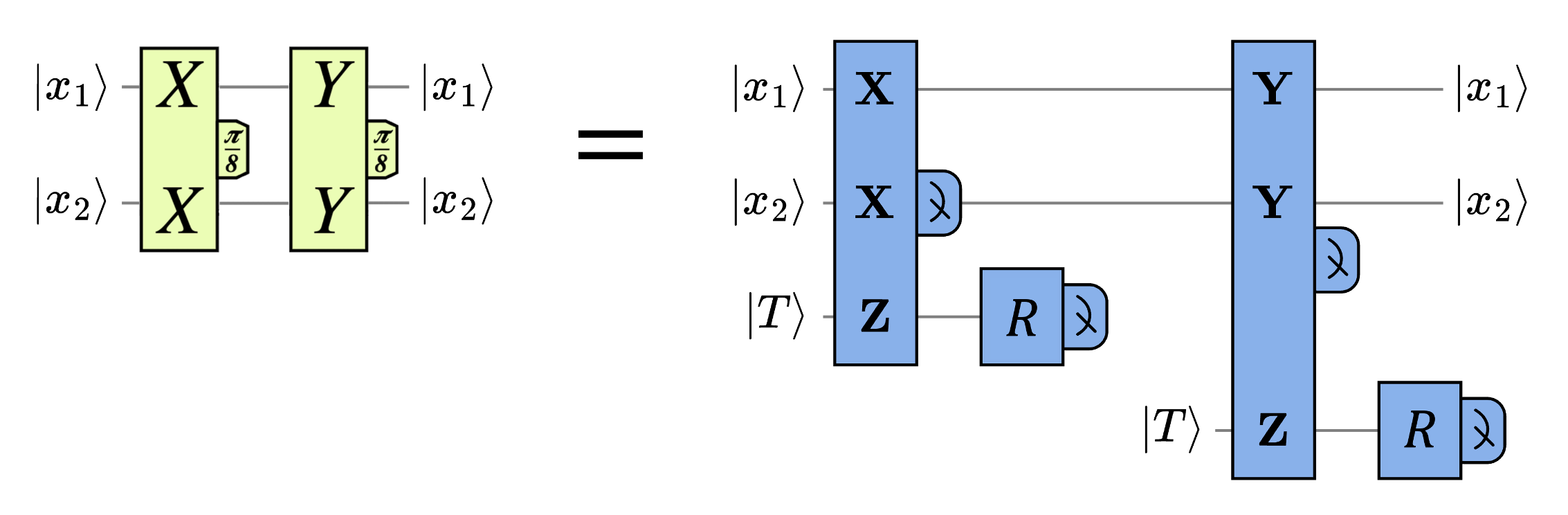}
        \caption[]
        {The PPM decomposition of the compiled 
        two-mode FFFT, which consumes two $\ket{T}$ states. The single qubit ``$R$'' measurements represent reactive measurements; based on the outcome of the previous PPM, we either apply an $X$- or $Y$-basis measurement. Based on the outcome of these measurements, we conditionally apply Pauli fix-ups. The fix-ups are $X \otimes X$ for the first ``$R$'' measurement (leftmost) and $Y \otimes Y$ for the second ``$R$'' measurement (rightmost). However, we exclude the fix-ups from the circuit as they can be implemented classically.} 
        \label{fig:FH_PPMs}
\end{figure}
\begin{figure}[h!]
        \centering
        \includegraphics[trim=10 20 10 10,clip,width=0.8\textwidth]{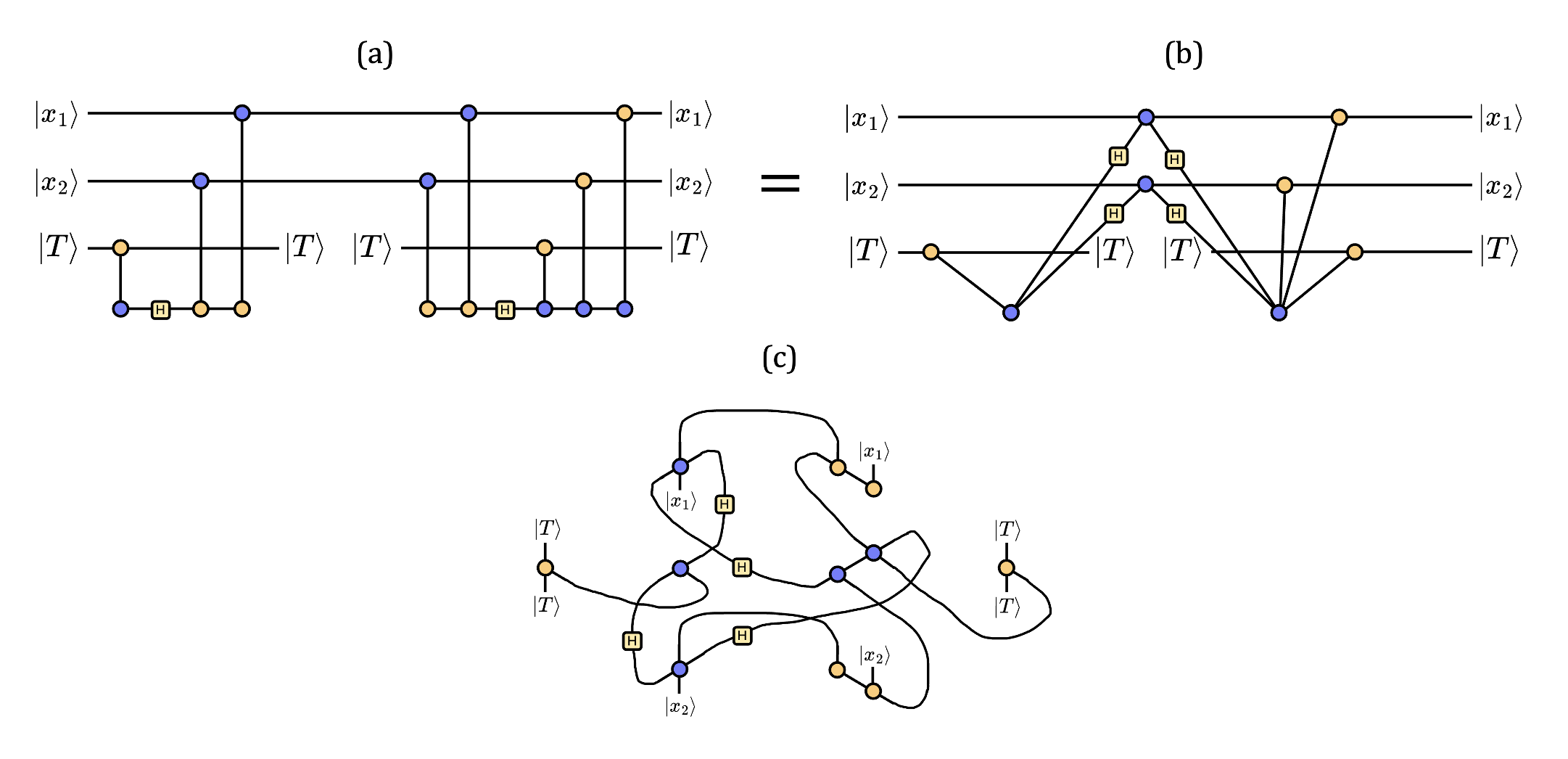}
        \caption{Active volume derivation for the PPM decomposition of the compiled two-mode FFFT. \textbf{(a)} The unoptimized ZX diagram for the PPM decomposition in~\cref{fig:FH_PPMs}. \textbf{(b)} The corresponding simplified ZX diagram. Note, we have used the spider merging and type switching rules. \textbf{(c)} The OZX diagram, it contains 11 spiders which correspond to 11 blocks. However, we also consume two $\ket{T}$ states which each have a 50\% chance of requiring reactive $Y$ measurements. The total cost of an $X\otimes X$ and $Y\otimes Y$ $\frac{\pi}{8}$ PPR is therefore 64 blocks.} 
        \label{fig:xx_yy_pprs}
\end{figure}

The unoptimized ZX diagram for these PPMs is shown in~\cref{fig:xx_yy_pprs} \textcolor{blue}{(a)}.
We do not optimize the parts of the circuit responsible for generation and consumption of magic states; we treat those as a pre-specified input and output and black-box those costs.
Simplifying the ZX diagram yields~\cref{fig:xx_yy_pprs} (b). Note that we are actually implementing a $(Y \otimes Y\otimes Z)^\dagger$ PPM so we can merge blue spiders from each $\pi/8$ PPR. This is only possible because flipping our interpretation of the measurement outcomes associated with the $(Y \otimes Y\otimes Z)^\dagger$ PPM effectively implements a $(Y \otimes Y\otimes Z)$ PPM.
Given that each pair of spiders was from a different PPR, this optimization wouldn't happen if we black-boxed the cost in blocks of each individual PPR. 
In~\cref{fig:xx_yy_pprs} (c) we write the logically equivalent OZX diagram. It contains 11 spiders and consumes two $\ket{T}$ states. Furthermore, these $\ket{T}$ states have a 50\% chance of requiring reactive $Y$-basis measurements, which are implemented by performing a Bell measurement between the $\ket{T}$ state and a $\ket{Y}$ state, consuming them both. $\ket{T}$ and $\ket{Y}$ states cost 25 and 3 blocks, respectively~\cite{litinski2022activevolume}, so the total cost of an $X\otimes X$ and $Y\otimes Y$ $\frac{\pi}{8}$ PPR is $64$ blocks.

\subsection{Deriving the active volume of Hamming weight phasing adder segments}\label{appen:adder_av}

The Hamming weight computation adder segments are highlighted in~\cref{fig:hw_adder_av_counts} and those used in the phasing circuit are drawn out in~\cref{fig:phasing_circuit_av_counts}.

\begin{figure}[H]
\centering
\includegraphics[trim=10 10 10 10,clip,width=0.9\textwidth]{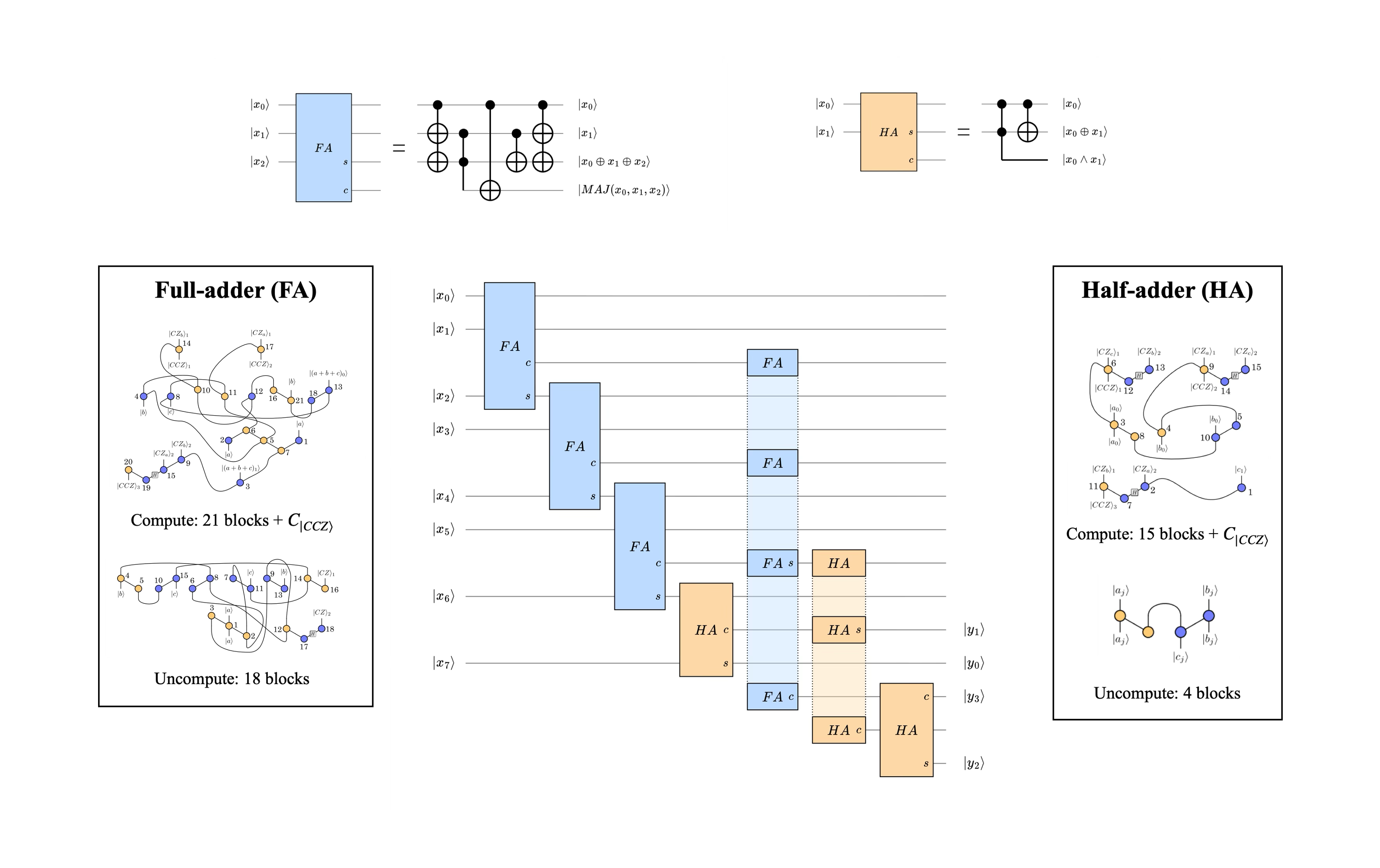}
\caption{Reduced active volume counts for adder segments used in Hamming weight computation.
We show an example circuit for computing the Hamming weight of an 8-qubit register with the circuit diagram adapted from Ref.~\cite{kan2024resource-optimized} and method from \cite{gidney2018halving}. This circuit comprises full adders (FA) and half adders (HA).
For both, their corresponding OZX diagrams (from Ref.~\cite{litinski2022activevolume}) have been compressed to reduce the active volume.
We note $C_{\ket{\text{CCZ}}}$ is the cost of the CCZ resource state, which is estimated to cost around 35 logical blocks using the two-stage distillation protocol described in Ref.~\cite{litinski2022activevolume}.
}
\label{fig:hw_adder_av_counts}
\end{figure}

\begin{figure}[H]
\centering
\includegraphics[trim=10 10 10 10,clip,width=0.9\textwidth]{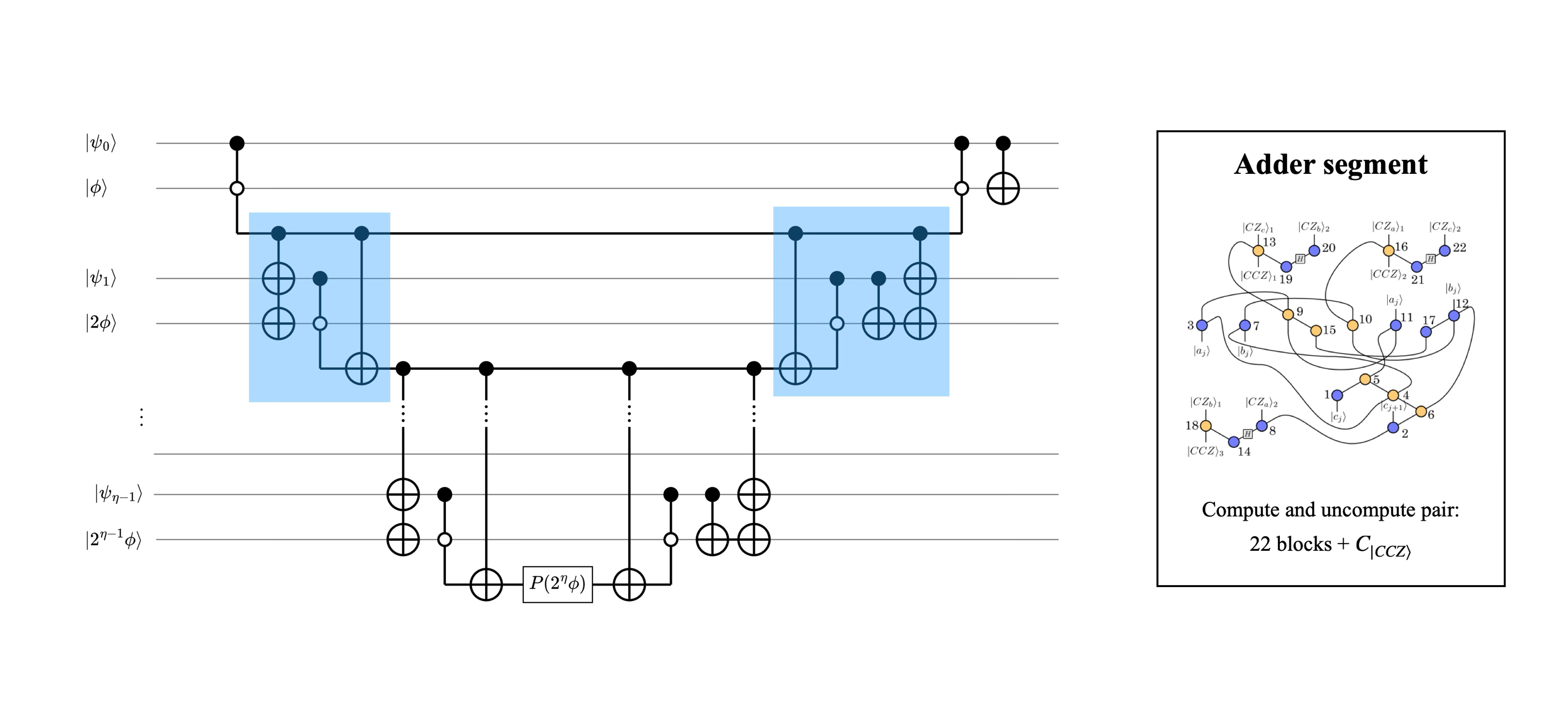}
\caption{Reduced active volume counts for adder-like segments (highlighted in blue) used in generalized phase gradient addition \cite{gidney2019efficient, kan2024resource-optimized, sanders2020compilation}. 
The OZX diagram is from Ref.~\cite{litinski2022activevolume}.
$C_{\ket{\text{CCZ}}}$ is the active volume cost of the CCZ resource state.
In this circuit diagram, we refer to the left half of the circuit (left of the payload rotation) as $A_{\llcorner}$ and the right half (right of the payload rotation) as $A_{\lrcorner}$.
}
\label{fig:phasing_circuit_av_counts}
\end{figure}

\section{Phase fix-up rotations in Hamming weight phasing}
\label{appen:phase-fix-ups}

In this appendix, we derive the local phase corrections required when replacing controlled and directionally controlled towers of equiangular $R_z$ rotations with controlled and directionally controlled HWP circuits, respectively.
Recall that an $R_z$ gate and phase gate ($P$) are defined as
\begin{align}
    R_z(\theta) &:= \begin{bmatrix}
        \mathrm{e}^{-\mathrm{i}\theta/2} & 0 \\
        0 & \mathrm{e}^{\mathrm{i}\theta/2}
    \end{bmatrix},
    &
    P(\theta) &:= \begin{bmatrix}
        1 & 0 \\
        0 & \mathrm{e}^{\mathrm{i}\theta}
    \end{bmatrix},
\end{align}
respectively.

\paragraph{Fix-up in controlled HWP.}

\begin{figure}[ht]
    \centering
    \begin{subfigure}[b]{0.46\textwidth}
        \centering
        \includegraphics[trim=30 30 30 30,clip,width=\textwidth]{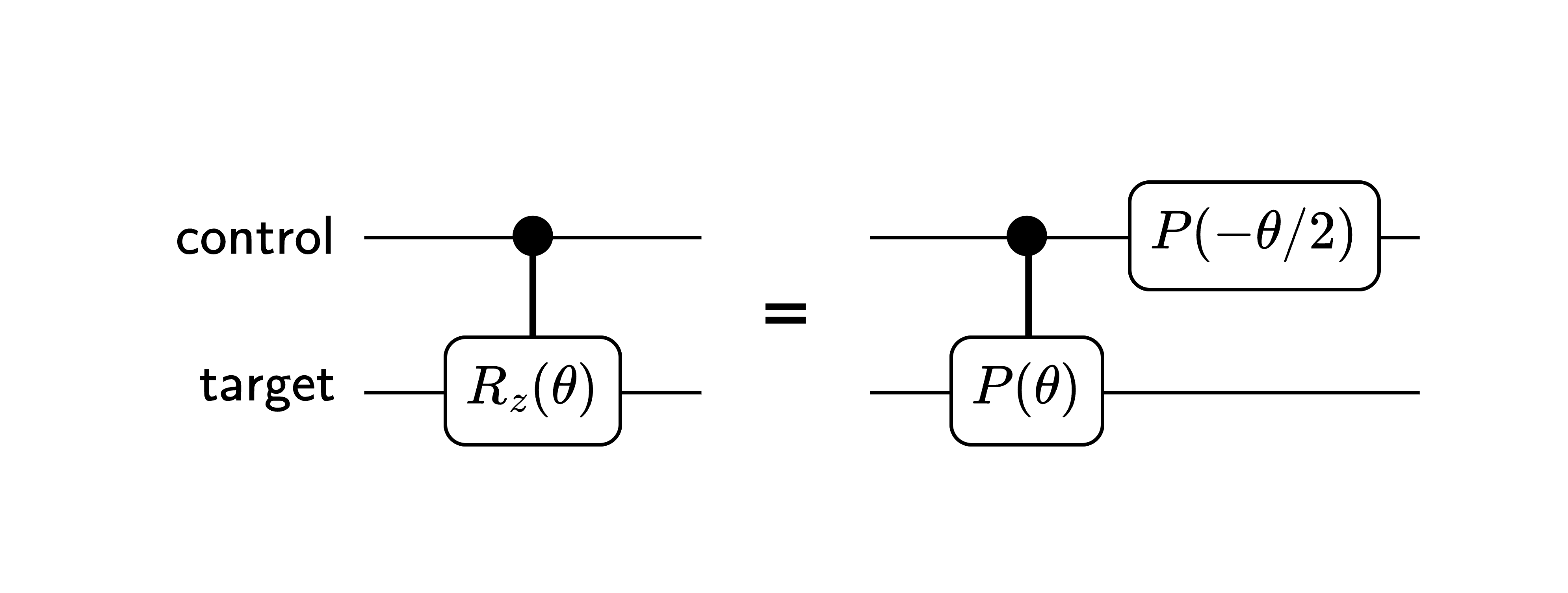}
        \caption{Single target case}
        \label{fig:controlled_rz_controlled_phase}
    \end{subfigure}
    \hfill
    \begin{subfigure}[b]{0.46\textwidth}
        \centering
        \includegraphics[trim=30 30 30 30,clip,width=\textwidth]{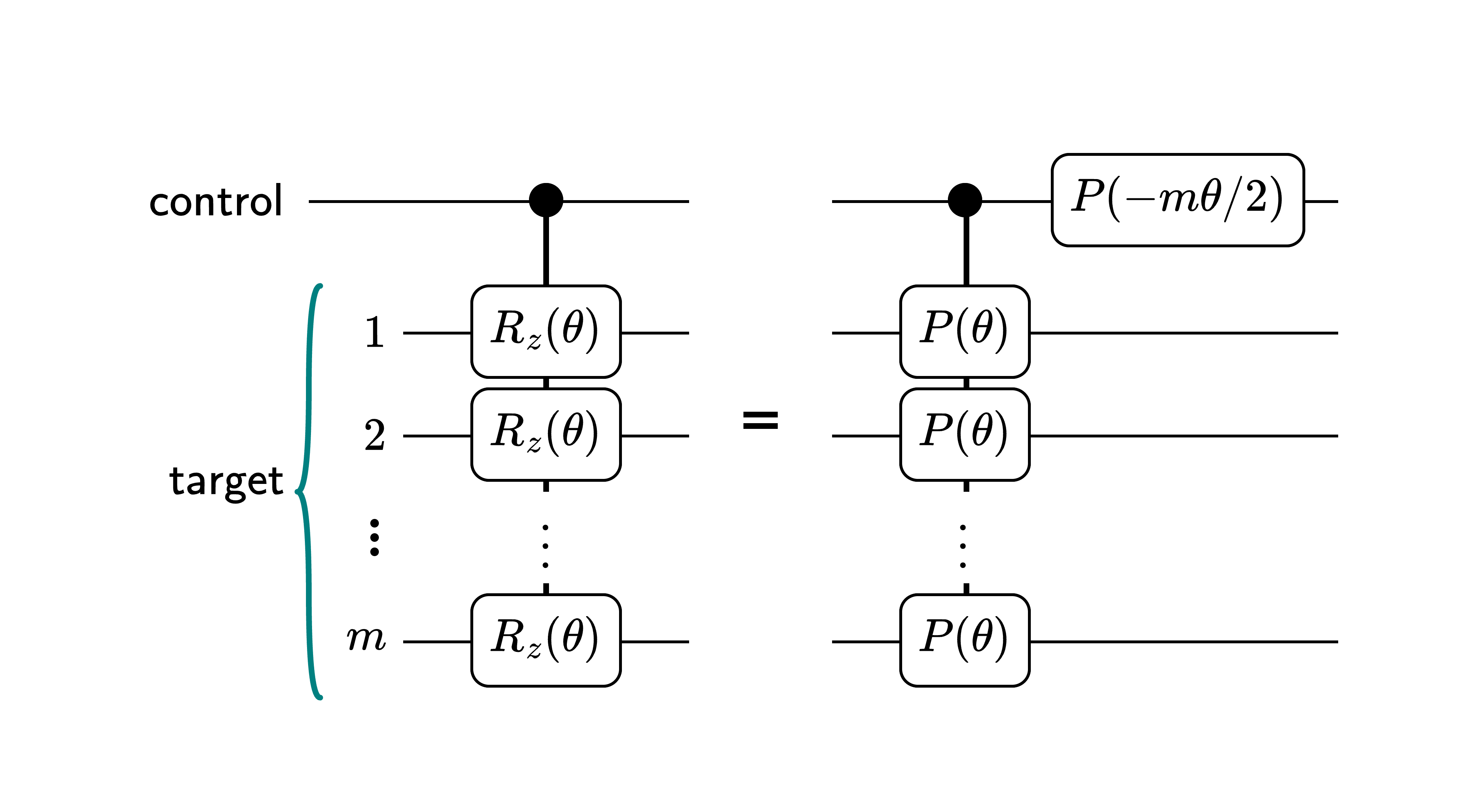}
        \caption{Multi-target case}
        \label{fig:controlled_rz_controlled_phase_multitarget}
    \end{subfigure}
    \caption{Circuit equivalence between a controlled $R_z$ rotation and a controlled phase rotation gate. (a) shows the equivalence in the case of a single target qubit, while (b) shows an extended version of the equivalence with $m$ target qubits. In (b), the phase fix-up gate is repeated $m$ times, which simplifies to a single phase gate.}\label{fig:rz_phase_controlled}
\end{figure}

A controlled $R_z$ and a controlled phase rotation can be made equivalent by applying a phase fix-up as shown in~\cref{fig:controlled_rz_controlled_phase}.
In the context of this work, we are substituting a controlled tower of $R_z$ rotations acting on $L^2$ qubits with a controlled implementation of Hamming weight phasing. The controlled version of the original HWP circuit using GPGA, in which the payload rotation is a phase gate, synthesizes the controlled multi-target phase gate. This equivalence is shown in~\cref{fig:controlled_rz_controlled_phase_multitarget}. 
This leaves $L^2$ applications of the phase fix-up gate, which accumulates to:
\begin{equation*}
P\bigg(-\frac{L^2\theta}{2}\bigg).
\end{equation*}
In other words, we need to correct by a local phase of $-\theta/2$ for each rotation in the tower, and there are $L^2$ of such rotations.

\paragraph{Fix-up in directionally controlled HWP.}

The original circuit we are trying to substitute with a directionally controlled HWP is the following: a closed-controlled implementation of the (forward) $R_z(\theta)$ tower and an open-controlled implementation of the (inverse) $R_z(-\theta)$ tower. That is, the angles between the two towers are equal in magnitude but opposite in sign.
We can express the two controlled $R_z$ towers in terms of two controlled phase rotation towers of opposite signs by leveraging the multi-target versions of the circuit equivalences~\cref{fig:controlled_rz_controlled_phase} and~\cref{fig:controlled_rz_controlled_phase_open_control}.
Substituting a controlled $R_z(\theta)$ tower with a controlled $P(\theta)$ tower contributes a phase fix-up gate of angle $-L^2\theta/2$ acting on the $\ket{1}$ state of the control qubit, while the inverse case contributes a phase gate of angle $L^2\theta/2$ acting on the $\ket{0}$ state of the control qubit.
These two fix-up gates are equivalent to a single $R_z(-L^2\theta)$ rotation on the control qubit, i.e. $P(-L^2\theta) \cdot X \cdot P(L^2\theta) \cdot X = R_z(-L^2\theta)$.

\begin{figure}[H]
\centering
\includegraphics[trim=20 20 20 20,clip,width=0.6\textwidth]{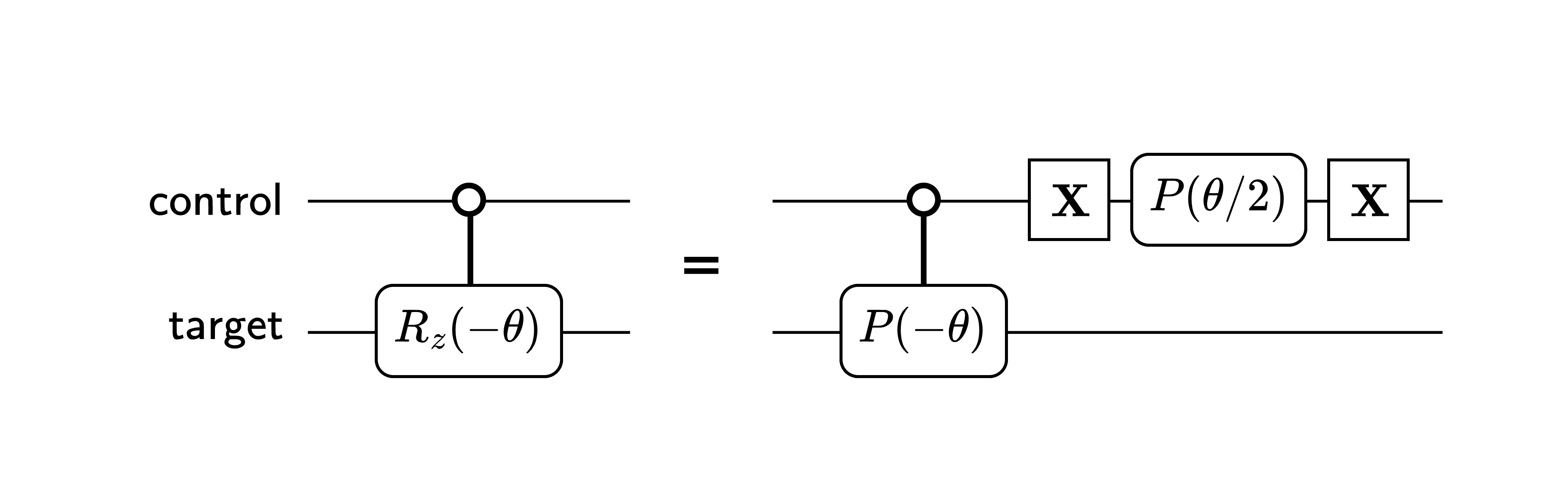}
\caption{Circuit equivalence to substitute an open-controlled $R_z$ gate, with an angle of the opposite sign, with an open-controlled phase gate and a fix-up phase gate.}
\label{fig:controlled_rz_controlled_phase_open_control}
\end{figure}

The directional implementation of HWP implements the forward and inverse towers of phase rotations by conjugating the circuit with open-controlled CNOT fanouts acting on the catalyst register.
In addition, the payload rotation needs to be directionally controlled, applying $P(-2^{\lfloor\log_2(L^2)\rfloor+1}\theta)$ when the control qubit is $\ket{0}$ and $P(2^{\lfloor\log_2(L^2)\rfloor+1}\theta)$ when the control qubit is $\ket{1}$. These two controlled payload rotations can be simplified to two uncontrolled phase rotations:
\begin{figure}[H]
\centering
\includegraphics[trim=20 20 20 20,clip,width=0.9\textwidth]{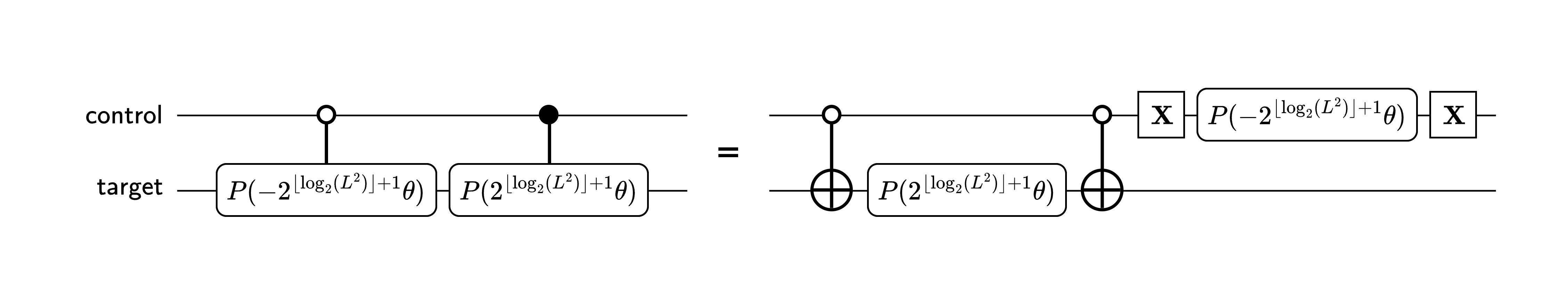}
\caption{Circuit equivalence to replace naive implementation of directionally controlled payload gate in GPGA.}
\label{fig:directionally_controlled_phase_gate}
\end{figure}
\noindent where we replace two controlled phase gates with an uncontrolled payload phase gate and a local phase correction of angle $-2^{\lfloor\log_2(L^2)\rfloor+1}\theta$.

The phase corrections therefore combine into a single local rotation
\begin{equation*}
R_z\left(-L^2\theta
+2^{\lfloor\log_2(L^2)\rfloor+1}\theta\right),
\end{equation*}
up to an overall global phase.

\paragraph{Accumulated phase corrections.}

Within a directionally controlled Trotter step, the local phase corrections introduced by successive directionally controlled HWP calls all act on the same QPE control qubit and therefore commute. Consequently, they may be combined into a single $R_z$ rotation on each of the $k-1$ directionally controlled phase qubits.

Each plaquette HWP contributes a phase correction of
\[
R_z(-L^2\theta_{\mathrm{plaq}}
+2^{\lfloor\log_2(L^2)\rfloor+1}\theta_{\mathrm{plaq}}),
\]
while each interaction HWP contributes
\[
R_z(-L^2\theta_{\mathrm{int}}
+2^{\lfloor\log_2(L^2)\rfloor+1}\theta_{\mathrm{int}}),
\]
where $\theta_{\mathrm{plaq}}$ and $\theta_{\mathrm{int}}$ are the payload rotation angles for the plaquette and interaction HWP circuits, respectively.

The accumulated correction on the $(j+1)$-th QPE phase qubit is therefore
\begin{equation*}
R_z\bigg(2^j r \left[
-L^2\left(\theta_{\mathrm{plaq}}+\theta_{\mathrm{int}}\right)
+2^{\lfloor\log_2(L^2)\rfloor+1}
\left(\theta_{\mathrm{plaq}}+\theta_{\mathrm{int}}\right)
\right]\bigg),
\qquad j\in[0,k-2],
\end{equation*}
where $r$ is the number of Trotter steps.
The accumulated corrections over the $k-1$ directionally controlled phase qubits therefore form a generalized phase gradient.
Since this phase gradient is uncontrolled, it can be synthesized directly using GPGA using $k-1$ Toffoli gates and a single phase rotation, without requiring any additional phase-fixup correction.

\section{Error budgeting details}\label{appen:error_budgets}

We use an effective-Hamiltonian analysis to relate the unitary error, which naturally characterizes the circuit approximation, to the resulting error in the estimated energy.
Any unitary operator can be written in the form $U(s) =\mathrm{e}^{\mathrm{i} H s}$ for some Hermitian operator $H$.
Accordingly, any approximate unitary $\tilde{U}(s)$ can be expressed as $\mathrm{e}^{\mathrm{i} H_\text{eff} s}$ for some effective Hamiltonian $H_\text{eff}$ that will depend on $s$.
This allows us to bound the unitary error in terms of the resulting eigenenergy deviations~\cite{bhatia1984bound,kivlichan2020improved}: 
\begin{align}
\| U(s) - \tilde{U}(s)\| = \| \mathrm{e}^{\mathrm{i} H s } -  \mathrm{e}^{\mathrm{i} H_\text{eff} s } \| \geq \left | \mathrm{e}^{\mathrm{i} E^{(k)} s } - \mathrm{e}^{\mathrm{i} E^{(k)}_\text{eff} s } \right | = 2 \sin \left(\frac{|E^{(k)}-E^{(k)}_\text{eff}|s}{2}\right),
\end{align}
for all $k$-th eigenvalues of $H$ and $H_\text{eff}$. Equivalently:
\begin{align}\label{eq:energy_to_unitary_error}
\abs{E^{(k)} - E^{(k)}_\text{eff}}  \leq \frac{2}{s} \arcsin \left(\frac{\| U(s) - \tilde{U}(s) \|}{2}\right) \leq \epsilon_{\text{eff}}.
\end{align}
We use the right hand inequality of \cref{eq:energy_to_unitary_error} to relate the allowable unitary error from the energy budget:
\begin{equation}
\|V_{\text{QPE}(\tilde{U})} -\tilde{V}_{\text{QPE}(\tilde{U})}\| \leq 2 \sin \left(\frac{\epsilon_{\text{cat}}\tau}{2}\right),\label{eq:energy_to_unitary_error_cat}
\end{equation}
\begin{equation}
 \|U(\tau) - (U_{\tau/r})^r\|
    + \|(U_{\tau/r})^r - (\tilde{U}_{\tau/r})^r \| \leq 2 \sin \left(\frac{\epsilon_{U}\tau}{2}\right).\label{eq:energy_to_unitary_error_evo}
\end{equation}

To simplify the optimization, we perform a first-order approximation of the $\arcsin$ function which introduces negligible error in the regime of interest\footnote{Note that if we were to include error budget for the QFT and window state preparation in QPE implementation error, one could replace $\epsilon_\text{cat}$ by $\epsilon_{V}$ and use a similar approximation $\epsilon_{V}\approx \epsilon_\text{cat} + \epsilon_\text{QFT} + \epsilon_\text{window}$.}:
\begin{subequations}
\begin{align}\label{eq:energy_to_unitary_error_evo_approx}
\epsilon_U &\geq \frac{2}{\tau} \arcsin \left(\frac{ \|U(\tau) - (U_{\tau/r})^r\|
    + \|(U_{\tau/r})^r - (\tilde{U}_{\tau/r})^r \| }{2}\right) \\
&\approx \frac{2}{\tau} \left(\arcsin \left(\frac{ \|U(\tau) - (U_{\tau/r})^r\|}{2}\right) + \arcsin \left(\frac{  \|(U_{\tau/r})^r - (\tilde{U}_{\tau/r})^r \|}{2}\right)  \right)\\
& = \epsilon_{\text{ST}} + \epsilon_{\text{rot}}.
\end{align}
\end{subequations}
This allows for a linear attribution of the energy error to each source:
\begin{subequations}
\begin{align}
 \epsilon &\geq \epsilon_{\text{QPE}} + \epsilon_{\text{ST}} + \epsilon_{\text{rot}} + \epsilon_{\text{cat}}\label{eqn:error_budget}\\
 &= \epsilon\Big[
\delta_{\text{QPE}}
+ (1-\delta_{\text{QPE}})\Big(
\delta_{\text{ST}}
+ (1-\delta_{\text{ST}})\big(
\delta_{\text{rot}} + (1-\delta_{\text{rot}})
\big)
\Big)
\Big]\\
 & =\delta_{\text{QPE}}\epsilon + (1-\delta_{\text{QPE}})\delta_{\text{ST}}\epsilon + (1-\delta_{\text{QPE}})(1-\delta_{\text{ST}})\delta_{\text{rot}}\epsilon + (1-\delta_{\text{QPE}})(1-\delta_{\text{ST}})(1-\delta_{\text{rot}})\epsilon.\label{eqn:deltatoepsilon}
\end{align}
\end{subequations}
By construction, all error allocation parameters $\delta$ lie in the interval $(0,1)$ which ensures the equality in \cref{eqn:error_budget}.
The optimization is therefore performed over the parameters $\{\delta_\text{QPE}, \delta_\text{ST}, \delta_\text{rot}, \tau\}$ using SciPy's basin-hopping algorithm with L-BFGS-B as the local minimizer \cite{scipy}.

We set the search bounds of the evolution time in the optimization to $[\tau_{\min},\tau_{\max}]$.
A conservative choice for $\tau_{\max}$ is $2\pi/\|H_{\text{FH}}\|_1$ to avoid aliasing \cite{O_Brien_2019,baysmidt2026quantumsimulationnanographenestrotter}\footnote{Aliasing occurs when eigenphases are not properly constrained to the interval $[0,2\pi)$, causing them to ``wrap around'' in phase space.}.
However, this bound is overly restrictive for Trotter-based approaches which, unlike qubitization, can exploit aliasing in ground-state energy estimation by increasing the evolution time. 
Increasing the evolution time reduces QPE error but increases Trotter error, leading to an optimal trade-off point that often exceeds $2\pi/\|H_{\text{FH}}\|_1$. 
Practically, as classical methods already provide ground-state energy estimates to some coarse precision, the quantum algorithm is used only to resolve the remaining bits.
Using well-established classical techniques, Ref.~\cite{leblanc2015solutions} provides benchmark zero-temperature estimates from 
AFQMC, DMRG, fixed-node QMC, and DMET that
already form a relatively narrow but still non-negligible band in the challenging regime $U/t=8$, and average electron density per site of $0.875$: DMRG gave $E/L_xL_yt\approx-0.759(4)$, and fixed-node VMC gave $-0.749(2)$. 
Thus, at this benchmark point, the spread is approximately $0.02t$ per site.
To remain conservative, we take the accuracy achievable with these established classical methods, \footnote{We note that classical accuracy has improved significantly in recent years through advances such as VAFQMC, neural quantum states, and tensor-backflow methods~\cite{sorella2023systematically,Levy2024,roth2025superconductivity,liang2026investigating,Viteritti2026}; however, we retain the conservative estimate here.} to be about $5\%\,t$ per site. 
We set the optimizer's upper bound on the evolution time based on this classical accuracy, i.e. $\tau_{\max} = 2\pi/ (0.05 t L_x 
L_y)$.
While the optimal time often exceeds $2\pi/\|H_{\text{FH}}\|_1$, it remains well within the $\tau_{\max}$ limit imposed by classical precision.

\subsection{Suzuki-Trotter error methodology}\label{appen:trotter_error}

The goal of this section is to generalize the analytic commutator bounds of Ref.~\cite{campbell2022early} to arbitrary second-order Suzuki product formulas, allowing different Hamiltonian decompositions, Trotter orderings, lattice geometries, and hopping ranges.
Rather than evaluating commutator norms on exponentially large Hilbert-space operators, our objective is to estimate the Trotter error using only operations acting on matrices whose dimension scales with the number of fermionic modes.
This provides an automated procedure for evaluating commutator bounds for a broad class of Hubbard Hamiltonians.

To this end, we represent the one- and two-body components of the Hubbard Hamiltonian by
\begin{align}
T(\Lambda) &= \sum_{jk}\sum_{\sigma=\uparrow,\downarrow}\Lambda_{jk}c^\dagger_{j\sigma}c_{k\sigma},\\
U(v) &= \sum_k v_k Z_{k\uparrow}Z_{k\downarrow},
\end{align}
where $Z_{k\sigma}=1-2c^\dagger_{k\sigma}c_{k\sigma}$, $\Lambda$ is the hopping matrix, and $v$ specifies the onsite interaction strengths.

Consider a second-order Suzuki product formula composed of Hamiltonians $H_1,H_2,\ldots,H_L$ each of which is represented by a pair $(\Lambda_i,v_i)$.
The corresponding product formula is
\begin{equation}
S_2(s) = \mathrm{e}^{\mathrm{i}H_1 s/2}\mathrm{e}^{\mathrm{i}H_2 s/2} \cdots \mathrm{e}^{\mathrm{i}H_{L-1} s/2} \mathrm{e}^{\mathrm{i}H_{L} s} \mathrm{e}^{\mathrm{i}H_{L-1} s/2} \cdots \mathrm{e}^{\mathrm{i}H_2 s/2}\mathrm{e}^{\mathrm{i}H_1 s/2}
\end{equation}
whose error is bounded by the telescopic commutator expansion
\begin{equation}
\label{eq:telescope0}
\left\| \vphantom{\sum}S_2(s) - \mathrm{e}^{\mathrm{i}\sum_k H_k s} \right\| \leq \abs{s}^3\sum_{k=1}^L \left( \frac{1}{12} \left\|\left[ \left[\vphantom{\sum} \mathcal{L}_k, H_k\right], \mathcal{L}_k\vphantom{\sum_k} \right]\right\| + \frac{1}{24} \left\|\left[\left[\vphantom{\sum}\mathcal{L}_k, H_k \right], H_k\vphantom{\sum_k} \right] \right\|\right)\, , 
\end{equation}
where
\begin{align}
\mathcal L_k=\sum_{m>k}H_m
\end{align}
is the accumulated Hamiltonian following $H_k$ in the Suzuki ordering.

The problem therefore reduces to estimating the two classes of nested commutators appearing in~\cref{eq:telescope0} using only the matrices $\Lambda_i$ and vectors $v_i$.
To accomplish this, we introduce a generic double commutator that encompasses every contribution to the telescopic bound and derive matrix-based expressions and norm bounds for each of its constituent terms.

The two nested commutators appearing in~\cref{eq:telescope0} differ only in the choice of their three Hamiltonian arguments. Since each Hamiltonian is specified by a matrix-vector pair $(\Lambda,v)$, it is sufficient to derive a bound for the generic operator
\begin{equation}\label{eqn:defn_g}
G(A,B,C, a, b, c) \coloneqq \left[ \vphantom{\sum_k}\left[\vphantom{\sum} T(A) + U(a),\;  T(B) + U(b) \right], \; T(C) + U(c) \right] 
\end{equation}
with the matrices $A,B,C$ and vectors $a,b,c$.

\begin{lemma}\label{lem:G_bound}
Let $A$, $B$ and $C$ be hopping matrices and let $a$, $b$ and $c$ be interaction vectors. Then $G(A,B,C, a, b, c)$ as defined in~\cref{eqn:defn_g} is norm bounded by
\begin{equation*}
\begin{multlined}
\|G(A,B,C,a,b,c)\| \leq
2\left\|
T_\sigma\!\left(
[[A,B],C]
+4\{A,D_bD_c\}
-4\{B,D_aD_c\}
\right)
\right\|
+\\
8\left\|
T_\sigma\!\left(
D_bAD_c+D_cAD_b
-D_aBD_c-D_cBD_a
\right)
\right\|
+
4\sum_j
\left\|
\left[
b_jQ_{j\sigma}(A)-a_jQ_{j\sigma}(B),
T_\sigma(C)
\right]
+c_jQ_{j\sigma}([A,B])
\right\|\\
+
8\sum_j
\|Q_{j\sigma}(C)\|
\,
\|b_jQ_{j\sigma}(A)-a_jQ_{j\sigma}(B)\|.
\end{multlined}
\end{equation*}
Here $D_x$ denotes the diagonal matrix with entries $(D_x)_{jk}=x_j\delta_{jk}$, $\{\cdot,\cdot\}$ denotes the matrix anticommutator, and
\begin{equation*}
Q_{j\sigma}(A) \coloneqq \sum_k
\left(
A_{jk}c^\dagger_{j\sigma}c_{k\sigma}
-
A_{kj}c^\dagger_{k\sigma}c_{j\sigma}
\right)
\end{equation*}
is the star operator associated with the $j$th site.
\end{lemma}

The significance of~\cref{lem:G_bound} is that every quantity on the right-hand side can be evaluated using only matrices of dimension equal to the number of fermionic modes.
In particular, the norm of $T_\sigma(A)$ is determined by the spectrum of $A$, while the norm of $Q_{j\sigma}(A)$ is determined by the spectrum of the matrix $D_jA-AD_j$.
Consequently, the many-body commutator bound is reduced to a sequence of polynomial-cost matrix operations.

For the $k$'th contribution to the telescopic bound, the arguments $(B,b)$ are identified with the Hamiltonian parameters $(\Lambda_k,v_k)$, while $(A,a)$ correspond to the accumulated Hamiltonian parameters
\begin{equation}
(A,a)=\left(\sum_{m>k}\Lambda_m,\sum_{m>k}v_m\right).
\end{equation}
The two nested commutators in~\cref{eq:telescope0} are then obtained by choosing $(C,c)=(A,a)$ for $[[\mathcal{L}_k,H_k],\mathcal{L}_k]$ and $(C,c)=(B,b)$ for $[[\mathcal{L}_k,H_k],H_k]$.
Applying~\cref{lem:G_bound} therefore yields every contribution required to evaluate the second-order Suzuki error bound.
The remainder of this appendix is a proof of the technical result.

\begin{proof}
We begin by decomposing $G$ into simpler nested commutators according to the one- and two-body components of the Hamiltonians.
Expanding the nested commutator using linearity isolates the distinct combinations of one- and two-body Hamiltonians that must be bounded. Since $[U(a),U(b)]=0$, only six terms remain,
\begin{equation}
\begin{multlined}
G(A,B,C,a,b,c)
=\left[\left[T(A),T(B)\right],T(C)\right] 
+\left[\left[T(A),U(b)\right],T(C)\right] 
-\left[\left[T(B),U(a)\right],T(C)\right] \\
+\left[\left[T(A),T(B)\right],U(c)\right] 
+\left[\left[T(A),U(b)\right],U(c)\right] 
-\left[\left[T(B),U(a)\right],U(c)\right]. 
\end{multlined}
\end{equation}

The first term is simply
\begin{equation}
\left[\left[T(A),T(B)\right],T(C)\right]
=
T([[A,B],C]),
\end{equation}
using the identity $[T(A),T(B)]=T([A,B])$.

Using the canonical anticommutation relations, the number operator
$n_{k\sigma}\coloneqq c^\dagger_{k\sigma}c_{k\sigma}$ satisfies
\begin{equation}
[n_{k\sigma},c^\dagger_{j\sigma}c_{\ell\sigma}]
=
(\delta_{kj}-\delta_{k\ell})
c^\dagger_{j\sigma}c_{\ell\sigma}.
\label{eq:number_hopping_commutator}
\end{equation}
Equivalently, $[n_{k\sigma},c^\dagger_{k\sigma}c_{j\sigma}]= c^\dagger_{k\sigma}c_{j\sigma}$ and $[n_{k\sigma},c^\dagger_{j\sigma}c_{k\sigma}]=-c^\dagger_{j\sigma}c_{k\sigma}$.
Since $Z_{k\sigma}=1-2n_{k\sigma}$, \cref{eq:number_hopping_commutator} gives
\begin{subequations}
\begin{align}
[T(A),U(b)]
&=
2\sum_{jk,\sigma}
A_{jk}
\left(
b_jZ_{j\bar{\sigma}}
-
b_kZ_{k\bar{\sigma}}
\right)
c^\dagger_{j\sigma}c_{k\sigma}\\
&=
2\sum_{j,\sigma}
b_jZ_{j\bar{\sigma}}Q_{j\sigma}(A).
\end{align}
\end{subequations}
Hence
\begin{align}
[[T(A),T(B)],U(c)]
&=
2\sum_{j,\sigma}
c_jZ_{j\bar{\sigma}}Q_{j\sigma}([A,B]),\\
[[T(A),U(b)],T(C)]
&=
2\sum_{j,\sigma}b_j
\left(
Z_{j\bar{\sigma}}
\left[
Q_{j\sigma}(A),T_\sigma(C)
\right]
-
2Q_{j\sigma}(A)Q_{j\bar{\sigma}}(C)
\right).
\end{align}
Subtracting the corresponding expression with $(A,b)\leftrightarrow(B,a)$ gives
\begin{equation}
\begin{multlined}
[[T(A),U(b)],T(C)] - [[T(B),U(a)],T(C)]=
2\sum_{j,\sigma}
Z_{j\bar{\sigma}}
\left[
b_jQ_{j\sigma}(A)-a_jQ_{j\sigma}(B),
T_\sigma(C)
\right] \\
-4\sum_{j,\sigma}
\left(
b_jQ_{j\sigma}(A)-a_jQ_{j\sigma}(B)
\right)
Q_{j\bar{\sigma}}(C).
\end{multlined}
\end{equation}

Finally, letting $D_x$ denote the diagonal matrix with entries
$(D_x)_{jk}=x_j\delta_{jk}$, there exists a unitary $V$ such that
\begin{equation}
\begin{multlined}
[[T(A),U(b)],U(c)]
= 4T\!\left(\{A,D_bD_c\}\right) -
4V^\dagger
T\!\left(D_bAD_c+D_cAD_b\right)
V,
\end{multlined}
\end{equation}
and hence
\begin{align}
&[[T(A),U(b)],U(c)]
-
[[T(B),U(a)],U(c)]
\nonumber\\
&=
4T\!\left(
\{A,D_bD_c\}
-
\{B,D_aD_c\}
\right)
\nonumber\\
&\quad-
4V^\dagger
T\!\left(
D_bAD_c+D_cAD_b
-
D_aBD_c-D_cBD_a
\right)
V.
\end{align}

Collecting these expressions yields
\begin{equation}
\begin{multlined}
G(A,B,C,a,b,c)
=
\sum_\sigma
T_\sigma\!\left(
[[A,B],C]
+
4\{A,D_bD_c\}
-
4\{B,D_aD_c\}
\right)\\
-4\sum_\sigma
V^\dagger
T_\sigma\!\left(
D_bAD_c+D_cAD_b
-
D_aBD_c-D_cBD_a
\right)
V \\+
2\sum_{j,\sigma}
Z_{j\bar{\sigma}}
\left(
\left[
b_jQ_{j\sigma}(A)-a_jQ_{j\sigma}(B),
T_\sigma(C)
\right]
+
c_jQ_{j\sigma}([A,B])
\right)\\
-
4\sum_{j,\sigma}
\left(
b_jQ_{j\sigma}(A)-a_jQ_{j\sigma}(B)
\right)
Q_{j\bar{\sigma}}(C).
\end{multlined}
\end{equation}

Applying the triangle inequality, unitary invariance of the spectral norm, $\|Z_{j\sigma}\|=1$, and submultiplicity gives the claimed bound.
\end{proof}

\section{Data for quantum resource estimates}\label{appen:qre_data}

This section includes additional data on the reported quantum resource estimates for reproducibility.

Details of the baseline circuit recorded for comparison in the results are described in~\cref{tab:baseline_vs_improved_description}.
For both implementations, we first consider a single batch of rotations for Hamming-weight phasing; the impact of varying the number of batches is investigated later in this section.
\begin{table}[H]
\centering
\begin{tabularx}{\linewidth}{|X|X|X|}
\hline
\multicolumn{1}{|c|}{\textbf{Component}} &
\multicolumn{1}{|c|}{\textbf{Baseline}} &
\multicolumn{1}{|c|}{\textbf{This work}} \\
\hline
Phase estimation &
Entanglement-free QPE~\cite{higgins2007entanglement} &
Sine-windowed QPE~\cite{babbush2018encoding} \\
\hline
Trotter term ordering &
IPG ordering \cite{campbell2022early, kan2024resource-optimized} &
PIG ordering \\
\hline
Site Jordan-Wigner ordering &
Conventional site enumeration (no plaquette-locality optimization)\cite{kan2024resource-optimized}  &
Plaquette-locality-optimized enumeration requiring fSWAP networks only for gold plaquette terms \\
\hline
Trotter circuit &
Trotter merging within a given query; no Trotter term merging between queries \cite{campbell2022early, kan2024resource-optimized} &
Trotter term merging within a given query and between queries \\
\hline
fSWAP compilation &
Naive fSWAP compilation &
Active-volume-optimized fSWAP compilation \\
\hline
Gate decompositions &
Standard gate decompositions &
OZX-compressed circuit representation, e.g., using PPRs \\
\hline
\end{tabularx}
\caption{Comparison between the baseline circuit and the circuit presented in this work.}
\label{tab:baseline_vs_improved_description}
\end{table}

We verified the baseline circuit construction by reproducing the Toffoli count derived in Ref.~\cite{kan2024resource-optimized}.
We implemented the equations from Kan et al. assuming a single batch of rotations for Hamming weight phasing and multiplying the cost per query by the number of queries from the error budget optimization. The two sets of counts, as shown in~\cref{fig:baseline_toff_comparison}, agree perfectly across all lattice sizes.
We then constructed the improved circuit by replacing baseline subroutines with their active volume optimized counterparts, updating circuit parameters such as rotation synthesis errors, and implementing appropriate controlled query structures based on the choice of QPE.

\begin{figure}[H]
\centering
\includegraphics[width=0.5\textwidth]{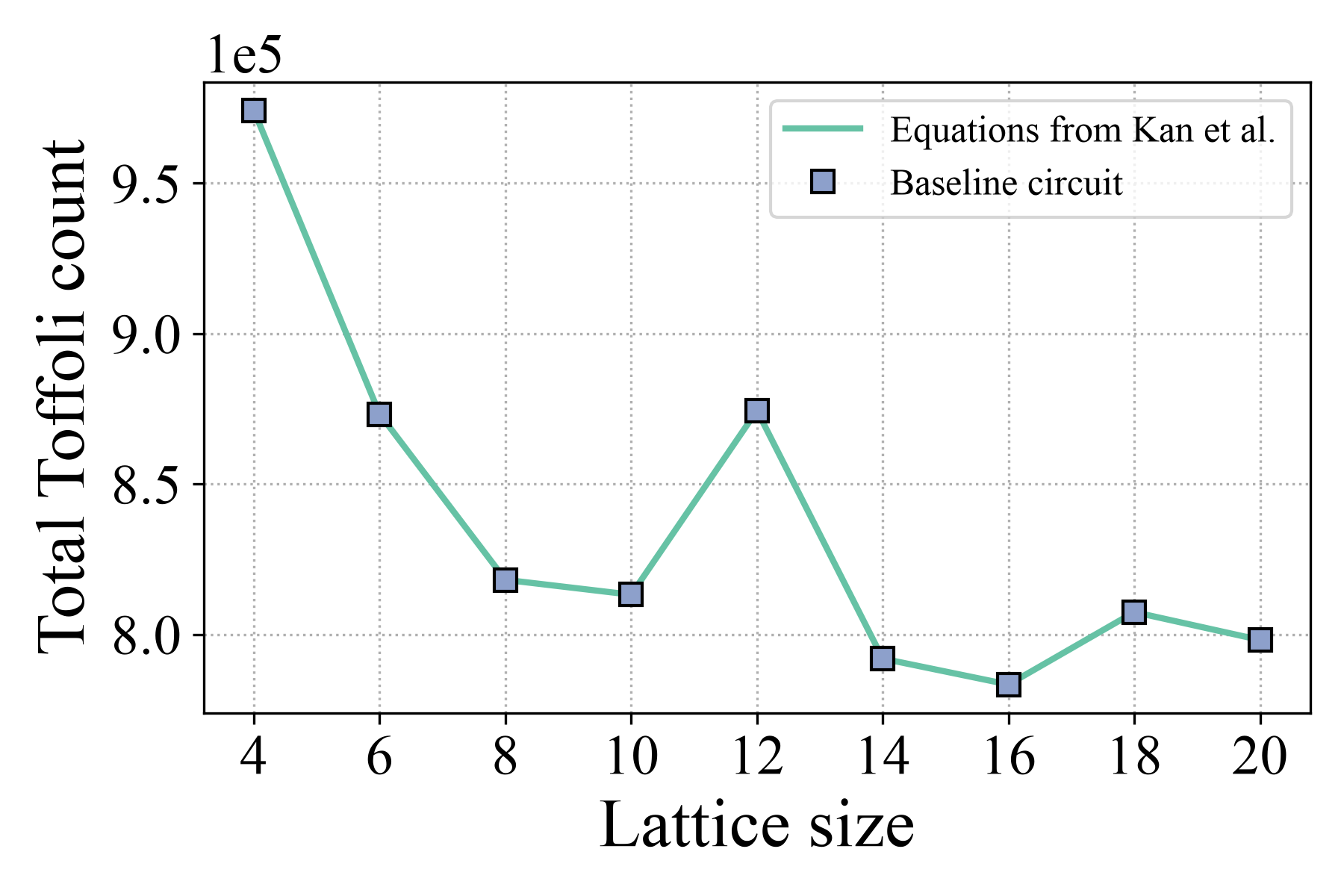}
\caption{Comparison of the expression for the total number of Toffoli gates in Ref.~\cite{kan2024resource-optimized}
    and the Toffoli gates counted from the corresponding baseline circuit constructed using Workbench \cite{psiquantum_qdk}. 
    The increase in cost at lattice size 12 is due to the error optimizer converging to a suboptimal solution. This is exacerbated by the discretized optimization landscape of entanglement-free QPE, since the number of QPE queries is restricted to multiples of 6 to achieve the Holevo variance constant factor. Compilation improvements in this work are built upon this circuit implementation.
}
\label{fig:baseline_toff_comparison}
\end{figure}

To produce our resource estimates we optimize the error parameters within the error budget constraint, the data for batch sizes $\beta = 1,2$ are reported in~\cref{tab:error_param_all_lattice_sizes}.

\begin{table}[H]
\centering
\begin{subtable}{\linewidth}
\centering
\caption{$\beta=1$ batch}
\setlength{\tabcolsep}{3.2pt}
\resizebox{\linewidth}{!}{%
\begin{tabular}{
    |c |
    c c c c |
    c c c c c |
    c c |c c|
}
\hline
$\boldsymbol{L}$
&
\multicolumn{4}{c|}{\textbf{Optimization parameters}}
&
\multicolumn{5}{c|}{\textbf{Energy error contributions}}
&
\multicolumn{2}{c|}{\textbf{Circuit params.}}
&
\multicolumn{2}{c|}{\textbf{Resources}}
\\

&
$\tau$
&
$\delta_{\mathrm{BD}}$
&
$\delta_{\mathrm{ST}}$
&
$\delta_{\mathrm{RS}}$
&
$\epsilon$
&
$\epsilon_{\mathrm{QPE}}$
&
$\epsilon_{\mathrm{TS}}$
&
$\epsilon_{\mathrm{RS}}$
&
$\epsilon_{\mathrm{cat}}$
&
$N$
&
$r$
&
$n_{\mathrm{Toff}} + n_T/2$
&
$n_{\mathrm{qubits}}$
\\

\hline

$4$
& 0.1210
& 0.9330
& 0.2990
& 0.1210
& 0.0816
& 0.0514
& 0.0282
& 0.0006
& 0.0014
& 256
& 12
& $5.94 \times 10^{5}$
& 73
\\
\hline

$6$
& 0.1141
& 0.5916
& 0.9469
& 0.5741
& 0.1836
& 0.1086
& 0.0710
& 0.0022
& 0.0017
& 128
& 11
& $5.01 \times 10^{5}$
& 134
\\
\hline

$8$
& 0.6912
& 0.9596
& 0.5603
& 0.1092
& 0.3264
& 0.2256
& 0.0967
& 0.0023
& 0.0018
& 64
& 12
& $4.55 \times 10^{5}$
& 221
\\
\hline

$10$
& 0.6527
& 0.9512
& 0.6294
& 0.0742
& 0.5100
& 0.3329
& 0.1685
& 0.0054
& 0.0032
& 64
& 8
& $4.47 \times 10^{5}$
& 327
\\
\hline

$12$
& 0.8003
& 0.9098
& 0.7574
& 0.0838
& 0.7344
& 0.5877
& 0.1279
& 0.0174
& 0.0032
& 32
& 12
& $4.76 \times 10^{5}$
& 462
\\
\hline

$14$
& 0.6643
& 0.9514
& 0.3511
& 0.0756
& 0.9996
& 0.6641
& 0.3192
& 0.0057
& 0.0106
& 32
& 8
& $4.25 \times 10^{5}$
& 617
\\
\hline

$16$
& 0.9700
& 0.5595
& 0.2165
& 0.0599
& 1.3056
& 0.8207
& 0.4638
& 0.0076
& 0.0135
& 32
& 6
& $4.16 \times 10^{5}$
& 802
\\
\hline

$18$
& 0.6118
& 0.9171
& 0.2464
& 0.0497
& 1.6524
& 1.0110
& 0.5882
& 0.0131
& 0.0400
& 32
& 5
& $4.34 \times 10^{5}$
& 1004
\\
\hline

$20$
& 0.6438
& 0.9078
& 0.1782
& 0.0376
& 2.0400
& 1.3134
& 0.6597
& 0.0119
& 0.0550
& 32
& 4
& $4.28 \times 10^{5}$
& 1232
\\

\hline
\end{tabular}}
\end{subtable}

\vspace{1em}

\begin{subtable}{\linewidth}
\centering
\caption{$\beta=2$ batches}
\setlength{\tabcolsep}{3.2pt}
\resizebox{\linewidth}{!}{%
\begin{tabular}{
    |c |
    c c c c |
    c c c c c |
    c c |c c|
}
\hline
$\boldsymbol{L}$
&
\multicolumn{4}{c|}{\textbf{Optimization parameters}}
&
\multicolumn{5}{c|}{\textbf{Energy error contributions}}
&
\multicolumn{2}{c|}{\textbf{Circuit params.}}
&
\multicolumn{2}{c|}{\textbf{Resources}}
\\

&
$\tau$
&
$\delta_{\mathrm{BD}}$
&
$\delta_{\mathrm{ST}}$
&
$\delta_{\mathrm{RS}}$
&
$\epsilon$
&
$\epsilon_{\mathrm{QPE}}$
&
$\epsilon_{\mathrm{TS}}$
&
$\epsilon_{\mathrm{RS}}$
&
$\epsilon_{\mathrm{cat}}$
&
$N$
&
$r$
&
$n_{\mathrm{Toff}} + n_T/2$
&
$n_{\mathrm{qubits}}$
\\

\hline

$4$
& 0.1132
& 0.6699
& 0.9447
& 0.9193
& 0.0816
& 0.0547
& 0.0255
& 0.0013
& 0.0001
& 256
& 12
& $7.38 \times 10^{5}$
& 62
\\
\hline

$6$
& 0.1052
& 0.6462
& 0.9398
& 0.4199
& 0.1836
& 0.1186
& 0.0610
& 0.0016
& 0.0023
& 128
& 11
& $5.75 \times 10^{5}$
& 113
\\
\hline

$8$
& 0.1100
& 0.6842
& 0.9636
& 0.3837
& 0.3264
& 0.2233
& 0.0993
& 0.0014
& 0.0023
& 64
& 12
& $5.05 \times 10^{5}$
& 186
\\
\hline

$10$
& 0.0776
& 0.6206
& 0.9232
& 0.7391
& 0.5100
& 0.3165
& 0.1787
& 0.0110
& 0.0039
& 64
& 8
& $4.74 \times 10^{5}$
& 274
\\
\hline

$12$
& 0.0838
& 0.7999
& 0.8733
& 0.8667
& 0.7344
& 0.5874
& 0.1283
& 0.0161
& 0.0025
& 32
& 12
& $5.01 \times 10^{5}$
& 387
\\
\hline

$14$
& 0.0756
& 0.6592
& 0.9379
& 0.8524
& 0.9996
& 0.6589
& 0.3195
& 0.0180
& 0.0031
& 32
& 8
& $4.40 \times 10^{5}$
& 516
\\
\hline

$16$
& 0.0580
& 0.6552
& 0.9792
& 0.5135
& 1.3056
& 0.8555
& 0.4408
& 0.0048
& 0.0045
& 32
& 6
& $4.31 \times 10^{5}$
& 671
\\
\hline

$18$
& 0.0195
& 0.6378
& 0.8937
& 0.6940
& 1.6524
& 1.0540
& 0.5348
& 0.0441
& 0.0195
& 32
& 5
& $4.45 \times 10^{5}$
& 839
\\
\hline

$20$
& 0.0382
& 0.6370
& 0.9224
& 0.5471
& 2.0400
& 1.2994
& 0.6831
& 0.0314
& 0.0260
& 32
& 4
& $4.36 \times 10^{5}$
& 1029
\\
\hline
\end{tabular}}
\end{subtable}

\caption{
Optimized error parameters, energy-error contributions, QPE queries $N$, Trotter-step counts $r$, the non-Clifford gate count and logical qubit count, $n_\text{qubits}$, for the two batching strategies considered for the 2D square $L\times L$ Fermi-Hubbard ground-state energy algorithm. 
}
\label{tab:error_param_all_lattice_sizes}
\end{table}

\section{Hardware model for runtime estimates}
\label{sec:runtime_data}

The active volume architecture is well-suited for photonic quantum computing. It makes use of a logarithmic number of non-local swap operations per qubit to substantially reduce idle times. In a photonic quantum computer, these non-local operations can be facilitated by optical switching and fiber-delay lines, although they may also be implemented in other architectures. For this work, we adopt the photonic hardware parameters summarized in~\cref{tab:hardware_parameters}, following Refs.~\cite{bombin2021interleaving,litinski2022activevolume,caesura2025faster,heavey2025improved}.

\begin{table}[H]
    \renewcommand{\arraystretch}{1.2}
    \centering
    \begin{tabular}{|c|c|}
        \hline
        \textbf{Hardware parameter} & \textbf{Value} \\
        \hline
        Delay length & $2.5\,\mathrm{km}$ \\
        \hline
        Speed of light in fiber & $\frac{2}{3}c$ \\
        \hline
        Target hardware-induced failure probability ($p_f$)
            & $10\%$ \\
        \hline
        Logical-error suppression parameter ($\alpha$) & $0.5$ \\
        \hline
        Reaction time ($t_{\mathrm{r}}$) & $1.5\times10^{-5}\,\mathrm{s}$ \\
        \hline
        Code-cycle time ($t_{\mathrm{cc}}$) & $10^{-5}\,\mathrm{s}$ \\
        \hline
    \end{tabular}
    \caption{
        Hardware parameters used for the illustrative runtime estimates.
        These parameters were selected to represent a feasible early
        fault-tolerant photonic quantum computer based on
        Ref.~\cite{litinski2022activevolume}.
    }
    \label{tab:hardware_parameters}
\end{table}

Our resource estimates use the block scheduler~\cite{heavey2026scheduler} to assign each lattice-surgery operation to a logical cycle, rather than estimating runtime from operation counts alone.
The resulting schedule requires $\ell$ logical cycles. Given a machine containing $n_t$ logical qubits, including both memory and workspace, we model the logical failure probability per qubit per logical cycle as
\begin{equation}
    p_{\mathrm{L}}(d)=10^{-\alpha d/2}.
\end{equation}
Assuming independent fault locations, we choose the smallest permitted code distance $d$ satisfying
\begin{equation}
    1-
    \left(1-p_{\mathrm{L}}(d)\right)^{n_t\ell}
    \leq p_f,
\end{equation}
where $p_f$ is the allocated probability of hardware-induced failure over the complete computation.

A logical cycle comprises $d$ code cycles, so its duration and the resulting runtime are
\begin{equation}
    t_{\mathrm{lc}}=d\,t_{\mathrm{cc}},
    \qquad
    t_{\mathrm{run}}=\ell d\,t_{\mathrm{cc}}.
\end{equation}
The reaction time $t_{\mathrm{r}}$ is distinct from the code-cycle time $t_{\mathrm{cc}}$.
It characterizes the latency of measurement-dependent classical processing and feed-forward, and therefore constrains the number of sequential reaction layers that can be completed within a logical cycle.

To further reduce runtime, we used an improved version of the scheduler introduced in \cite{heavey2025improved}, which will be described in more detail in an upcoming manuscript \cite{UpcomingBlockScheduler}.
Unlike the original approach, which assembles the schedule from separate circuit segments, this scheduler is fast enough to process entire circuits.
This is achieved by limiting the size of the directed acyclic graph (DAG) used to store the execution order of operations.
While previous work has applied simpler DAG-based approaches to this problem \cite{ruh2025quantum, silva_et_al:LIPIcs.TQC.2024.1}, we use a sliding-window DAG \cite{herzog2025movable}, yielding a substantial reduction in compilation time compared with the implementation described in \cite{heavey2026scheduler}.
Once the DAG exceeds a specified size ($500$ operations in this case) the oldest operations are removed and scheduled.

\end{document}